\documentclass{article}
\usepackage{graphicx} % Required for inserting images
\usepackage{todonotes}
\usepackage{geometry}

\usepackage{amsmath,amsfonts,amssymb,amsthm,physics,mathtools,times}
 \mathtoolsset{showonlyrefs=true}
\usepackage{ifthen}
\usepackage{tikz-3dplot}
\usepackage{authblk}

\usetikzlibrary{shapes}
\usetikzlibrary{shapes.geometric}
\usetikzlibrary{backgrounds}
\usetikzlibrary{plotmarks}
\usetikzlibrary{patterns}
\usetikzlibrary{arrows.meta}

\usepackage{hyperref,placeins}
\usepackage[most]{tcolorbox} 
\newtcolorbox{todobox}{ colback=orange!10, colframe=orange!70!black, title=Strategy (draft), fonttitle=\bfseries, breakable }
\usepackage{booktabs}

\newcommand{\dist}[1]{\operatorname{dist}(#1)}
\newcommand{\checkset}{\mathcal{S}}
\newcommand{\dir}{a}

\newcommand{\sumtauP}[1]{\Sigma_{#1}}

\newcommand{\sumtauPup}[2]{\Sigma_{#1}^{#2}}

\usepackage[
  backend=biber,
  style=numeric-comp,
  sorting=none,
  url=false,
  maxbibnames=999
]{biblatex}
\definecolor{teal}{HTML}{00C2B0}
\definecolor{violet}{HTML}{9F00FF}
\definecolor{yellow}{HTML}{F0CE00}

\makeatletter

\newlength{\uvw@unit}
\newlength{\uvw@axis}
\def\uvw@setaxis{\settoheight{\uvw@axis}{\ensuremath{\vcenter{\hbox{}}}}}

\tikzset{
  uvw box/.style={line width=0.9pt},
  uvw label/.style={font=\normalsize},
  uvw dim/.style={stealth-stealth, line width=0.5pt},
  uvw dim label/.style={font=\footnotesize, inner sep=2pt},
  uvw guide/.style={line width=0.3pt, gray!70},
  uvw span/.style={line width=0.6pt},
  uvw span label/.style={font=\small, inner sep=2pt},
  uvw axes/.style={-stealth, line width=0.6pt},
  uvw axis label/.style={font=\footnotesize, inner sep=1.5pt},
  uvw note/.style={font=\footnotesize, align=center, inner sep=2pt},
}

\def\uvw@ifempty#1{\if\relax\detokenize\expandafter{#1}\relax
  \expandafter\@firstoftwo\else\expandafter\@secondoftwo\fi}
\def\uvw@ifaxes#1{\uvw@ifempty#1\@secondoftwo
  {\ifthenelse{\equal{#1}{none}}\@secondoftwo\@firstoftwo}}
\newsavebox{\uvw@nbox}
\def\uvw@noteheight#1#2{%
  \uvw@ifempty#1{\def#2{0}}{%
    \sbox\uvw@nbox{\footnotesize\begin{tabular}{@{}c@{}}#1\end{tabular}}%
    \dimen@=\ht\uvw@nbox \advance\dimen@\dp\uvw@nbox \advance\dimen@ 4pt
    \pgfmathsetmacro#2{\dimen@/\uvw@unit}}}

\def\uvw@defaults{\pgfkeys{/uvw/.cd,
  unit=9mm,
  u width=3, v width=1.6, w width=3, height=3,
  u label=U, v label=V, w label=W,
  u length={\Theta(L)}, v length={L^{2/3}}, w length={\Theta(L)},
  total length={\Theta(L)}, height label={},
  v note={}, note side=above,
  u color=teal,   u fill=teal!30,
  v color=yellow, v fill=yellow!30,
  w color=violet, w fill=violet!30,
  axes=none, tikz={}}}
\pgfkeys{/uvw/.cd,
  v note/.store in=\uvw@nv,      note side/.store in=\uvw@ns,
  axes length/.code={\def\uvwaxislength{#1}},
  axes gap/.code={\def\uvwaxisgap{#1}},
  unit/.code={\setlength{\uvw@unit}{#1}},
  u width/.store in=\uvw@wu,     v width/.store in=\uvw@wv,
  w width/.store in=\uvw@ww,     height/.store in=\uvw@h,
  u label/.store in=\uvw@lu,     v label/.store in=\uvw@lv,
  w label/.store in=\uvw@lw,
  u length/.store in=\uvw@du,    v length/.store in=\uvw@dv,
  w length/.store in=\uvw@dw,    total length/.store in=\uvw@dt,
  height label/.store in=\uvw@dh,
  u color/.store in=\uvw@cu,     u fill/.store in=\uvw@fu,
  v color/.store in=\uvw@cv,     v fill/.store in=\uvw@fv,
  w color/.store in=\uvw@cw,     w fill/.store in=\uvw@fw,
  axes/.store in=\uvw@axspec,    tikz/.store in=\uvw@tikz,
  measures/.code={\ifthenelse{\equal{#1}{false}}{%
      \def\uvw@du{}\def\uvw@dv{}\def\uvw@dw{}\def\uvw@dt{}\def\uvw@dh{}}{}},
}

\newcommand{\uvwaxislength}{5mm}
\newcommand{\uvwaxisgap}{2mm}
\def\uvw@axsplit#1/#2\@nil{\def\uvw@axh{#1}\def\uvw@axv{#2}}
\def\uvw@axinfo#1#2\@nil{%
  \if-#1\def\uvw@axd{-1}\def\uvw@axl{#2}\else\def\uvw@axd{1}\def\uvw@axl{#1#2}\fi}
\def\uvw@drawaxes#1#2#3#4{% #1 = <h>/<v>, #2/#3 = left/right x, #4 = bottom y
  \def\uvw@axtmp{#1}\expandafter\uvw@axsplit\uvw@axtmp\@nil
  \expandafter\uvw@axinfo\uvw@axh\@nil \let\uvw@axhd\uvw@axd \let\uvw@axhl\uvw@axl
  \expandafter\uvw@axinfo\uvw@axv\@nil \let\uvw@axvd\uvw@axd \let\uvw@axvl\uvw@axl
  \ifnum\uvw@axhd>0 \def\uvw@axha{west}\def\uvw@axbx{#2}\def\uvw@axox{-\uvwaxisgap}%
  \else             \def\uvw@axha{east}\def\uvw@axbx{#3}\def\uvw@axox{\uvwaxisgap}\fi
  \ifnum\uvw@axvd>0 \def\uvw@axva{south}\else\def\uvw@axva{north}\fi
  \begin{scope}[shift={({\uvw@axbx},{#4})}, xshift={\uvw@axox}, yshift={-\uvwaxisgap}]
    \draw[uvw axes] (0pt,0pt) -- ({\uvw@axhd*\uvwaxislength},0pt)
      node[uvw axis label, anchor=\uvw@axha] {$\uvw@axhl$};
    \draw[uvw axes] (0pt,0pt) -- (0pt,{\uvw@axvd*\uvwaxislength})
      node[uvw axis label, anchor=\uvw@axva] {$\uvw@axvl$};
  \end{scope}}

\def\uvw@dimh#1#2#3#4#5{% x1, x2, y of the line, y of the box edge, text
  \draw[uvw guide] ({#1},{#4}) -- ({#1},{#3});
  \draw[uvw guide] ({#2},{#4}) -- ({#2},{#3});
  \draw[uvw dim] ({#1},{#3}) -- ({#2},{#3})
    node[uvw dim label, midway, \uvw@dimside] {$#5$};}

\def\uvw@vdots#1#2{% centre x, centre y
  \foreach \uvw@dd in {-1,0,1}
    {\fill ({#1},{#2}) ++(0pt,{\uvw@dd*2.8pt}) circle[radius=0.55pt];}}

\newcommand{\uvw}[1][]{%
  \uvw@setaxis
  \begingroup
  \uvw@defaults \pgfkeys{/uvw/.cd,#1}%
  \pgfmathsetmacro{\uvw@xv}{\uvw@wu}%              left edge of V
  \pgfmathsetmacro{\uvw@xw}{\uvw@wu+\uvw@wv}%      left edge of W
  \pgfmathsetmacro{\uvw@tot}{\uvw@wu+\uvw@wv+\uvw@ww}%
  \pgfmathsetmacro{\uvw@yb}{-\uvw@h-0.85}%         measures below
  \uvw@noteheight\uvw@nv\uvw@nh
  \pgfmathsetmacro{\uvw@yt}{max(0.85,0.4+\uvw@nh)}% measure above
  \ifthenelse{\equal{\uvw@ns}{below}}{%
    \def\uvw@na{north}\def\uvw@nb{0}%
    \uvw@ifempty\uvw@du{}{\def\uvw@nb{1}}%
    \uvw@ifempty\uvw@dv{}{\def\uvw@nb{1}}%
    \uvw@ifempty\uvw@dw{}{\def\uvw@nb{1}}%
    \ifnum\uvw@nb=1 \pgfmathsetmacro{\uvw@ny}{\uvw@yb-14pt/\uvw@unit}%
    \else           \pgfmathsetmacro{\uvw@ny}{-\uvw@h-0.1}\fi
  }{\def\uvw@na{south}\def\uvw@ny{0.1}}%
  \begin{tikzpicture}[x=\uvw@unit, y=\uvw@unit,
      baseline={([yshift=-\uvw@axis]uvw@c)}]
    \expandafter\tikzset\expandafter{\uvw@tikz}%
    \draw[uvw box, draw=\uvw@cv, fill=\uvw@fv]
      ({\uvw@xv},0) rectangle ({\uvw@xw},{-\uvw@h});
    \draw[uvw box, draw=\uvw@cu, fill=\uvw@fu]
      (0,0) rectangle ({\uvw@wu},{-\uvw@h});
    \draw[uvw box, draw=\uvw@cw, fill=\uvw@fw]
      ({\uvw@xw},0) rectangle ({\uvw@tot},{-\uvw@h});
    \coordinate (uvw@c) at ({\uvw@tot/2},{-\uvw@h/2});
    \if\relax\detokenize\expandafter{\uvw@lu}\relax\else
      \node[uvw label, text=black] at ({\uvw@wu/2},{-\uvw@h/2}) {$\uvw@lu$};\fi
    \if\relax\detokenize\expandafter{\uvw@lv}\relax\else
      \node[uvw label, text=black] at ({(\uvw@xv+\uvw@xw)/2},{-\uvw@h/2}) {$\uvw@lv$};\fi
    \if\relax\detokenize\expandafter{\uvw@lw}\relax\else
      \node[uvw label, text=black] at ({(\uvw@xw+\uvw@tot)/2},{-\uvw@h/2}) {$\uvw@lw$};\fi
    \def\uvw@dimside{below}%
    \if\relax\detokenize\expandafter{\uvw@du}\relax\else
      \uvw@dimh{0}{\uvw@wu}{\uvw@yb}{-\uvw@h}{\uvw@du}\fi
    \if\relax\detokenize\expandafter{\uvw@dv}\relax\else
      \uvw@dimh{\uvw@xv}{\uvw@xw}{\uvw@yb}{-\uvw@h}{\uvw@dv}\fi
    \if\relax\detokenize\expandafter{\uvw@dw}\relax\else
      \uvw@dimh{\uvw@xw}{\uvw@tot}{\uvw@yb}{-\uvw@h}{\uvw@dw}\fi
    \def\uvw@dimside{above}%
    \if\relax\detokenize\expandafter{\uvw@dt}\relax\else
      \uvw@dimh{0}{\uvw@tot}{\uvw@yt}{0}{\uvw@dt}\fi
    \if\relax\detokenize\expandafter{\uvw@dh}\relax\else
      \draw[uvw guide] (0,0) -- (-0.95,0);
      \draw[uvw guide] (0,{-\uvw@h}) -- (-0.95,{-\uvw@h});
      \draw[uvw dim] (-0.85,0) -- (-0.85,{-\uvw@h})
        node[uvw dim label, midway, left] {$\uvw@dh$};\fi
    \uvw@ifempty\uvw@nv{}{%
      \node[uvw note, fill=white, anchor=\uvw@na]
        at ({(\uvw@xv+\uvw@xw)/2},{\uvw@ny}) {\uvw@nv};}%
    \if\relax\detokenize\expandafter{\uvw@axspec}\relax\else
      \ifthenelse{\equal{\uvw@axspec}{none}}{}{%
        \expandafter\uvw@drawaxes\expandafter{\uvw@axspec}%
          {0}{\uvw@tot}{-\uvw@h}}%
    \fi
  \end{tikzpicture}%
  \endgroup}

\def\uvws@defaults{\pgfkeys{/uvws/.cd,
  unit=8mm, total=11, margin=2.3, v width=0.8, height=1.6,
  v labels={1,2,dots,i,dots,{L^{1/6}}},
  u label=U, v label=V, w label=W,
  total length={\Theta(L)}, margin label={\Omega(L/3)}, v length={L^{2/3}},
  u color=teal,   u fill=teal!30,
  v color=yellow, v fill=yellow!30,
  w color=violet, w fill=violet!30,
  stacked=false, stack labels={1,2,dots,i,dots,I},
  row height=1.6, row gap=0.5, dots gap=1.1, touching=true, dots weight=0,
  axes=none, tikz={}}}
\def\uvws@stackdefaults{\pgfkeys{/uvws/.cd,
  unit=4mm, total=17.5, margin=3.6, v width=1.35, total length={}}}
\newif\ifuvws@stacked
\newif\ifuvws@touch
\pgfkeys{/uvws/.cd,
  touching/.is if=uvws@touch,
  axes length/.code={\def\uvwaxislength{#1}},
  axes gap/.code={\def\uvwaxisgap{#1}},
  stacked/.is if=uvws@stacked,
  stack labels/.store in=\uvws@sl,
  row height/.store in=\uvws@rh,  row gap/.store in=\uvws@rg,
  dots gap/.store in=\uvws@dg,  dots weight/.store in=\uvws@dw,
  measures/.code={\ifthenelse{\equal{#1}{false}}{%
      \def\uvws@dt{}\def\uvws@dm{}\def\uvws@dv{}}{}},
  unit/.code={\setlength{\uvw@unit}{#1}},
  total/.store in=\uvws@tot,     margin/.store in=\uvws@mar,
  v width/.store in=\uvws@wv,    height/.store in=\uvws@h,
  v labels/.store in=\uvws@vl,
  u label/.store in=\uvws@lu,    v label/.store in=\uvws@lv,
  w label/.store in=\uvws@lw,
  total length/.store in=\uvws@dt, margin label/.store in=\uvws@dm,
  v length/.store in=\uvws@dv,
  u color/.store in=\uvws@cu,    u fill/.store in=\uvws@fu,
  v color/.store in=\uvws@cv,    v fill/.store in=\uvws@fv,
  w color/.store in=\uvws@cw,    w fill/.store in=\uvws@fw,
  axes/.store in=\uvws@axspec,   tikz/.store in=\uvws@tikz,
}

\def\uvws@row#1#2#3{% #1 = offset of V from U_1, #2 = entry, #3 = top of the row
  \pgfmathsetmacro{\uvwsxv}{\uvws@mar+#1}%
  \pgfmathsetmacro{\uvwsxw}{\uvwsxv+\uvws@wv}%
  \pgfmathsetmacro{\uvwsyb}{#3-\uvws@rh}%
  \pgfmathsetmacro{\uvwsym}{#3-\uvws@rh/2}%
  \draw[uvw box, draw=\uvws@cv, fill=\uvws@fv] ({\uvwsxv},{#3}) rectangle ({\uvwsxw},{\uvwsyb});
  \draw[uvw box, draw=\uvws@cu, fill=\uvws@fu] (0,{#3}) rectangle ({\uvwsxv},{\uvwsyb});
  \draw[uvw box, draw=\uvws@cw, fill=\uvws@fw] ({\uvwsxw},{#3}) rectangle ({\uvws@tot},{\uvwsyb});
  \if\relax\detokenize\expandafter{\uvws@lu}\relax\else
    \node[uvw label, text=black] at ({\uvwsxv/2},{\uvwsym}) {$\uvws@lu_{#2}$};\fi
  \if\relax\detokenize\expandafter{\uvws@lv}\relax\else
    \node[uvw label, text=black, font=\footnotesize]
      at ({(\uvwsxv+\uvwsxw)/2},{\uvwsym}) {$\uvws@lv_{#2}$};\fi
  \if\relax\detokenize\expandafter{\uvws@lw}\relax\else
    \node[uvw label, text=black] at ({(\uvwsxw+\uvws@tot)/2},{\uvwsym}) {$\uvws@lw_{#2}$};\fi}
\def\uvws@stackpic{%
  \xdef\uvws@n{0}\xdef\uvws@m{0}%
  \foreach \uvwse in \uvws@sl {%
    \ifthenelse{\equal{\uvwse}{dots}}%
      {\pgfmathtruncatemacro{\uvwsc}{\uvws@m+1}\xdef\uvws@m{\uvwsc}}%
      {\pgfmathtruncatemacro{\uvwsc}{\uvws@n+1}\xdef\uvws@n{\uvwsc}}}%
  \ifuvws@touch \ifnum\uvws@m>0
      \def\uvws@ra{0}\def\uvws@rb{1}%
    \else \def\uvws@ra{1}\let\uvws@rb\uvws@dw\fi
  \else \def\uvws@ra{1}\let\uvws@rb\uvws@dw\fi
  \pgfmathsetmacro{\uvws@G}{(\uvws@n-1)*\uvws@ra+\uvws@m*\uvws@rb}%
  \pgfmathsetmacro{\uvws@gap}{\uvws@G>0 ?
      (\uvws@tot-2*\uvws@mar-\uvws@n*\uvws@wv)/\uvws@G : 0}%
  \ifdim\uvws@gap pt<0pt
    \PackageWarning{uvw}{\string\uvwshift[stacked]: the windows do not fit
      between U_1 and W_I and overlap; make total larger or margin or
      v width smaller}\fi
  \xdef\uvwsy{0}\xdef\uvwsfirst{1}%
  \foreach \uvwse in \uvws@sl {%
    \ifthenelse{\equal{\uvwse}{dots}}{%
      \pgfmathsetmacro{\uvwsyn}{\uvwsy-\uvws@dg}\xdef\uvwsy{\uvwsyn}\xdef\uvwsfirst{1}%
    }{%
      \ifnum\uvwsfirst=1 \xdef\uvwsfirst{0}\else
        \pgfmathsetmacro{\uvwsyn}{\uvwsy-\uvws@rg}\xdef\uvwsy{\uvwsyn}\fi
      \pgfmathsetmacro{\uvwsyn}{\uvwsy-\uvws@rh}\xdef\uvwsy{\uvwsyn}%
    }}%
  \let\uvws@ylast\uvwsy
  \begin{tikzpicture}[x=\uvw@unit, y=\uvw@unit,
      baseline={([yshift=-\uvw@axis]uvws@c)}]
    \expandafter\tikzset\expandafter{\uvws@tikz}%
    \xdef\uvwsy{0}\xdef\uvwsfirst{1}\xdef\uvwsj{0}\xdef\uvwsg{0}%
    \foreach \uvwse in \uvws@sl {%
      \ifthenelse{\equal{\uvwse}{dots}}{%
        \uvw@vdots{\uvws@tot/2}{\uvwsy-\uvws@dg/2}%
        \pgfmathsetmacro{\uvwsyn}{\uvwsy-\uvws@dg}\xdef\uvwsy{\uvwsyn}\xdef\uvwsfirst{1}%
        \pgfmathsetmacro{\uvwsgn}{\uvwsg+\uvws@rb}\xdef\uvwsg{\uvwsgn}%
      }{%
        \ifnum\uvwsfirst=1 \xdef\uvwsfirst{0}\else
          \pgfmathsetmacro{\uvwsyn}{\uvwsy-\uvws@rg}\xdef\uvwsy{\uvwsyn}\fi
        \pgfmathsetmacro{\uvwsoff}{\uvwsj*\uvws@wv+\uvwsg*\uvws@gap}%
        \uvws@row{\uvwsoff}{\uvwse}{\uvwsy}%
        \pgfmathsetmacro{\uvwsyn}{\uvwsy-\uvws@rh}\xdef\uvwsy{\uvwsyn}%
        \pgfmathtruncatemacro{\uvwsjn}{\uvwsj+1}\xdef\uvwsj{\uvwsjn}%
        \pgfmathsetmacro{\uvwsgn}{\uvwsg+\uvws@ra}\xdef\uvwsg{\uvwsgn}%
      }}%
    \coordinate (uvws@c) at ({\uvws@tot/2},{\uvws@ylast/2});
    \def\uvw@dimside{above}%
    \if\relax\detokenize\expandafter{\uvws@dm}\relax\else
      \uvw@dimh{0}{\uvws@mar}{0.5}{0}{\uvws@dm}\fi
    \if\relax\detokenize\expandafter{\uvws@dv}\relax\else
      \uvw@dimh{\uvws@mar}{\uvws@mar+\uvws@wv}{0.5}{0}{\uvws@dv}\fi
    \def\uvw@dimside{below}%
    \pgfmathsetmacro{\uvws@yb}{\uvws@ylast-0.85}%
    \uvw@ifempty\uvws@dm{%
      \let\uvws@ytb\uvws@yb
      \uvw@ifaxes\uvws@axspec{% clear of the coordinate arrows
        \pgfmathsetmacro{\uvws@ytb}{min(\uvws@ytb,
          \uvws@ylast-(\uvwaxisgap+2.8mm)/\uvw@unit)}}{}%
    }{%
      \uvw@dimh{\uvws@tot-\uvws@mar}{\uvws@tot}{\uvws@yb}{\uvws@ylast}{\uvws@dm}%
      \pgfmathsetmacro{\uvws@ytb}{\uvws@yb-17pt/\uvw@unit}%
    }%
    \uvw@ifempty\uvws@dt{}{%
      \uvw@dimh{0}{\uvws@tot}{\uvws@ytb}{\uvws@ylast}{\uvws@dt}}%
    \if\relax\detokenize\expandafter{\uvws@axspec}\relax\else
      \ifthenelse{\equal{\uvws@axspec}{none}}{}{%
        \expandafter\uvw@drawaxes\expandafter{\uvws@axspec}%
          {0}{\uvws@tot}{\uvws@ylast}}%
    \fi
  \end{tikzpicture}}

\newcommand{\uvwshift}[1][]{%
  \uvw@setaxis
  \begingroup
  \uvws@defaults \pgfkeys{/uvws/.cd,#1}%
  \ifuvws@stacked
    \uvws@stackdefaults \pgfkeys{/uvws/.cd,#1}\uvws@stackpic
  \else
  \xdef\uvws@k{1}%
  \foreach \uvwse [count=\uvwsc] in \uvws@vl {\xdef\uvws@k{\uvwsc}}%
  \pgfmathsetmacro{\uvws@step}{\uvws@k>1 ?
      (\uvws@tot-2*\uvws@mar-\uvws@wv)/(\uvws@k-1) : 0}%
  \ifdim\uvws@step pt<\uvws@wv pt \ifnum\uvws@k>1
    \PackageWarning{uvw}{\string\uvwshift: the windows overlap; make total
      larger or margin or v width smaller}\fi\fi
  \pgfmathsetmacro{\uvws@xa}{\uvws@mar}%
  \pgfmathsetmacro{\uvws@yl}{-\uvws@h-0.3}%   the V labels
  \pgfmathsetmacro{\uvws@yb}{-\uvws@h-1.15}%  L/3 measures
  \pgfmathsetmacro{\uvws@yv}{0.5}%            width of V_1
  \if\relax\detokenize\expandafter{\uvws@dv}\relax
    \pgfmathsetmacro{\uvws@yt}{0.5}%          total length (no V measure)
  \else
    \pgfmathsetmacro{\uvws@yt}{1.3}%          total length, above the V one
  \fi
  \begin{tikzpicture}[x=\uvw@unit, y=\uvw@unit,
      baseline={([yshift=-\uvw@axis]uvws@c)}]
    \expandafter\tikzset\expandafter{\uvws@tikz}%
    \coordinate (uvws@c) at ({\uvws@tot/2},{-\uvws@h/2});
    \fill[\uvws@fu] (0,0) rectangle ({\uvws@mar},{-\uvws@h});
    \fill[\uvws@fw] ({\uvws@tot-\uvws@mar},0) rectangle ({\uvws@tot},{-\uvws@h});
    \draw[uvw guide, densely dashed] ({\uvws@mar},0) -- ({\uvws@mar},{-\uvws@h});
    \draw[uvw guide, densely dashed]
      ({\uvws@tot-\uvws@mar},0) -- ({\uvws@tot-\uvws@mar},{-\uvws@h});
    \foreach \uvwse [count=\uvwsi] in \uvws@vl {%
      \pgfmathsetmacro{\uvwsxc}{\uvws@mar+(\uvwsi-1)*\uvws@step+0.5*\uvws@wv}%
      \ifthenelse{\equal{\uvwse}{dots}}{%
        \node[uvw label] at ({\uvwsxc},{-\uvws@h/2}) {$\cdots$};%
      }{%
        \draw[uvw box, draw=\uvws@cv, fill=\uvws@fv]
          ({\uvwsxc-\uvws@wv/2},0) rectangle ({\uvwsxc+\uvws@wv/2},{-\uvws@h});
        \if\relax\detokenize\expandafter{\uvws@lv}\relax\else
          \node[uvw dim label, text=black, anchor=north]
            at ({\uvwsxc},{\uvws@yl}) {$\uvws@lv_{\uvwse}$};\fi
      }}%
    \draw[uvw box] (0,0) rectangle ({\uvws@tot},{-\uvws@h});
    \if\relax\detokenize\expandafter{\uvws@lu}\relax\else
      \node[uvw label, text=black] at ({\uvws@mar/2},{-\uvws@h/2})
        {$\uvws@lu$};\fi
    \if\relax\detokenize\expandafter{\uvws@lw}\relax\else
      \node[uvw label, text=black]
        at ({\uvws@tot-\uvws@mar/2},{-\uvws@h/2}) {$\uvws@lw$};\fi
    \def\uvw@dimside{above}%
    \if\relax\detokenize\expandafter{\uvws@dv}\relax\else
      \uvw@dimh{\uvws@xa}{\uvws@xa+\uvws@wv}{\uvws@yv}{0}{\uvws@dv}\fi
    \if\relax\detokenize\expandafter{\uvws@dt}\relax\else
      \uvw@dimh{0}{\uvws@tot}{\uvws@yt}{0}{\uvws@dt}\fi
    \def\uvw@dimside{below}%
    \if\relax\detokenize\expandafter{\uvws@dm}\relax\else
      \uvw@dimh{0}{\uvws@mar}{\uvws@yb}{-\uvws@h}{\uvws@dm}%
      \uvw@dimh{\uvws@tot-\uvws@mar}{\uvws@tot}{\uvws@yb}{-\uvws@h}{\uvws@dm}\fi
    \if\relax\detokenize\expandafter{\uvws@axspec}\relax\else
      \ifthenelse{\equal{\uvws@axspec}{none}}{}{%
        \expandafter\uvw@drawaxes\expandafter{\uvws@axspec}%
          {0}{\uvws@tot}{-\uvws@h}}%
    \fi
  \end{tikzpicture}%
  \fi
  \endgroup}

\newif\ifuvwp@wrap
\def\uvwp@defaults{\pgfkeys{/uvwp/.cd,
  unit=9mm, u width=3, v width=1.2, w width=3, height=2.2,
  u label=U, v label=V, w label=W,
  u length={\Theta(L)}, v length={L^{2/3}}, w length={\Theta(L)},
  total length={\Theta(L)}, wrap label={},
  v note={}, note side=above,
  u color=teal,   u fill=teal!30,
  v color=yellow, v fill=yellow!30,
  w color=violet, w fill=violet!30,
  wrap=true, axes=none, tikz={}}}
\pgfkeys{/uvwp/.cd,
  v note/.code={\def\uvwp@nva{#1}\def\uvwp@nvb{#1}},
  v1 note/.store in=\uvwp@nva,    v2 note/.store in=\uvwp@nvb,
  note side/.store in=\uvwp@ns,
  axes length/.code={\def\uvwaxislength{#1}},
  axes gap/.code={\def\uvwaxisgap{#1}},
  unit/.code={\setlength{\uvw@unit}{#1}},
  u width/.store in=\uvwp@wu,
  v width/.code={\def\uvwp@wva{#1}\def\uvwp@wvb{#1}},
  v1 width/.store in=\uvwp@wva,   v2 width/.store in=\uvwp@wvb,
  w width/.store in=\uvwp@ww,     height/.store in=\uvwp@h,
  u label/.store in=\uvwp@lu,
  v label/.code={\def\uvwp@lv{#1}\uvw@ifempty\uvwp@lv
      {\def\uvwp@lva{}\def\uvwp@lvb{}}{\def\uvwp@lva{#1_{1}}\def\uvwp@lvb{#1_{2}}}},
  v1 label/.store in=\uvwp@lva,   v2 label/.store in=\uvwp@lvb,
  w label/.store in=\uvwp@lw,
  u length/.store in=\uvwp@du,
  v length/.code={\def\uvwp@dva{#1}\def\uvwp@dvb{#1}},
  v1 length/.store in=\uvwp@dva,  v2 length/.store in=\uvwp@dvb,
  w length/.store in=\uvwp@dw,    total length/.store in=\uvwp@dt,
  wrap label/.store in=\uvwp@dwrap,
  u color/.store in=\uvwp@cu,     u fill/.store in=\uvwp@fu,
  v color/.store in=\uvwp@cv,     v fill/.store in=\uvwp@fv,
  w color/.store in=\uvwp@cw,     w fill/.store in=\uvwp@fw,
  wrap/.is if=uvwp@wrap,
  axes/.store in=\uvwp@axspec,    tikz/.store in=\uvwp@tikz,
  measures/.code={\ifthenelse{\equal{#1}{false}}{%
      \def\uvwp@du{}\def\uvwp@dva{}\def\uvwp@dvb{}\def\uvwp@dw{}\def\uvwp@dt{}}{}},
}

\newcommand{\uvwper}[1][]{%
  \uvw@setaxis
  \begingroup
  \uvwp@defaults \pgfkeys{/uvwp/.cd,#1}%
  \pgfmathsetmacro{\uvwp@xa}{\uvwp@wu}%                   U  | V_1
  \pgfmathsetmacro{\uvwp@xb}{\uvwp@wu+\uvwp@wva}%          V_1| W
  \pgfmathsetmacro{\uvwp@xc}{\uvwp@wu+\uvwp@wva+\uvwp@ww}% W  | V_2
  \pgfmathsetmacro{\uvwp@tot}{\uvwp@xc+\uvwp@wvb}%
  \pgfmathsetmacro{\uvwp@yb}{-\uvwp@h-0.85}%   length measures, below
  \pgfmathsetmacro{\uvwp@yw}{\uvwp@yb-0.95}%   the identification arc
  \uvw@noteheight\uvwp@nva\uvwp@nha
  \uvw@noteheight\uvwp@nvb\uvwp@nhb
  \pgfmathsetmacro{\uvwp@nh}{max(\uvwp@nha,\uvwp@nhb)}%
  \pgfmathsetmacro{\uvwp@yt}{max(0.85,0.4+\uvwp@nh)}% total measure, above
  \ifthenelse{\equal{\uvwp@ns}{below}}{%
    \def\uvwp@na{north}\def\uvwp@nb{0}%
    \uvw@ifempty\uvwp@du{}{\def\uvwp@nb{1}}%
    \uvw@ifempty\uvwp@dva{}{\def\uvwp@nb{1}}%
    \uvw@ifempty\uvwp@dvb{}{\def\uvwp@nb{1}}%
    \uvw@ifempty\uvwp@dw{}{\def\uvwp@nb{1}}%
    \ifnum\uvwp@nb=1 \pgfmathsetmacro{\uvwp@ny}{\uvwp@yb-14pt/\uvw@unit}%
    \else            \pgfmathsetmacro{\uvwp@ny}{-\uvwp@h-0.1}\fi
    \pgfmathsetmacro{\uvwp@yw}{min(\uvwp@yw,\uvwp@ny-\uvwp@nh-0.3)}%
  }{\def\uvwp@na{south}\def\uvwp@ny{0.1}}%
  \begin{tikzpicture}[x=\uvw@unit, y=\uvw@unit,
      baseline={([yshift=-\uvw@axis]uvwp@c)}]
    \expandafter\tikzset\expandafter{\uvwp@tikz}%
    \coordinate (uvwp@c) at ({\uvwp@tot/2},{-\uvwp@h/2});
    \draw[uvw box, draw=\uvwp@cv, fill=\uvwp@fv]
      ({\uvwp@xa},0) rectangle ({\uvwp@xb},{-\uvwp@h});
    \draw[uvw box, draw=\uvwp@cv, fill=\uvwp@fv]
      ({\uvwp@xc},0) rectangle ({\uvwp@tot},{-\uvwp@h});
    \draw[uvw box, draw=\uvwp@cu, fill=\uvwp@fu]
      (0,0) rectangle ({\uvwp@xa},{-\uvwp@h});
    \draw[uvw box, draw=\uvwp@cw, fill=\uvwp@fw]
      ({\uvwp@xb},0) rectangle ({\uvwp@xc},{-\uvwp@h});
    \if\relax\detokenize\expandafter{\uvwp@lu}\relax\else
      \node[uvw label, text=black] at ({\uvwp@xa/2},{-\uvwp@h/2}) {$\uvwp@lu$};\fi
    \uvw@ifempty\uvwp@lva{}{%
      \node[uvw label, text=black, font=\footnotesize]
        at ({(\uvwp@xa+\uvwp@xb)/2},{-\uvwp@h/2}) {$\uvwp@lva$};}%
    \uvw@ifempty\uvwp@lvb{}{%
      \node[uvw label, text=black, font=\footnotesize]
        at ({(\uvwp@xc+\uvwp@tot)/2},{-\uvwp@h/2}) {$\uvwp@lvb$};}%
    \if\relax\detokenize\expandafter{\uvwp@lw}\relax\else
      \node[uvw label, text=black]
        at ({(\uvwp@xb+\uvwp@xc)/2},{-\uvwp@h/2}) {$\uvwp@lw$};\fi
    \def\uvw@dimside{below}%
    \if\relax\detokenize\expandafter{\uvwp@du}\relax\else
      \uvw@dimh{0}{\uvwp@xa}{\uvwp@yb}{-\uvwp@h}{\uvwp@du}\fi
    \uvw@ifempty\uvwp@dva{}{%
      \uvw@dimh{\uvwp@xa}{\uvwp@xb}{\uvwp@yb}{-\uvwp@h}{\uvwp@dva}}%
    \uvw@ifempty\uvwp@dvb{}{%
      \uvw@dimh{\uvwp@xc}{\uvwp@tot}{\uvwp@yb}{-\uvwp@h}{\uvwp@dvb}}%
    \if\relax\detokenize\expandafter{\uvwp@dw}\relax\else
      \uvw@dimh{\uvwp@xb}{\uvwp@xc}{\uvwp@yb}{-\uvwp@h}{\uvwp@dw}\fi
    \def\uvw@dimside{above}%
    \if\relax\detokenize\expandafter{\uvwp@dt}\relax\else
      \uvw@dimh{0}{\uvwp@tot}{\uvwp@yt}{0}{\uvwp@dt}\fi
    \ifuvwp@wrap
      \draw[uvw guide] (0,{-\uvwp@h}) -- (0,{\uvwp@yw+0.15});
      \draw[uvw guide] ({\uvwp@tot},{-\uvwp@h}) -- ({\uvwp@tot},{\uvwp@yw+0.15});
      \draw[uvw dim, rounded corners=2pt]
        (0,{\uvwp@yw}) -- (0,{\uvwp@yw-0.45})
        -- node[uvw dim label, midway, below] {$\uvwp@dwrap$}
           ({\uvwp@tot},{\uvwp@yw-0.45}) -- ({\uvwp@tot},{\uvwp@yw});
    \fi
    \uvw@ifempty\uvwp@nva{}{%
      \node[uvw note, fill=white, anchor=\uvwp@na]
        at ({(\uvwp@xa+\uvwp@xb)/2},{\uvwp@ny}) {\uvwp@nva};}%
    \uvw@ifempty\uvwp@nvb{}{%
      \node[uvw note, fill=white, anchor=\uvwp@na]
        at ({(\uvwp@xc+\uvwp@tot)/2},{\uvwp@ny}) {\uvwp@nvb};}%
    \if\relax\detokenize\expandafter{\uvwp@axspec}\relax\else
      \ifthenelse{\equal{\uvwp@axspec}{none}}{}{%
        \expandafter\uvw@drawaxes\expandafter{\uvwp@axspec}%
          {0}{\uvwp@tot}{-\uvwp@h}}%
    \fi
  \end{tikzpicture}%
  \endgroup}

\def\dsuvw@defaults{\pgfkeys{/dsuvw/.cd,
  unit=9mm, u width=3, v width=1.6, w width=3, height=3,
  u label=U, v label=V, w label=W,
  v length={L^{2/3}},
  v note={separating region}, note side=above,
  u color=teal,   u fill=teal!30,
  v color=yellow, v fill=yellow!30,
  w color=violet, w fill=violet!30,
  axes=none, tikz={}}}
\pgfkeys{/dsuvw/.cd,
  axes length/.code={\def\uvwaxislength{#1}},
  axes gap/.code={\def\uvwaxisgap{#1}},
  unit/.code={\setlength{\uvw@unit}{#1}},
  u width/.store in=\dsuvw@wu,   v width/.store in=\dsuvw@wv,
  w width/.store in=\dsuvw@ww,   height/.store in=\dsuvw@h,
  u label/.store in=\dsuvw@lu,   v label/.store in=\dsuvw@lv,
  w label/.store in=\dsuvw@lw,
  v length/.store in=\dsuvw@dv,
  v note/.store in=\dsuvw@nv,    note side/.store in=\dsuvw@ns,
  u color/.store in=\dsuvw@cu,   u fill/.store in=\dsuvw@fu,
  v color/.store in=\dsuvw@cv,   v fill/.store in=\dsuvw@fv,
  w color/.store in=\dsuvw@cw,   w fill/.store in=\dsuvw@fw,
  axes/.store in=\dsuvw@axspec,  tikz/.store in=\dsuvw@tikz,
}

\newcommand{\dsuvw}[1][]{%
  \uvw@setaxis
  \begingroup
  \dsuvw@defaults \pgfkeys{/dsuvw/.cd,#1}%
  \pgfmathsetmacro{\dsuvw@xv}{\dsuvw@wu}%               left edge of V
  \pgfmathsetmacro{\dsuvw@xw}{\dsuvw@wu+\dsuvw@wv}%     left edge of W
  \pgfmathsetmacro{\dsuvw@tot}{\dsuvw@wu+\dsuvw@wv+\dsuvw@ww}%
  \pgfmathsetmacro{\dsuvw@xm}{(\dsuvw@xv+\dsuvw@xw)/2}% centre of V
  \pgfmathsetmacro{\dsuvw@yb}{-\dsuvw@h-0.85}%          V measure
  \begin{tikzpicture}[x=\uvw@unit, y=\uvw@unit,
      baseline={([yshift=-\uvw@axis]dsuvw@c)}]
    \expandafter\tikzset\expandafter{\dsuvw@tikz}%
    \draw[uvw box, draw=\dsuvw@cv, fill=\dsuvw@fv]
      ({\dsuvw@xv},0) rectangle ({\dsuvw@xw},{-\dsuvw@h});
    \draw[uvw box, draw=\dsuvw@cu, fill=\dsuvw@fu]
      (0,0) rectangle ({\dsuvw@xv},{-\dsuvw@h});
    \draw[uvw box, draw=\dsuvw@cw, fill=\dsuvw@fw]
      ({\dsuvw@xw},0) rectangle ({\dsuvw@tot},{-\dsuvw@h});
    \coordinate (dsuvw@c) at ({\dsuvw@tot/2},{-\dsuvw@h/2});
    \if\relax\detokenize\expandafter{\dsuvw@lu}\relax\else
      \node[uvw label, text=black] at ({\dsuvw@xv/2},{-\dsuvw@h/2}) {$\dsuvw@lu$};\fi
    \if\relax\detokenize\expandafter{\dsuvw@lv}\relax\else
      \node[uvw label, text=black] at ({\dsuvw@xm},{-\dsuvw@h/2}) {$\dsuvw@lv$};\fi
    \if\relax\detokenize\expandafter{\dsuvw@lw}\relax\else
      \node[uvw label, text=black]
        at ({(\dsuvw@xw+\dsuvw@tot)/2},{-\dsuvw@h/2}) {$\dsuvw@lw$};\fi
    \def\uvw@dimside{below}%
    \if\relax\detokenize\expandafter{\dsuvw@dv}\relax\else
      \uvw@dimh{\dsuvw@xv}{\dsuvw@xw}{\dsuvw@yb}{-\dsuvw@h}{\dsuvw@dv}\fi
    \if\relax\detokenize\expandafter{\dsuvw@nv}\relax\else
      \ifthenelse{\equal{\dsuvw@ns}{below}}{%
        \node[uvw note, anchor=north] at ({\dsuvw@xm},{\dsuvw@yb-0.45})
          {\dsuvw@nv};%
      }{%
        \node[uvw note, anchor=south] at ({\dsuvw@xm},0.1) {\dsuvw@nv};%
      }\fi
    \if\relax\detokenize\expandafter{\dsuvw@axspec}\relax\else
      \ifthenelse{\equal{\dsuvw@axspec}{none}}{}{%
        \expandafter\uvw@drawaxes\expandafter{\dsuvw@axspec}%
          {0}{\dsuvw@tot}{-\dsuvw@h}}%
    \fi
  \end{tikzpicture}%
  \endgroup}

\newif\ifuvwbb@showb
\newif\ifuvwbb@span
\newif\ifuvwbb@showp
\def\uvwbb@defaults{\pgfkeys{/uvwbb/.cd,
  unit=7mm, b width=4.0, bp width=5.8, height=4.8, gap=1.7, ell=1.0,
  blobs={s/-0.02/1.1, x/0.50/1.7, s/1.02/1.25},
  blob size=0.5, long roughness=0.30, pixel=0.2,
  bp long seed=9, bp long scale=1.12,
  region=both, select=all,
  b label=B, bp label=B', ell label={\ell}, ell measure={},
  label base={\Sigma}, short sup={x \le \ell}, long sup={x > \ell},
  span line=false,
  short color=teal,   short fill=teal!30,
  cross color=violet, cross fill=violet!30,
  long color=yellow,  long fill=yellow!30,
  axes=x/y, tikz={}}}
\pgfkeys{/uvwbb/.cd,
  axes length/.code={\def\uvwaxislength{#1}},
  axes gap/.code={\def\uvwaxisgap{#1}},
  unit/.code={\setlength{\uvw@unit}{#1}},
  b width/.store in=\uvwbb@wb,    bp width/.store in=\uvwbb@wp,
  height/.store in=\uvwbb@h,      gap/.store in=\uvwbb@gap,
  ell/.store in=\uvwbb@ell,       blobs/.store in=\uvwbb@list,
  blob size/.store in=\uvwbb@bs,  long roughness/.store in=\uvwbb@rl,
  bp long seed/.store in=\uvwbb@sh, bp long scale/.store in=\uvwbb@lsc,
  pixel/.store in=\uvwbb@pix,
  select/.store in=\uvwbb@sel,
  region/.code={%
    \ifthenelse{\equal{#1}{b}}{\uvwbb@showbtrue\uvwbb@showpfalse}{%
    \ifthenelse{\equal{#1}{bp}}{\uvwbb@showbfalse\uvwbb@showptrue}{%
      \uvwbb@showbtrue\uvwbb@showptrue}}},
  b label/.store in=\uvwbb@lb,    bp label/.store in=\uvwbb@lp,
  label base/.store in=\uvwbb@base,
  span line/.is if=uvwbb@span,
  short sup/.store in=\uvwbb@ssup, long sup/.store in=\uvwbb@lsup,
  ell label/.store in=\uvwbb@le,  ell measure/.store in=\uvwbb@dme,
  short color/.store in=\uvwbb@cs, short fill/.store in=\uvwbb@fs,
  cross color/.store in=\uvwbb@cc, cross fill/.store in=\uvwbb@fc,
  long color/.store in=\uvwbb@cl,  long fill/.store in=\uvwbb@fl,
  axes/.store in=\uvwbb@axspec,   tikz/.store in=\uvwbb@tikz,
}

\def\uvwbb@blob#1#2#3#4#5#6{%
  \pgfmathsetmacro{\uvwbb@xz}{\uvwbb@pix*floor(-0.35/\uvwbb@pix)}%
  \pgfmathtruncatemacro{\uvwbb@nc}{ceil((#2-\uvwbb@xz)/\uvwbb@pix)}%
  \gdef\uvwbb@ptop{}\gdef\uvwbb@pbot{}%
  \foreach \uvwbbk in {1,...,\uvwbb@nc} {%
    \pgfmathsetmacro{\uvwbbxa}{\uvwbb@xz+(\uvwbbk-1)*\uvwbb@pix}%
    \pgfmathsetmacro{\uvwbbxb}{\uvwbb@xz+\uvwbbk*\uvwbb@pix}%
    \pgfmathsetmacro{\uvwbbf}{min((\uvwbbk-0.5)*\uvwbb@pix/(#2-\uvwbb@xz),1)}%
    \pgfmathsetmacro{\uvwbbm}{#3*0.22*sin(#4*47+\uvwbbf*250)}%
    \pgfmathsetmacro{\uvwbbru}{#3*sqrt(1-\uvwbbf^6)*(1+0.24*sin(180*\uvwbbf+30))
        *(1+#5*sin(#4*63+\uvwbbf*410)+0.6*#5*sin(#4*29+\uvwbbf*790))}%
    \pgfmathsetmacro{\uvwbbrd}{#3*sqrt(1-\uvwbbf^6)*(1+0.20*sin(180*\uvwbbf-25))
        *(1+#5*sin(#4*51+\uvwbbf*310+40)+0.6*#5*sin(#4*37+\uvwbbf*700+120))}%
    \pgfmathsetmacro{\uvwbbyd}{\uvwbb@pix*round((#1+\uvwbbm-\uvwbbrd)/\uvwbb@pix)}%
    \pgfmathsetmacro{\uvwbbyu}{max(\uvwbb@pix*round((#1+\uvwbbm+\uvwbbru)/\uvwbb@pix),
                                   \uvwbbyd+\uvwbb@pix)}%   at least one cell
    \xdef\uvwbb@ptop{\uvwbb@ptop (\uvwbbxa,\uvwbbyu) -- (\uvwbbxb,\uvwbbyu) -- }%
    \xdef\uvwbb@pbot{(\uvwbbxb,\uvwbbyd) -- (\uvwbbxa,\uvwbbyd) -- \uvwbb@pbot}%
  }%
  \expandafter\draw\expandafter[#6] \uvwbb@ptop \uvwbb@pbot cycle;}

\def\uvwbb@entry#1/#2/#3/#4\@nil{%
  \def\uvwbbt{#1}\def\uvwbbyf{#2}\def\uvwbbsz{#3}}
\def\uvwbb@blobs#1#2{%
  \begin{scope}
    \clip (0,0) rectangle ({#1},{-\uvwbb@h});
    \foreach \uvwbbe [count=\uvwbbi] in \uvwbb@list {%
      \expandafter\uvwbb@entry\uvwbbe///\@nil
      \if\relax\detokenize\expandafter{\uvwbbyf}\relax
        \pgfmathsetmacro{\uvwbby}{-(\uvwbbi-0.5)*\uvwbb@h/\uvwbb@n}%
      \else
        \pgfmathsetmacro{\uvwbby}{-\uvwbbyf*\uvwbb@h}%
      \fi
      \if\relax\detokenize\expandafter{\uvwbbsz}\relax
        \pgfmathsetmacro{\uvwbbhh}{\uvwbb@bs}%
      \else
        \pgfmathsetmacro{\uvwbbhh}{\uvwbbsz*\uvwbb@bs}%
      \fi
      \pgfmathsetmacro{\uvwbbts}{\uvwbb@xl*(0.38+0.16*(0.5+0.5*sin(\uvwbbi*77)))}%
      \pgfmathsetmacro{\uvwbbtc}{\uvwbb@xl+(\uvwbb@wb-\uvwbb@xl)
                                 *(0.45+0.40*(0.5+0.5*sin(\uvwbbi*59+30)))}%
      \pgfmathsetmacro{\uvwbbtl}{\uvwbb@wb+(\uvwbb@wp-\uvwbb@wb)
                                 *(0.30+0.35*(0.5+0.5*sin(\uvwbbi*83+70)))}%
      \ifthenelse{\equal{\uvwbb@sel}{long}}{\def\uvwbb@dos{0}}{\def\uvwbb@dos{1}}%
      \ifthenelse{\equal{\uvwbb@sel}{short}}{\def\uvwbb@dol{0}}{\def\uvwbb@dol{1}}%
      \ifthenelse{\equal{\uvwbbt}{s} \AND \equal{\uvwbb@dos}{1}}{%
        \uvwbb@blob{\uvwbby}{\uvwbbts}{\uvwbbhh}{\uvwbbi}{0.16}%
          {uvw box, draw=\uvwbb@cs, fill=\uvwbb@fs}}{}%
      \ifthenelse{\equal{\uvwbbt}{c} \AND \equal{\uvwbb@dol}{1}}{%
        \uvwbb@blob{\uvwbby}{\uvwbbtc}{\uvwbbhh}{\uvwbbi}{\uvwbb@rl}%
          {uvw box, draw=\uvwbb@cc, fill=\uvwbb@fc}}{}%
      \pgfmathsetmacro{\uvwbbsd}{\uvwbbi+\uvwbb@sh}%   seed of the B' lobe
      \pgfmathsetmacro{\uvwbbhp}{\uvwbb@lsc*\uvwbbhh}%
      \ifthenelse{\equal{\uvwbbt}{l} \AND \equal{\uvwbb@dol}{1}}{%
        \ifnum#2>1
          \uvwbb@blob{\uvwbby}{\uvwbbtl}{\uvwbbhp}{\uvwbbsd}{\uvwbb@rl}%
            {uvw box, draw=\uvwbb@cl, fill=\uvwbb@fl}%
        \fi}{}%
      \ifthenelse{\equal{\uvwbbt}{x} \AND \equal{\uvwbb@dol}{1}}{%
        \ifnum#2>1
          \uvwbb@blob{\uvwbby}{\uvwbbtl}{\uvwbbhp}{\uvwbbsd}{\uvwbb@rl}%
            {uvw box, draw=\uvwbb@cl, fill=\uvwbb@fl}%
        \else
          \uvwbb@blob{\uvwbby}{\uvwbbtc}{\uvwbbhh}{\uvwbbi}{\uvwbb@rl}%
            {uvw box, draw=\uvwbb@cc, fill=\uvwbb@fc}%
        \fi}{}%
    }%
  \end{scope}}

\def\uvwbb@region#1#2#3{%
  \uvwbb@blobs{#1}{#2}%
  \draw[uvw box] (0,0) rectangle ({#1},{-\uvwbb@h});
  \draw[uvw guide, densely dashed, black]
    ({\uvwbb@xl},0.12) -- ({\uvwbb@xl},{-\uvwbb@h});
  \if\relax\detokenize\expandafter{#3}\relax\else
    \ifthenelse{\equal{\uvwbb@sel}{short}}%
      {\edef\uvwbb@sup{^{\noexpand\uvwbb@ssup}}}%
      {\ifthenelse{\equal{\uvwbb@sel}{long}}%
        {\edef\uvwbb@sup{^{\noexpand\uvwbb@lsup}}}{\def\uvwbb@sup{}}}%
    \if\relax\detokenize\expandafter{\uvwbb@base}\relax
      \def\uvwbb@lab{#3\uvwbb@sup}%
    \else
      \def\uvwbb@lab{\uvwbb@base_{#3}\uvwbb@sup}%
    \fi
    \ifuvwbb@span
      \pgfmathsetmacro{\uvwbb@ys}{-\uvwbb@h-0.3}%
      \draw[uvw span] (0,{\uvwbb@ys}) -- ({#1},{\uvwbb@ys});
      \draw[uvw span] (0,{\uvwbb@ys}) -- (0,{\uvwbb@ys+0.16});
      \draw[uvw span] ({#1},{\uvwbb@ys}) -- ({#1},{\uvwbb@ys+0.16});
      \node[uvw span label, text=black, anchor=north]
        at ({#1/2},{\uvwbb@ys-0.08}) {$\uvwbb@lab$};
    \else
      \node[uvw span label, text=black, anchor=north]
        at ({#1/2},{-\uvwbb@h-0.1}) {$\uvwbb@lab$};
    \fi\fi
  \if\relax\detokenize\expandafter{\uvwbb@le}\relax\else
    \node[uvw dim label, anchor=south] at ({\uvwbb@xl},0.1) {$\uvwbb@le$};\fi}

\newcommand{\uvwbb}[1][]{%
  \uvw@setaxis
  \begingroup
  \uvwbb@defaults \pgfkeys{/uvwbb/.cd,#1}%
  \xdef\uvwbb@n{1}%
  \foreach \uvwbbe [count=\uvwbbc] in \uvwbb@list {\xdef\uvwbb@n{\uvwbbc}}%
  \pgfmathsetmacro{\uvwbb@xl}{\uvwbb@wb-\uvwbb@ell}%   the dashed line
  \ifuvwbb@showb \ifuvwbb@showp
    \pgfmathsetmacro{\uvwbb@ytwo}{-\uvwbb@h-\uvwbb@gap}% both, B' below B
  \else
    \pgfmathsetmacro{\uvwbb@ytwo}{0}%
  \fi \else \pgfmathsetmacro{\uvwbb@ytwo}{0}\fi
  \ifuvwbb@showp \pgfmathsetmacro{\uvwbb@bot}{\uvwbb@ytwo-\uvwbb@h}%
  \else          \pgfmathsetmacro{\uvwbb@bot}{-\uvwbb@h}\fi
  \pgfmathsetmacro{\uvwbb@yb}{\uvwbb@bot-0.85}%        the ell measure
  \pgfmathsetmacro{\uvwbb@wmax}{\ifuvwbb@showp \uvwbb@wp \else \uvwbb@wb \fi}%
  \begin{tikzpicture}[x=\uvw@unit, y=\uvw@unit,
      baseline={([yshift=-\uvw@axis]uvwbb@c)}]
    \expandafter\tikzset\expandafter{\uvwbb@tikz}%
    \coordinate (uvwbb@c) at ({\uvwbb@wmax/2},{\uvwbb@bot/2});
    \ifuvwbb@showb \uvwbb@region{\uvwbb@wb}{1}{\uvwbb@lb}\fi
    \ifuvwbb@showp
      \begin{scope}[shift={(0,{\uvwbb@ytwo})}]
        \uvwbb@region{\uvwbb@wp}{2}{\uvwbb@lp}
      \end{scope}
    \fi
    \def\uvw@dimside{below}%
    \ifuvwbb@showp \let\uvwbb@lbot\uvwbb@lp \else \let\uvwbb@lbot\uvwbb@lb \fi
    \if\relax\detokenize\expandafter{\uvwbb@lbot}\relax
      \pgfmathsetmacro{\uvwbb@yls}{0}%
    \else
      \ifuvwbb@span \pgfmathsetmacro{\uvwbb@yls}{0.95}% room for line + label
      \else         \pgfmathsetmacro{\uvwbb@yls}{0.75}\fi% room for the label
    \fi
    \if\relax\detokenize\expandafter{\uvwbb@dme}\relax
      \pgfmathsetmacro{\uvwbb@yax}{\uvwbb@bot-0.45-\uvwbb@yls}%
    \else
      \uvw@dimh{\uvwbb@xl}{\uvwbb@wb}{\uvwbb@yb-\uvwbb@yls}{\uvwbb@bot}{\uvwbb@dme}%
      \pgfmathsetmacro{\uvwbb@yax}{\uvwbb@yb-0.2-\uvwbb@yls}%
    \fi
    \if\relax\detokenize\expandafter{\uvwbb@axspec}\relax\else
      \ifthenelse{\equal{\uvwbb@axspec}{none}}{}{%
        \expandafter\uvw@drawaxes\expandafter{\uvwbb@axspec}%
          {0}{\uvwbb@wmax}{\uvwbb@bot}}%
    \fi
  \end{tikzpicture}%
  \endgroup}

\newlength{\uvwbbidsep}
\newcommand{\uvwbbid}[1][]{%
  \begin{tabular}{@{}c@{}}
    \uvwbb[unit=4.2mm, #1, axes=none, region=b]$\;=\;$%
    \uvwbb[unit=4.2mm, #1, ell label={}, axes=none, region=b, select=short]$\;+\;$%
    \uvwbb[unit=4.2mm, #1, ell label={}, axes=none, region=b, select=long]%
    \hspace{2.2em}
    \uvwbb[unit=4.2mm, #1, axes=none, region=bp]$\;=\;$%
    \uvwbb[unit=4.2mm, #1, ell label={}, axes=none, region=bp, select=short]$\;+\;$%
    \uvwbb[unit=4.2mm, #1, ell label={}, axes=none, region=bp, select=long]%
    \\[\uvwbbidsep]
    \uvwbb[unit=4.2mm, #1, ell label={}, axes=x/y, axes length=3.5mm,
           axes gap=1.5mm, region=b]$\;-\;$%
    \uvwbb[unit=4.2mm, #1, ell label={}, axes=none, region=bp]$\;=\;$%
    \uvwbb[unit=4.2mm, #1, ell label={}, axes=none, region=b, select=long]$\;-\;$%
    \uvwbb[unit=4.2mm, #1, ell label={}, axes=none, region=bp, select=long]%
  \end{tabular}}

\newif\ifuvwsp@wrap
\def\uvwsp@defaults{\pgfkeys{/uvwsp/.cd,
  unit=7mm, total=16, u width=2.4, w width=1.3, v width=0.8, height=1.6,
  v labels={1,dots,i,dots,I}, right index=1, left index=2,
  u label=U, v label=V, w label=W,
  total length={\Theta(L)}, v length={L^{2/3}}, u length={\Omega(L/3)},
  w length={}, wrap label={},
  u color=teal,   u fill=teal!30,
  v color=yellow, v fill=yellow!30,
  w color=violet, w fill=violet!30,
  stacked=false, stack labels={1,2,dots,i,dots,I},
  row height=1.6, row gap=0.5, dots gap=1.1, touching=true, dots weight=0,
  wrap=true, axes=none, tikz={}}}
\def\uvwsp@stackdefaults{\pgfkeys{/uvwsp/.cd,
  unit=4mm, total=17.5, u width=2.8, w width=1.4, v width=1.35,
  total length={}}}
\newif\ifuvwsp@stacked
\newif\ifuvwsp@touch
\pgfkeys{/uvwsp/.cd,
  touching/.is if=uvwsp@touch,
  axes length/.code={\def\uvwaxislength{#1}},
  axes gap/.code={\def\uvwaxisgap{#1}},
  stacked/.is if=uvwsp@stacked,
  stack labels/.store in=\uvwsp@sl,
  row height/.store in=\uvwsp@rh,  row gap/.store in=\uvwsp@rg,
  dots gap/.store in=\uvwsp@dg,  dots weight/.store in=\uvwsp@dw@,
  measures/.code={\ifthenelse{\equal{#1}{false}}{%
      \def\uvwsp@dt{}\def\uvwsp@dv{}\def\uvwsp@du{}\def\uvwsp@dw{}}{}},
  unit/.code={\setlength{\uvw@unit}{#1}},
  total/.store in=\uvwsp@tot,     u width/.store in=\uvwsp@wu,
  w width/.store in=\uvwsp@ww,    v width/.store in=\uvwsp@wv,
  height/.store in=\uvwsp@h,
  v labels/.store in=\uvwsp@vl,
  right index/.store in=\uvwsp@ir, left index/.store in=\uvwsp@il,
  u label/.store in=\uvwsp@lu,    v label/.store in=\uvwsp@lv,
  w label/.store in=\uvwsp@lw,
  total length/.store in=\uvwsp@dt, v length/.store in=\uvwsp@dv,
  u length/.store in=\uvwsp@du,     w length/.store in=\uvwsp@dw,
  wrap label/.store in=\uvwsp@dwrap,
  u color/.store in=\uvwsp@cu,    u fill/.store in=\uvwsp@fu,
  v color/.store in=\uvwsp@cv,    v fill/.store in=\uvwsp@fv,
  w color/.store in=\uvwsp@cw,    w fill/.store in=\uvwsp@fw,
  wrap/.is if=uvwsp@wrap,
  axes/.store in=\uvwsp@axspec,   tikz/.store in=\uvwsp@tikz,
}

\def\uvwsp@window#1#2#3{%
  \ifthenelse{\equal{#3}{dots}}{%
    \node[uvw label] at ({#1},{-\uvwsp@h/2}) {$\cdots$};%
  }{%
    \draw[uvw box, draw=\uvwsp@cv, fill=\uvwsp@fv]
      ({#1-\uvwsp@wv/2},0) rectangle ({#1+\uvwsp@wv/2},{-\uvwsp@h});
    \if\relax\detokenize\expandafter{\uvwsp@lv}\relax\else
      \node[uvw dim label, text=black, anchor=north]
        at ({#1},{\uvwsp@yl}) {$\uvwsp@lv_{#2,#3}$};\fi
  }}

\def\uvwsp@row#1#2#3{% #1 = offset of the V's from U_1, #2 = entry, #3 = top
  \pgfmathsetmacro{\uvwspxr}{\uvwsp@xur+#1}%              V_{1,k}
  \pgfmathsetmacro{\uvwspxl}{\uvwsp@xul-#1-\uvwsp@wv}%    V_{2,k}
  \pgfmathsetmacro{\uvwspyb}{#3-\uvwsp@rh}%
  \pgfmathsetmacro{\uvwspym}{#3-\uvwsp@rh/2}%
  \draw[uvw box, draw=\uvwsp@cv, fill=\uvwsp@fv]
    ({\uvwspxl},{#3}) rectangle ({\uvwspxl+\uvwsp@wv},{\uvwspyb});
  \draw[uvw box, draw=\uvwsp@cv, fill=\uvwsp@fv]
    ({\uvwspxr},{#3}) rectangle ({\uvwspxr+\uvwsp@wv},{\uvwspyb});
  \draw[uvw box, draw=\uvwsp@cw, fill=\uvwsp@fw] (0,{#3}) rectangle ({\uvwspxl},{\uvwspyb});
  \draw[uvw box, draw=\uvwsp@cw, fill=\uvwsp@fw]
    ({\uvwspxr+\uvwsp@wv},{#3}) rectangle ({\uvwsp@tot},{\uvwspyb});
  \draw[uvw box, draw=\uvwsp@cu, fill=\uvwsp@fu]
    ({\uvwspxl+\uvwsp@wv},{#3}) rectangle ({\uvwspxr},{\uvwspyb});
  \if\relax\detokenize\expandafter{\uvwsp@lu}\relax\else
    \node[uvw label, text=black] at ({\uvwsp@tot/2},{\uvwspym}) {$\uvwsp@lu_{#2}$};\fi
  \if\relax\detokenize\expandafter{\uvwsp@lw}\relax\else
    \node[uvw label, text=black] at ({\uvwspxl/2},{\uvwspym}) {$\uvwsp@lw_{#2}$};
    \node[uvw label, text=black] at ({(\uvwspxr+\uvwsp@wv+\uvwsp@tot)/2},{\uvwspym})
      {$\uvwsp@lw_{#2}$};\fi
  \if\relax\detokenize\expandafter{\uvwsp@lv}\relax\else
    \node[uvw label, text=black, font=\scriptsize]
      at ({\uvwspxl+\uvwsp@wv/2},{\uvwspym}) {$\uvwsp@lv_{\uvwsp@il,#2}$};
    \node[uvw label, text=black, font=\scriptsize]
      at ({\uvwspxr+\uvwsp@wv/2},{\uvwspym}) {$\uvwsp@lv_{\uvwsp@ir,#2}$};\fi}
\def\uvwsp@stackpic{%
  \xdef\uvwsp@k{1}%
  \foreach \uvwspe [count=\uvwspc] in \uvwsp@sl {\xdef\uvwsp@k{\uvwspc}}%
  \pgfmathsetmacro{\uvwsp@xul}{(\uvwsp@tot-\uvwsp@wu)/2}%
  \pgfmathsetmacro{\uvwsp@xur}{(\uvwsp@tot+\uvwsp@wu)/2}%
  \pgfmathsetmacro{\uvwsp@xwl}{\uvwsp@ww}%
  \pgfmathsetmacro{\uvwsp@xwr}{\uvwsp@tot-\uvwsp@ww}%
  \xdef\uvwsp@n{0}\xdef\uvwsp@m{0}%
  \foreach \uvwspe in \uvwsp@sl {%
    \ifthenelse{\equal{\uvwspe}{dots}}%
      {\pgfmathtruncatemacro{\uvwspc}{\uvwsp@m+1}\xdef\uvwsp@m{\uvwspc}}%
      {\pgfmathtruncatemacro{\uvwspc}{\uvwsp@n+1}\xdef\uvwsp@n{\uvwspc}}}%
  \ifuvwsp@touch \ifnum\uvwsp@m>0
      \def\uvwsp@ra{0}\def\uvwsp@rb{1}%
    \else \def\uvwsp@ra{1}\let\uvwsp@rb\uvwsp@dw@\fi
  \else \def\uvwsp@ra{1}\let\uvwsp@rb\uvwsp@dw@\fi
  \pgfmathsetmacro{\uvwsp@G}{(\uvwsp@n-1)*\uvwsp@ra+\uvwsp@m*\uvwsp@rb}%
  \pgfmathsetmacro{\uvwsp@gap}{\uvwsp@G>0 ?
      (\uvwsp@xwr-\uvwsp@xur-\uvwsp@n*\uvwsp@wv)/\uvwsp@G : 0}%
  \ifdim\uvwsp@gap pt<0pt
    \PackageWarning{uvw}{\string\uvwshiftper[stacked]: the windows do not
      fit between U_1 and W_I and overlap; make total larger or u width,
      w width or v width smaller}\fi
  \xdef\uvwspy{0}\xdef\uvwspfirst{1}%
  \foreach \uvwspe in \uvwsp@sl {%
    \ifthenelse{\equal{\uvwspe}{dots}}{%
      \pgfmathsetmacro{\uvwspyn}{\uvwspy-\uvwsp@dg}\xdef\uvwspy{\uvwspyn}\xdef\uvwspfirst{1}%
    }{%
      \ifnum\uvwspfirst=1 \xdef\uvwspfirst{0}\else
        \pgfmathsetmacro{\uvwspyn}{\uvwspy-\uvwsp@rg}\xdef\uvwspy{\uvwspyn}\fi
      \pgfmathsetmacro{\uvwspyn}{\uvwspy-\uvwsp@rh}\xdef\uvwspy{\uvwspyn}%
    }}%
  \let\uvwsp@ylast\uvwspy
  \begin{tikzpicture}[x=\uvw@unit, y=\uvw@unit,
      baseline={([yshift=-\uvw@axis]uvwsp@c)}]
    \expandafter\tikzset\expandafter{\uvwsp@tikz}%
    \xdef\uvwspy{0}\xdef\uvwspfirst{1}\xdef\uvwspj{0}\xdef\uvwspg{0}%
    \foreach \uvwspe in \uvwsp@sl {%
      \ifthenelse{\equal{\uvwspe}{dots}}{%
        \uvw@vdots{\uvwsp@tot/2}{\uvwspy-\uvwsp@dg/2}%
        \pgfmathsetmacro{\uvwspyn}{\uvwspy-\uvwsp@dg}\xdef\uvwspy{\uvwspyn}\xdef\uvwspfirst{1}%
        \pgfmathsetmacro{\uvwspgn}{\uvwspg+\uvwsp@rb}\xdef\uvwspg{\uvwspgn}%
      }{%
        \ifnum\uvwspfirst=1 \xdef\uvwspfirst{0}\else
          \pgfmathsetmacro{\uvwspyn}{\uvwspy-\uvwsp@rg}\xdef\uvwspy{\uvwspyn}\fi
        \pgfmathsetmacro{\uvwspoff}{\uvwspj*\uvwsp@wv+\uvwspg*\uvwsp@gap}%
        \uvwsp@row{\uvwspoff}{\uvwspe}{\uvwspy}%
        \pgfmathsetmacro{\uvwspyn}{\uvwspy-\uvwsp@rh}\xdef\uvwspy{\uvwspyn}%
        \pgfmathtruncatemacro{\uvwspjn}{\uvwspj+1}\xdef\uvwspj{\uvwspjn}%
        \pgfmathsetmacro{\uvwspgn}{\uvwspg+\uvwsp@ra}\xdef\uvwspg{\uvwspgn}%
      }}%
    \coordinate (uvwsp@c) at ({\uvwsp@tot/2},{\uvwsp@ylast/2});
    \def\uvw@dimside{above}%
    \if\relax\detokenize\expandafter{\uvwsp@du}\relax\else
      \uvw@dimh{\uvwsp@xul}{\uvwsp@xur}{0.5}{0}{\uvwsp@du}\fi
    \if\relax\detokenize\expandafter{\uvwsp@dv}\relax\else
      \uvw@dimh{\uvwsp@xul-\uvwsp@wv}{\uvwsp@xul}{0.5}{0}{\uvwsp@dv}%
      \uvw@dimh{\uvwsp@xur}{\uvwsp@xur+\uvwsp@wv}{0.5}{0}{\uvwsp@dv}\fi
    \def\uvw@dimside{below}%
    \pgfmathsetmacro{\uvwsp@yb}{\uvwsp@ylast-0.85}%
    \def\uvwsp@nbelow{0}%
    \uvw@ifempty\uvwsp@dw{%
      \let\uvwsp@ytb\uvwsp@yb
      \uvw@ifaxes\uvwsp@axspec{% clear of the coordinate arrows
        \pgfmathsetmacro{\uvwsp@ytb}{min(\uvwsp@ytb,
          \uvwsp@ylast-(\uvwaxisgap+2.8mm)/\uvw@unit)}}{}%
    }{%
      \uvw@dimh{0}{\uvwsp@xwl}{\uvwsp@yb}{\uvwsp@ylast}{\uvwsp@dw}%
      \uvw@dimh{\uvwsp@xwr}{\uvwsp@tot}{\uvwsp@yb}{\uvwsp@ylast}{\uvwsp@dw}%
      \pgfmathsetmacro{\uvwsp@ytb}{\uvwsp@yb-17pt/\uvw@unit}%
      \let\uvwsp@ylow\uvwsp@yb \def\uvwsp@nbelow{1}%
    }%
    \uvw@ifempty\uvwsp@dt{}{%
      \uvw@dimh{0}{\uvwsp@tot}{\uvwsp@ytb}{\uvwsp@ylast}{\uvwsp@dt}%
      \let\uvwsp@ylow\uvwsp@ytb \def\uvwsp@nbelow{1}}%
    \ifnum\uvwsp@nbelow=1
      \pgfmathsetmacro{\uvwsp@yw}{\uvwsp@ylow-17pt/\uvw@unit-0.1}%
    \else
      \pgfmathsetmacro{\uvwsp@yw}{\uvwsp@ylast-0.35}%
      \uvw@ifaxes\uvwsp@axspec{%
        \pgfmathsetmacro{\uvwsp@yw}{min(\uvwsp@yw,
          \uvwsp@ylast+0.45-(\uvwaxisgap+3.5mm)/\uvw@unit)}}{}%
    \fi
    \ifuvwsp@wrap
      \draw[uvw guide] (0,{\uvwsp@ylast}) -- (0,{\uvwsp@yw+0.15});
      \draw[uvw guide] ({\uvwsp@tot},{\uvwsp@ylast}) -- ({\uvwsp@tot},{\uvwsp@yw+0.15});
      \draw[uvw dim, rounded corners=2pt]
        (0,{\uvwsp@yw}) -- (0,{\uvwsp@yw-0.45})
        -- node[uvw dim label, midway, below] {$\uvwsp@dwrap$}
           ({\uvwsp@tot},{\uvwsp@yw-0.45}) -- ({\uvwsp@tot},{\uvwsp@yw});
    \fi
    \if\relax\detokenize\expandafter{\uvwsp@axspec}\relax\else
      \ifthenelse{\equal{\uvwsp@axspec}{none}}{}{%
        \expandafter\uvw@drawaxes\expandafter{\uvwsp@axspec}%
          {0}{\uvwsp@tot}{\uvwsp@ylast}}%
    \fi
  \end{tikzpicture}}

\newcommand{\uvwshiftper}[1][]{%
  \uvw@setaxis
  \begingroup
  \uvwsp@defaults \pgfkeys{/uvwsp/.cd,#1}%
  \ifuvwsp@stacked
    \uvwsp@stackdefaults \pgfkeys{/uvwsp/.cd,#1}\uvwsp@stackpic
  \else
  \xdef\uvwsp@k{1}%
  \foreach \uvwspe [count=\uvwspc] in \uvwsp@vl {\xdef\uvwsp@k{\uvwspc}}%
  \pgfmathsetmacro{\uvwsp@xul}{(\uvwsp@tot-\uvwsp@wu)/2}%   U starts
  \pgfmathsetmacro{\uvwsp@xur}{(\uvwsp@tot+\uvwsp@wu)/2}%   U ends
  \pgfmathsetmacro{\uvwsp@xwl}{\uvwsp@ww}%                  left W ends
  \pgfmathsetmacro{\uvwsp@xwr}{\uvwsp@tot-\uvwsp@ww}%       right W starts
  \pgfmathsetmacro{\uvwsp@step}{\uvwsp@k>1 ?
      (\uvwsp@xwr-\uvwsp@xur-\uvwsp@wv)/(\uvwsp@k-1) : 0}%
  \ifdim\uvwsp@step pt<\uvwsp@wv pt \ifnum\uvwsp@k>1
    \PackageWarning{uvw}{\string\uvwshiftper: the windows overlap; make
      total larger or u width, w width or v width smaller}\fi\fi
  \pgfmathsetmacro{\uvwsp@yl}{-\uvwsp@h-0.3}%   the V labels
  \pgfmathsetmacro{\uvwsp@yb}{-\uvwsp@h-1.15}%  the U (and W) measures
  \pgfmathsetmacro{\uvwsp@yw}{\uvwsp@yb-0.95}%  the bracket glueing the ends
  \pgfmathsetmacro{\uvwsp@yv}{0.5}%             width of V_{2,1}, V_{1,1}
  \if\relax\detokenize\expandafter{\uvwsp@dv}\relax
    \pgfmathsetmacro{\uvwsp@yt}{0.5}%
  \else
    \pgfmathsetmacro{\uvwsp@yt}{1.3}%
  \fi
  \begin{tikzpicture}[x=\uvw@unit, y=\uvw@unit,
      baseline={([yshift=-\uvw@axis]uvwsp@c)}]
    \expandafter\tikzset\expandafter{\uvwsp@tikz}%
    \coordinate (uvwsp@c) at ({\uvwsp@tot/2},{-\uvwsp@h/2});
    \fill[\uvwsp@fu] ({\uvwsp@xul},0) rectangle ({\uvwsp@xur},{-\uvwsp@h});
    \fill[\uvwsp@fw] (0,0) rectangle ({\uvwsp@xwl},{-\uvwsp@h});
    \fill[\uvwsp@fw] ({\uvwsp@xwr},0) rectangle ({\uvwsp@tot},{-\uvwsp@h});
    \foreach \uvwspx in {\uvwsp@xwl,\uvwsp@xul,\uvwsp@xur,\uvwsp@xwr}
      {\draw[uvw guide, densely dashed] ({\uvwspx},0) -- ({\uvwspx},{-\uvwsp@h});}
    \foreach \uvwspe [count=\uvwspi] in \uvwsp@vl {%
      \pgfmathsetmacro{\uvwspxr}{\uvwsp@xur+(\uvwspi-1)*\uvwsp@step+0.5*\uvwsp@wv}%
      \pgfmathsetmacro{\uvwspxl}{\uvwsp@xul-(\uvwspi-1)*\uvwsp@step-0.5*\uvwsp@wv}%
      \uvwsp@window{\uvwspxr}{\uvwsp@ir}{\uvwspe}%
      \uvwsp@window{\uvwspxl}{\uvwsp@il}{\uvwspe}%
    }%
    \draw[uvw box] (0,0) rectangle ({\uvwsp@tot},{-\uvwsp@h});
    \if\relax\detokenize\expandafter{\uvwsp@lu}\relax\else
      \node[uvw label, text=black] at ({\uvwsp@tot/2},{-\uvwsp@h/2})
        {$\uvwsp@lu$};\fi
    \if\relax\detokenize\expandafter{\uvwsp@lw}\relax\else
      \node[uvw label, text=black] at ({\uvwsp@xwl/2},{-\uvwsp@h/2})
        {$\uvwsp@lw$};
      \node[uvw label, text=black] at ({(\uvwsp@xwr+\uvwsp@tot)/2},{-\uvwsp@h/2})
        {$\uvwsp@lw$};\fi
    \def\uvw@dimside{above}%
    \if\relax\detokenize\expandafter{\uvwsp@dv}\relax\else
      \uvw@dimh{\uvwsp@xul-\uvwsp@wv}{\uvwsp@xul}{\uvwsp@yv}{0}{\uvwsp@dv}%
      \uvw@dimh{\uvwsp@xur}{\uvwsp@xur+\uvwsp@wv}{\uvwsp@yv}{0}{\uvwsp@dv}\fi
    \if\relax\detokenize\expandafter{\uvwsp@dt}\relax\else
      \uvw@dimh{0}{\uvwsp@tot}{\uvwsp@yt}{0}{\uvwsp@dt}\fi
    \def\uvw@dimside{below}%
    \if\relax\detokenize\expandafter{\uvwsp@du}\relax\else
      \uvw@dimh{\uvwsp@xul}{\uvwsp@xur}{\uvwsp@yb}{-\uvwsp@h}{\uvwsp@du}\fi
    \if\relax\detokenize\expandafter{\uvwsp@dw}\relax\else
      \uvw@dimh{0}{\uvwsp@xwl}{\uvwsp@yb}{-\uvwsp@h}{\uvwsp@dw}%
      \uvw@dimh{\uvwsp@xwr}{\uvwsp@tot}{\uvwsp@yb}{-\uvwsp@h}{\uvwsp@dw}\fi
    \ifuvwsp@wrap
      \draw[uvw guide] (0,{-\uvwsp@h}) -- (0,{\uvwsp@yw+0.15});
      \draw[uvw guide] ({\uvwsp@tot},{-\uvwsp@h}) -- ({\uvwsp@tot},{\uvwsp@yw+0.15});
      \draw[uvw dim, rounded corners=2pt]
        (0,{\uvwsp@yw}) -- (0,{\uvwsp@yw-0.45})
        -- node[uvw dim label, midway, below] {$\uvwsp@dwrap$}
           ({\uvwsp@tot},{\uvwsp@yw-0.45}) -- ({\uvwsp@tot},{\uvwsp@yw});
    \fi
    \if\relax\detokenize\expandafter{\uvwsp@axspec}\relax\else
      \ifthenelse{\equal{\uvwsp@axspec}{none}}{}{%
        \expandafter\uvw@drawaxes\expandafter{\uvwsp@axspec}%
          {0}{\uvwsp@tot}{-\uvwsp@h}}%
    \fi
  \end{tikzpicture}%
  \fi
  \endgroup}

\makeatother
\tdplotsetmaincoords{60}{130}
\tikzset{
  haxis/.style={-stealth, line width=0.6pt, black},
  haxis label/.style={font=\footnotesize, inner sep=1.5pt, text=black},
  hcheck face/.style={line width=0.4pt, black!70, line cap=round,
                      line join=round},
  hcut/.style={line width=0.4pt, line cap=round, draw=haahface!60!gray},
  hwall/.style={fill=black!30, draw=black!55, line width=0.4pt},
}

\newif\ifhaahwalls
\pgfkeys{/haah/.cd,
  scale/.store in=\haahscale,
  cube opacity/.store in=\haahop,
  plane opacity/.store in=\haahpop,
  planes/.is if=haahwalls,
}
\def\haahdefaults{\pgfkeys{/haah/.cd, scale=0.8, cube opacity=1,
  plane opacity=0.35, planes=true}}

\newcommand{\haahface}[6]{%
  \colorlet{haahface}{#1}%
  \ifx#2z\def\haahP##1##2{(##1,##2,#3)}\fi
  \ifx#2x\def\haahP##1##2{(#3,##1,##2)}\fi
  \ifx#2y\def\haahP##1##2{(##1,#3,##2)}\fi
  \pgfmathtruncatemacro{\haahuo}{#4}\pgfmathtruncatemacro{\haahui}{#4+1}%
  \pgfmathtruncatemacro{\haahvo}{#5}\pgfmathtruncatemacro{\haahvi}{#5+1}%
  \ifhaahwalls
    \pgfmathtruncatemacro{\haahcin}{#3>0.5 ? 1 : 0}%
    \pgfmathtruncatemacro{\haahus}{#4<0.5 ? 1 : 0}%
    \pgfmathtruncatemacro{\haahvs}{#5<0.5 ? 1 : 0}%
  \else \def\haahcin{1}\def\haahus{0}\def\haahvs{0}\fi
  \ifnum\haahus=1 \def\haahum{0.5}\else\let\haahum\haahuo\fi
  \ifnum\haahvs=1 \def\haahvm{0.5}\else\let\haahvm\haahvo\fi
  \ifnum#6=1
    \ifnum\haahcin=1
      \fill[fill=#1, fill opacity=\haahop]
        \haahP{\haahum}{\haahvm} -- \haahP{\haahui}{\haahvm}
        -- \haahP{\haahui}{\haahvi} -- \haahP{\haahum}{\haahvi} -- cycle;
      \ifnum\haahus=1
        \draw[hcut] \haahP{0.5}{\haahvm} -- \haahP{0.5}{\haahvi};\fi
      \ifnum\haahvs=1
        \draw[hcut] \haahP{\haahum}{0.5} -- \haahP{\haahui}{0.5};\fi
      \draw[hcheck face] \haahP{\haahui}{\haahvm} -- \haahP{\haahui}{\haahvi}
        -- \haahP{\haahum}{\haahvi};
      \ifnum\haahus=0
        \draw[hcheck face] \haahP{\haahum}{\haahvi} -- \haahP{\haahum}{\haahvm};\fi
      \ifnum\haahvs=0
        \draw[hcheck face] \haahP{\haahum}{\haahvm} -- \haahP{\haahui}{\haahvm};\fi
    \fi
  \else
    \ifnum\haahcin=0
      \fill[fill=#1, fill opacity=\haahop]
        \haahP{\haahuo}{\haahvo} -- \haahP{\haahui}{\haahvo}
        -- \haahP{\haahui}{\haahvi} -- \haahP{\haahuo}{\haahvi} -- cycle;
      \draw[hcheck face]
        \haahP{\haahuo}{\haahvo} -- \haahP{\haahui}{\haahvo}
        -- \haahP{\haahui}{\haahvi} -- \haahP{\haahuo}{\haahvi} -- cycle;
    \else
      \pgfmathtruncatemacro{\haahany}{\haahus+\haahvs}%
      \ifnum\haahany>0
        \fill[fill=#1, fill opacity=\haahop, even odd rule]
          \haahP{\haahuo}{\haahvo} -- \haahP{\haahui}{\haahvo}
          -- \haahP{\haahui}{\haahvi} -- \haahP{\haahuo}{\haahvi} -- cycle
          \haahP{\haahum}{\haahvm} -- \haahP{\haahui}{\haahvm}
          -- \haahP{\haahui}{\haahvi} -- \haahP{\haahum}{\haahvi} -- cycle;
        \ifnum\haahus=1
          \draw[hcheck face] \haahP{\haahuo}{\haahvo} -- \haahP{\haahuo}{\haahvi};
          \draw[hcheck face] \haahP{\haahuo}{\haahvi} -- \haahP{0.5}{\haahvi};
        \fi
        \ifnum\haahvs=1
          \draw[hcheck face] \haahP{\haahuo}{\haahvo} -- \haahP{\haahui}{\haahvo};
          \draw[hcheck face] \haahP{\haahui}{\haahvo} -- \haahP{\haahui}{0.5};
        \fi
        \ifnum\haahus=1 \ifnum\haahvs=0
          \draw[hcheck face] \haahP{\haahuo}{\haahvo} -- \haahP{0.5}{\haahvo};
        \fi\fi
        \ifnum\haahvs=1 \ifnum\haahus=0
          \draw[hcheck face] \haahP{\haahuo}{\haahvo} -- \haahP{\haahuo}{0.5};
        \fi\fi
      \fi
    \fi
  \fi}

\newcommand{\haahcubepart}[5]{%
  \pgfmathtruncatemacro{\haahxb}{#2+1}%
  \pgfmathtruncatemacro{\haahyb}{#3+1}%
  \pgfmathtruncatemacro{\haahzb}{#4+1}%
  \ifdim\haahop pt<1pt
    \haahface{#1!30}{z}{#4}{#2}{#3}{#5}%
    \haahface{#1!55}{x}{#2}{#3}{#4}{#5}%
    \haahface{#1!75}{y}{#3}{#2}{#4}{#5}%
  \fi
  \haahface{#1!50}{z}{\haahzb}{#2}{#3}{#5}%
  \haahface{#1}{x}{\haahxb}{#3}{#4}{#5}%
  \haahface{#1!85!black}{y}{\haahyb}{#2}{#4}{#5}}

\newcommand{\haahwalls}[1]{%
  \draw[hwall, fill opacity=\haahpop]
    (0.5,0.5,0.5) -- (0.5,#1,0.5) -- (0.5,#1,#1) -- (0.5,0.5,#1) -- cycle;
  \draw[hwall, fill opacity=\haahpop]
    (0.5,0.5,0.5) -- (#1,0.5,0.5) -- (#1,0.5,#1) -- (0.5,0.5,#1) -- cycle;
  \draw[hwall, fill opacity=\haahpop]
    (0.5,0.5,0.5) -- (#1,0.5,0.5) -- (#1,#1,0.5) -- (0.5,#1,0.5) -- cycle;}

\newcommand{\haahaxes}[1]{%
  \draw[haxis] (#1,#1,#1) -- (\haahA,#1,#1)
    node[haxis label, anchor=north east] {$x$};
  \draw[haxis] (#1,#1,#1) -- (#1,\haahA,#1)
    node[haxis label, anchor=north west] {$y$};
  \draw[haxis] (#1,#1,#1) -- (#1,#1,\haahA)
    node[haxis label, anchor=south] {$z$};}

\newcommand{\haahconfig}[4][]{%
  \begingroup
  \haahdefaults\pgfkeys{/haah/.cd,#1}%
  \pgfmathsetmacro{\haahE}{#2-0.1}%   the planes reach close to ...
  \pgfmathsetmacro{\haahA}{#2+0.5}%   ... and the axes a bit beyond them
  \begin{tikzpicture}[tdplot_main_coords, scale=\haahscale,
      baseline=(current bounding box.center)]
    \ifhaahwalls
      \foreach \cx/\cy/\cz in {#4} {\haahcubepart{#3}{\cx}{\cy}{\cz}{0}}
      \haahwalls{\haahE}%
      \haahaxes{0.5}%
    \else
      \haahaxes{0}%
    \fi
    \foreach \cx/\cy/\cz in {#4} {\haahcubepart{#3}{\cx}{\cy}{\cz}{1}}
  \end{tikzpicture}%
  \endgroup}

\newcommand{\haahconfigs}[1][]{%
  \haahconfig[#1]{5}{teal!30!white}{%
    0/1/2, 0/2/1, 1/0/2, 1/2/0, 2/0/1, 2/1/0}\quad
  \haahconfig[#1]{5}{violet!30!white}{%
    0/0/3, 0/2/0, 0/2/1, 0/3/0, 1/0/0, 1/0/2, 1/1/0, 1/1/1, 2/0/0, 2/1/0, 3/0/0}\quad
  \haahconfig[#1]{5}{yellow!30!white}{%
    0/0/1, 0/0/3, 0/1/0, 0/1/1, 0/2/1, 0/3/0, 1/0/0, 1/0/1, 1/0/2, 1/1/0, 1/1/1, 2/1/0, 3/0/0}}

\tikzset{hcut/.style={line width=0.4pt, line cap=round, draw=haahface!70!black}}
\newlength{\pgridcell}
\makeatletter
\tikzset{
  p token/.style={font=\footnotesize, inner sep=0pt},
  p rect label/.style={p token, anchor=west, draw=none, fill=none,
                       inner sep=1pt, xshift=1.5pt},
  p outer ring/.style={line width=0.5pt},
  dot radius/.store in=\p@dotr,
  ring radius/.store in=\p@ringr,
  outer radius/.store in=\p@outr,
}
\def\p@dotr{0.09}
\def\p@ringr{0.24}
\def\p@outr{0.38}
\newcommand{\p@token}[2]{%
  \ifthenelse{\equal{#2}{-}}{}{%
  \ifthenelse{\equal{#2}{.}}{\fill (#1) circle (\p@dotr);}{%
  \ifthenelse{\equal{#2}{o}}{\draw[line width=0.5pt] (#1) circle (\p@ringr);}{%
  \ifthenelse{\equal{#2}{.o}}{\draw[line width=0.5pt] (#1) circle (\p@ringr);%
                             \fill (#1) circle (\p@dotr);}{%
  \ifthenelse{\equal{#2}{.oo}}{\draw[p outer ring] (#1) circle (\p@outr);%
                             \draw[line width=0.5pt] (#1) circle (\p@ringr);%
                             \fill (#1) circle (\p@dotr);}{%
  \ifthenelse{\equal{#2}{oo}}{\draw[p outer ring] (#1) circle (\p@outr);}{%
  \node[p token] at (#1) {$#2$};}}}}}}}

\newlength{\p@axis}
\def\p@setaxis{\settoheight{\p@axis}{\ensuremath{\vcenter{\hbox{}}}}}
\tikzset{
  p grid baseline/.style={baseline={([yshift=-\p@axis]p@gc)}},
  axes/.store in=\p@axesspec,
  axes gap/.store in=\p@axgap,      % grid corner -> arrow origin, in cells
  axes length/.store in=\p@axlenspec,% arrow length in cells, or auto
  axes width/.code={\tikzset{p axes/.append style={line width=#1}}},
  p axes defaults/.style={axes gap=0.35, axes length=auto},
  p axes/.style={-stealth, line width=0.6pt},
  p axis label/.style={p token, inner sep=1.5pt},
}
\newif\ifp@grid
\def\p@parsegrid#1/#2\@nil{\gdef\p@gridn{#1}\gdef\p@gridm{#2}\global\p@gridtrue}
\tikzset{
  grid/.code={\p@parsegrid#1\@nil\tikzset{p grid baseline}},
  p grid lines/.style={line width=0.4pt, gray!65},
  p boxes/.style={line width=0.4pt},
  boxes width/.code={\tikzset{p boxes/.append style={line width=#1}}},
  grid width/.code={\tikzset{p grid lines/.append style={line width=#1}}},
}
\def\p@drawgrid{%
  \ifp@grid
    \foreach \px in {1,...,\p@gridn} {%
      \foreach \py in {1,...,\p@gridm} {%
        \draw[p grid lines] (\px-1,-\py) rectangle (\px,-\py+1);%
      }}%
    \coordinate (p@gc) at ({\p@gridn/2},{-\p@gridm/2});%
  \fi}

\def\p@axsplit#1/#2\@nil{\def\p@axh{#1}\def\p@axv{#2}}
\def\p@axinfo#1#2\@nil{%
  \if-#1\def\p@axd{-1}\def\p@axl{#2}\else\def\p@axd{1}\def\p@axl{#1#2}\fi}
\def\p@axes@draw#1#2#3{% #1 = <h>/<v>, #2 = columns, #3 = rows
  \def\p@axtmp{#1}\expandafter\p@axsplit\p@axtmp\@nil
  \expandafter\p@axinfo\p@axh\@nil \let\p@axhd\p@axd \let\p@axhl\p@axl
  \expandafter\p@axinfo\p@axv\@nil \let\p@axvd\p@axd \let\p@axvl\p@axl
  \ifthenelse{\equal{\p@axlenspec}{auto}}%
    {\pgfmathsetmacro{\p@axlh}{#2>3 ? #2/2 : 0.8}%
     \pgfmathsetmacro{\p@axlv}{#3>3 ? #3/2 : 0.8}}%
    {\pgfmathsetmacro{\p@axlh}{\p@axlenspec}\let\p@axlv\p@axlh}%
  \ifnum\p@axhd>0 \def\p@axha{west}\def\p@axox{-\p@axgap}%
  \else           \def\p@axha{east}\def\p@axox{#2+\p@axgap}\fi
  \ifnum\p@axvd>0 \def\p@axva{south}\def\p@axoy{-#3-\p@axgap}%
  \else           \def\p@axva{north}\def\p@axoy{\p@axgap}\fi
  \begin{scope}[shift={({\p@axox},{\p@axoy})}]
    \draw[p axes] (0,0) -- ({\p@axhd*\p@axlh},0)
      node[p axis label, anchor=\p@axha] {$\p@axhl$};
    \draw[p axes] (0,0) -- (0,{\p@axvd*\p@axlv})
      node[p axis label, anchor=\p@axva] {$\p@axvl$};
  \end{scope}}
\def\p@axes@maybe#1#2#3{%
  \ifthenelse{\equal{\p@axesspec}{none}\OR\equal{\p@axesspec}{false}}{}{%
    \ifthenelse{\equal{\p@axesspec}{auto}\OR\equal{\p@axesspec}{true}}%
      {\def\p@axspec{#1}}{\let\p@axspec\p@axesspec}%
    \expandafter\p@axes@draw\expandafter{\p@axspec}{#2}{#3}}}

\newcommand{\pgrid}[2][]{%
  \p@setaxis
  \begin{tikzpicture}[p grid baseline, x=\pgridcell, y=\pgridcell, #1]
    \foreach \prow [count=\py] in {#2} {%
      \foreach \pc [count=\px] in \prow {%
        \draw[p boxes] (\px-1,-\py) rectangle (\px,-\py+1);
        \p@token{\px-0.5,-\py+0.5}{\pc}%
        \xdef\pgrid@n{\px}\xdef\pgrid@m{\py}%
      }}%
    \coordinate (p@gc) at ({\pgrid@n/2},{-\pgrid@m/2});
  \end{tikzpicture}}

\def\p@gridlabel#1/#2/#3/#4\@nil{%
  \if\relax\detokenize{#2}\relax
    \node[pgrid mark] at ({\pgrid@n/2},{-\pgrid@m/2}) {$#1$};%
  \else
    \node[pgrid mark] at (#1,-#2) {$#3$};%
  \fi}
\newcommand{\pgridv}[3][]{%
  \p@setaxis
  \begin{tikzpicture}[p grid baseline, p axes defaults,
                      x=\pgridcell, y=\pgridcell, axes=none, #1]
    \foreach \prow [count=\py] in {#3} {%
      \foreach \pc [count=\px] in \prow {%
        \draw[p boxes] (\px-1,-\py) rectangle (\px,-\py+1);
        \p@token{\px-0.5,-\py+0.5}{\pc}%
        \xdef\pgrid@n{\px}\xdef\pgrid@m{\py}%
      }}%
    \coordinate (p@gc) at ({\pgrid@n/2},{-\pgrid@m/2});
    \if\relax\detokenize{#2}\relax\else
      \foreach \pl in {#2} {\expandafter\p@gridlabel\pl///\@nil}%
    \fi
    \p@axes@maybe{x/y}{\pgrid@n}{\pgrid@m}%
  \end{tikzpicture}}

\newcommand{\pcells}[2][0,0]{%
  \begin{scope}[shift={(#1)}]
    \foreach \prow [count=\py] in {#2} {%
      \foreach \pc [count=\px] in \prow {%
        \p@token{\px-0.5,0.5-\py}{\pc}%
      }}%
  \end{scope}}

\newcommand{\pboxes}[2][0,0]{%
  \begin{scope}[shift={(#1)}]
    \foreach \prow [count=\py] in {#2} {%
      \foreach \pc [count=\px] in \prow {%
        \ifthenelse{\equal{\pc}{-}}{}{%   - leaves the cell out entirely
          \draw[p boxes] (\px-1,1-\py) rectangle (\px,-\py);}%
        \p@token{\px-0.5,0.5-\py}{\pc}%
      }}%
  \end{scope}}
\newcommand{\pboxesv}[3][0,0]{%
  \begin{scope}[shift={(#1)}]
    \foreach \prow [count=\py] in {#3} {%
      \foreach \pc [count=\px] in \prow {%
        \ifthenelse{\equal{\pc}{-}}{}{%   - leaves the cell out entirely
          \draw[p boxes] (\px-1,1-\py) rectangle (\px,-\py);}%
        \p@token{\px-0.5,0.5-\py}{\pc}%
        \xdef\pgrid@n{\px}\xdef\pgrid@m{\py}%
      }}%
    \if\relax\detokenize{#2}\relax\else
      \foreach \pl in {#2} {\expandafter\p@gridlabel\pl///\@nil}%
    \fi
  \end{scope}}

\tikzset{
  p check line/.style={line width=0.35pt, black!55},
  p check label/.style={p token, font=\scriptsize, outer sep=1.2pt},
  check size/.store in=\p@chksize,
  check labels/.store in=\p@chklab,
}
\def\p@chksize{0.5}
\def\p@chklab{IX/XI/XI/II}
\def\p@chksplit#1/#2/#3/#4\@nil{%
  \def\p@chkbl{#1}\def\p@chkbr{#2}\def\p@chktl{#3}\def\p@chktr{#4}}
\newcommand{\pchecks}[3][0,0]{%
  \begin{scope}[shift={(#1)}]
    \expandafter\p@chksplit\p@chklab\@nil
    \pgfmathsetmacro{\p@chkh}{\p@chksize/2}%
    \foreach \pchkx in {1,...,#2} {%
      \foreach \pchky in {1,...,#3} {%
        \pgfmathsetmacro{\pchkcx}{\pchkx-0.5}%
        \pgfmathsetmacro{\pchkcy}{0.5-\pchky}%
        \node[p check label] (pchk-\pchkx-\pchky-bl)
          at ({\pchkcx-\p@chkh},{\pchkcy-\p@chkh}) {$\p@chkbl$};
        \node[p check label] (pchk-\pchkx-\pchky-br)
          at ({\pchkcx+\p@chkh},{\pchkcy-\p@chkh}) {$\p@chkbr$};
        \node[p check label] (pchk-\pchkx-\pchky-tl)
          at ({\pchkcx-\p@chkh},{\pchkcy+\p@chkh}) {$\p@chktl$};
        \node[p check label] (pchk-\pchkx-\pchky-tr)
          at ({\pchkcx+\p@chkh},{\pchkcy+\p@chkh}) {$\p@chktr$};
        \draw[p check line]
          (pchk-\pchkx-\pchky-bl) -- (pchk-\pchkx-\pchky-br)
          (pchk-\pchkx-\pchky-br) -- (pchk-\pchkx-\pchky-tr)
          (pchk-\pchkx-\pchky-tr) -- (pchk-\pchkx-\pchky-tl)
          (pchk-\pchkx-\pchky-tl) -- (pchk-\pchkx-\pchky-bl);
      }}%
  \end{scope}}

\newcommand{\pcellsv}[3][0,0]{%
  \begin{scope}[shift={(#1)}]
    \foreach \prow [count=\py] in {#3} {%
      \foreach \pc [count=\px] in \prow {%
        \p@token{\px-0.5,0.5-\py}{\pc}%
        \xdef\pgrid@n{\px}\xdef\pgrid@m{\py}%
      }}%
    \if\relax\detokenize{#2}\relax\else
      \foreach \pl in {#2} {\expandafter\p@gridlabel\pl///\@nil}%
    \fi
  \end{scope}}
\makeatother

\makeatletter
\newenvironment{pconfig}[1][]{%
  \p@setaxis \global\p@gridfalse
  \begin{tikzpicture}[x=\pgridcell, y=\pgridcell, axes=none, p axes defaults,
    baseline={([yshift=-2.5pt]current bounding box.center)}, #1]
  \p@drawgrid}%
  {\ifp@grid \p@axes@maybe{x/y}{\p@gridn}{\p@gridm}\fi
  \end{tikzpicture}}
\makeatother

\tikzset{
  pgrid mark/.style={regular polygon, regular polygon sides=5, draw=black, fill=white, line width=0.4pt,
                     inner sep=0.8pt, minimum size=3.4mm, font=\scriptsize},
  boundary/.style={line width=0.9pt},
  cut/.style={line width=0.9pt, orange!90!black},
  guide/.style={densely dashed, red!80!black, thin},
}

\makeatletter
\def\p@rectlabel#1#2#3{% #1 = tikz options, (#2,#3) = lower right corner
  \begin{scope}[#1]\node[p rect label] at ({#2},{#3}) {$\p@rectlabeltext$};\end{scope}}

\NewDocumentCommand{\prectA}{O{} m m O{}}{%
  \draw[boundary,#1] (1,0) rectangle ({#2+1},{-(#3)});%
  \if\relax\detokenize{#4}\relax\else
    \def\p@rectlabeltext{#4}\p@rectlabel{#1}{#2+1}{-(#3)}%
  \fi}

\NewDocumentCommand{\prectB}{O{} m m O{}}{%
  \draw[boundary,#1] (0,0) rectangle ({#2+1},{-(#3)-1});%
  \if\relax\detokenize{#4}\relax\else
    \def\p@rectlabeltext{#4}\p@rectlabel{#1}{#2+1}{-(#3)-1}%
  \fi}

\NewDocumentCommand{\prectC}{O{} m m O{}}{%
  \draw[boundary,#1] (0,0) -- ({#2+1},0) -- ({#2+1},{-(#3)-1})
        -- (1,{-(#3)-1}) -- (1,{-(#3)}) -- (0,{-(#3)}) -- cycle;%
  \if\relax\detokenize{#4}\relax\else
    \def\p@rectlabeltext{#4}\p@rectlabel{#1}{#2+1}{-(#3)-1}%
  \fi}
\makeatother

\makeatletter
\def\p@stairLU#1{foreach \ptristep in {1,...,#1} { -- ++(-1,0) -- ++(0,1) }}
\NewDocumentCommand{\ptri}{O{} m O{}}{%
  \draw[boundary,#1] (0,0) -- ({#2},0) -- ({#2},{-(#2)})
       \p@stairLU{#2} -- cycle;
  \if\relax\detokenize{#3}\relax\else
    \begin{scope}[#1]
      \gdef\pgrid@n{#2}\gdef\pgrid@m{#2}%
      \foreach \pl in {#3} {\expandafter\p@gridlabel\pl///\@nil}%
    \end{scope}
  \fi}
\makeatother

\makeatletter
\def\p@diagmode{rainbow}
\def\p@diagsat{0.35}
\def\p@diaglist{red,green!60!black,blue}
\def\p@diagmix{30}
\tikzset{
  diag colors/.store in=\p@diagmode,
  diag saturation/.store in=\p@diagsat,
  diag cycle/.store in=\p@diaglist,
  diag mix/.store in=\p@diagmix,
}
\newcommand{\pdiag}[3][]{%
  \begin{scope}[on background layer]
  \begin{scope}[#1]
    \pgfmathtruncatemacro{\pdnd}{#2+#3-1}%
    \xdef\pdnc{1}%
    \foreach \pdc [count=\pdci] in \p@diaglist {\xdef\pdnc{\pdci}}%
    \foreach \pdi in {1,...,#2} {%
      \foreach \pdj in {1,...,#3} {%
        \pgfmathtruncatemacro{\pdd}{\pdi+\pdj-1}%
        \ifthenelse{\equal{\p@diagmode}{cycle}}{%
          \pgfmathtruncatemacro{\pdk}{mod(\pdd-1,\pdnc)+1}%
          \gdef\pdfillcol{}%
          \foreach \pdc [count=\pdci] in \p@diaglist {%
            \ifnum\pdci=\pdk\relax \xdef\pdfillcol{\pdc!\p@diagmix}\fi}%
          \fill[\pdfillcol] (\pdi-1,-\pdj) rectangle (\pdi,1-\pdj);%
        }{%
          \pgfmathsetmacro{\pdh}{\pdnd>1 ? 0.85*(\pdd-1)/(\pdnd-1) : 0}%
          \definecolor{pdiagcol}{hsb}{\pdh,\p@diagsat,1}%
          \fill[pdiagcol] (\pdi-1,-\pdj) rectangle (\pdi,1-\pdj);%
        }%
      }}%
  \end{scope}
  \end{scope}}
\makeatother

\newcommand{\pshade}[2][gray!25]{%
  \begin{scope}[on background layer]\fill[#1] #2;\end{scope}}

\newlength{\phaahcell}
\newlength{\phaahinset}          % grid corner -> label centre, in x and y
\newlength{\phaahnamedinset}     % the same when the grid has a named qubit
\newlength{\phaahshadeinset}     % gap between the shading and the grid lines
\makeatletter
\newif\ifphaah@named             % does the grid have a named qubit?
\tikzset{
  haah grid/.style={line width=0.8pt, gray!70},
  haah check/.style={fill=gray!18},
  haah label/.style={font=\footnotesize, inner sep=0.4pt, outer sep=1pt},
  haah face/.style={line width=0.35pt, black!55},
  haah mark/.style={draw=black, fill=white, line width=0.5pt, radius=0.06},
  haah named mark/.style={circle, draw=black, fill=white, line width=0.5pt,
                          inner sep=0.3pt, minimum size=4.4mm, font=\footnotesize},
  plane/.code={\def\phaah@plane{#1}},
}

\def\phaah@regmark#1/#2/#3/#4\@nil{%
  \if\relax\detokenize{#3}\relax\else\global\phaah@namedtrue\fi}
\def\phaah@drawmark#1/#2/#3/#4\@nil{%
  \if\relax\detokenize{#3}\relax
    \draw[haah mark] (#1,-#2) circle;%
  \else
    \node[haah named mark] at (#1,-#2) {$#3$};%
  \fi}

\def\phaah@put#1#2#3#4#5#6{%
  \node[haah label, xshift=\phaah@d,  yshift=\phaah@d]  (phaah-bl) at (#1-1,-#2)  {#3};
  \node[haah label, xshift=-\phaah@d, yshift=\phaah@d]  (phaah-br) at (#1,-#2)    {#4};
  \node[haah label, xshift=\phaah@d,  yshift=-\phaah@d] (phaah-tl) at (#1-1,1-#2) {#5};
  \node[haah label, xshift=-\phaah@d, yshift=-\phaah@d] (phaah-tr) at (#1,1-#2)   {#6};
  \draw[haah face] (phaah-bl) -- (phaah-br)  (phaah-br) -- (phaah-tr)
                   (phaah-tr) -- (phaah-tl)  (phaah-tl) -- (phaah-bl);}

\def\phaah@face#1#2#3{%
  \ifthenelse{\equal{\phaah@plane}{-zy}}{%
    \ifnum#1=0 \phaah@put{#2}{#3}{$IX$}{$XX$}{$XI$}{$IX$}%
    \else      \phaah@put{#2}{#3}{$XI$}{$IX$}{$II$}{$XI$}\fi}{%
    \ifnum#1=0 \phaah@put{#2}{#3}{$XX$}{$IX$}{$IX$}{$XI$}%
    \else      \phaah@put{#2}{#3}{$IX$}{$XI$}{$XI$}{$II$}\fi}}

\NewDocumentCommand{\phaah}{O{} m m m O{}}{%
  \p@setaxis
  \begin{tikzpicture}[p grid baseline, p axes defaults, axes length=0.8,
                      x=\phaahcell, y=\phaahcell, plane=xy, axes=auto, #1]
    \global\phaah@namedfalse
    \foreach \f [count=\j] in {#3} {\xdef\phaah@m{\j}}%
    \foreach \f [count=\i] in {#4} {\xdef\phaah@n{\i}}%
    \foreach \rf [count=\j] in {#3} {%
      \foreach \cf [count=\i] in {#4} {%
        \ifnum\numexpr\rf*\cf\relax>0
          \fill[haah check] ([xshift=\phaahshadeinset, yshift=\phaahshadeinset]\i-1,-\j)
                  rectangle ([xshift=-\phaahshadeinset, yshift=-\phaahshadeinset]\i,1-\j);%
        \fi}}%
    \draw[haah grid, step=1] (0,-\phaah@m) grid (\phaah@n,0);
    \coordinate (p@gc) at ({\phaah@n/2},{-\phaah@m/2});
    \if\relax\detokenize{#5}\relax\else
      \foreach \m in {#5} {\expandafter\phaah@regmark\m///\@nil}%
    \fi
    \ifphaah@named \edef\phaah@d{\the\phaahnamedinset}%
    \else          \edef\phaah@d{\the\phaahinset}\fi
    \foreach \rf [count=\j] in {#3} {%
      \foreach \cf [count=\i] in {#4} {%
        \ifnum\numexpr\rf*\cf\relax>0
          \phaah@face{#2}{\i}{\j}%
        \fi}}%
    \if\relax\detokenize{#5}\relax\else
      \foreach \m in {#5} {\expandafter\phaah@drawmark\m///\@nil}%
    \fi
    \ifthenelse{\equal{\phaah@plane}{-zy}}%
      {\p@axes@maybe{-z/y}{\phaah@n}{\phaah@m}}%
      {\p@axes@maybe{x/y}{\phaah@n}{\phaah@m}}%
  \end{tikzpicture}}
\makeatother
\tikzset{p token/.style={font=\small, inner sep=0pt}}  % the style inside cells and for coordinate axis labels

\newtheorem{theorem}{Theorem}[section]
\newtheorem{lemma}[theorem]{Lemma}
\newtheorem{proposition}[theorem]{Proposition}

\title{Haah's 3D cubic code thermalizes rapidly}
 \author[1,2,3,*]{Sebastian Stengele}
 \author[1,3,*]{Libor Caha}
 \author[4,5]{\'Angela Capel}
 \author[1,2,3]{Simone Warzel}

 \affil[1]{\small Department of Mathematics, Technical University of Munich, 85748 Garching, Germany} 
\affil[2]{\small Department of Physics, Technical University of Munich, 85748 Garching, Germany} 
 \affil[3]{\small Munich Center for Quantum Science and Technology,
 80799 M\"{u}nchen, Germany}
 \affil[4]{\small Department of Applied Mathematics and Theoretical Physics, University of Cambridge, Wilberforce Road, Cambridge, CB3 0WA, United Kingdom}
 \affil[5]{\small Fachbereich Mathematik, Universität Tübingen, 72076 Tübingen, Germany\vspace{1.5em}} 
 \affil[*]{\small Contributed equally, ordered randomly}
\date{\small September 30, 2026}
\begin{document}

\maketitle
\begin{abstract}
Haah's $3D$ cubic code was introduced as a candidate for a self-correcting quantum memory in three dimensions, using the absence of string logical operators to obstruct thermal errors.
It has been a long-standing open question whether 
it is a self-correcting quantum memory.
Even though the memory time was explored numerically, no analytical upper bound was known.
Here, we settle this question rigorously and show that the cubic code
thermalizes rapidly at every positive temperature. 
Our proof relies on the framework of 
[Commun. Math. Phys. 407, 195 (2026)] to turn a decay-of-correlations condition into a modified logarithmic Sobolev inequality for the Davies generator.
To the best of our knowledge, this is the first analytic proof of rapid mixing for all positive temperatures for any fracton code.
\end{abstract}

\section{Introduction}

A self-correcting quantum memory would preserve quantum information stored in a Hamiltonian's ground state against thermal noise without active error correction. The 4D toric code shows that such passive protection is theoretically possible at sufficiently low temperatures \cite{AlickiEtAl2010}. 
Haah's cubic code, introduced as Code 1 in~\cite{Haah_2011}, provided a striking candidate in 3D: a translation-invariant stabilizer Hamiltonian with no string-like logical operators, which are deemed responsible for low-energy barriers and lack of self-correction~\cite{Bravyi_2009, YOSHIDA20112566}. The unusual error-geometry of the cubic code eliminates the constant-energy mechanism by
which point defects propagate along strings in familiar topological codes. By establishing a logarithmically growing energy barrier, Bravyi and Haah raised the prospect of a memory whose lifetime increases with the system's size \cite{Bravyi_2011}.  This prospect was solidified in a restricted regime by their result on partial self-correction~\cite{Bravyi_2013}. At sufficiently low temperature and for suitable system sizes with bounded ground-state degeneracy, their recovery procedure yields a lower bound on the memory time which is polynomial in the linear length $ L_{\Lambda} $ with an exponent increasing with inverse temperature $ \beta $. Their bound applies up to an $L_{\Lambda}$-cutoff, which is exponential in $ \beta $. Their simulations also exhibit an optimal $ L_{\Lambda} $ beyond which the measured memory time decreases. These results establish a substantial protection over a growing range of system sizes as the temperature is lowered, but the lower bound does not determine the asymptotic lifetime at fixed positive temperature. 

The observed decrease in memory time beyond an optimal system size is suggestive, but finite-size data alone do not settle the asymptotic question. For example, the $\mathbb Z_2$ gauge theory studied in~\cite{Stahl_2026} exhibits pronounced finite-size lifetime peaks even outside its stable classical-memory phase~\cite[Fig.~10]{Stahl_2026}. This illustrates why a quantitative bound is needed to resolve the cubic code’s behaviour at fixed positive temperature.

For the cubic code, Siva and Yoshida~\cite{SivaYoshida2017} established thermal triviality in a Gibbs-state preparation sense. They also developed an ergodic decomposition of Davies generators, reducing dynamical estimates to classical fractal spin models~\cite[Appendix~A]{SivaYoshida2017}. These results provide important equilibrium and dynamical information; the quantitative result proved here is a size-independent MLSI and logarithmic worst-case mixing at every fixed positive temperature. General stabilizer bounds already give analytical thermalization estimates, including bounds linear in system size above a sufficiently high temperature~\cite{TemmeKastoryano2015}. Our contribution is therefore the all-positive-temperature logarithmic bound for the specified cubic-code dynamics.

The cubic code's logarithmic energy barrier does not by itself give this conclusion through the bounds in~\cite{Temme17}. 
The argument there depends on the \emph{generalised} energy barrier, a quantity that characterises the worst-case energy barrier to constructing arbitrary Pauli errors through optimal sequences of single-qubit operations (while counting  only intermediate excitations absent from the final error configuration). The standard energy barrier provides only a lower bound on this quantity. This distinction is particularly striking in the $3D$ toric code, whose standard energy barrier is $O(1)$, whereas its generalised energy barrier is $\Omega(L_{\Lambda})$.  To our knowledge, no non-trivial upper bound on the generalised energy barrier of the cubic code has been established. Thus, an estimate of the proposed form in~\cite{Temme17} would require control
of the generalised barrier, which the cubic code’s known logarithmic logical barrier does not provide.

In this work, we establish rapid thermalization directly. We consider the cubic code Hamiltonian specified below and its local Davies dynamics with single-qubit Pauli couplings. The bath rates satisfy detailed balance, and, at each fixed positive temperature, are bounded above and away from zero independently of the volume. Our main result can then be informally stated as follows.

\begin{theorem}\label{thm:rm}
    The Davies generator of Haah’s cubic code mixes rapidly at every positive inverse temperature $\beta <\infty$.
\end{theorem}

More precisely, at every fixed positive temperature, we prove a modified logarithmic Sobolev inequality (MLSI) with a constant independent of the system size for a cubic code of linear size $ L_{\Lambda} $ (and hence $ \Theta(L_{\Lambda}^3) $-many physical qubits). As a consequence, the worst-case trace-distance mixing time  is shown to be bounded by 
\begin{equation}
 t_{\mathrm{mix}}(\varepsilon) \leq C _\beta \left(\log L_{\Lambda}  + \log \varepsilon^{-1}\right) ,
\end{equation}
where $ C_\beta $ is independent of $ L_{\Lambda} $ and $ \varepsilon$. This bounds convergence from arbitrary initial states to the Gibbs state. 

The operational consequence is an $ \mathcal{O}(\log L_{\Lambda}) $ upper bound on passive-memory time irrespective of the final decoder. Indeed, mixing makes initially distinct encoded states nearly indistinguishable, and no recovery channel can increase their trace distance. The result hence excludes asymptotic polynomial growth of the memory lifetime in $ L_{\Lambda} $ at fixed positive temperature. 
It remains compatible with the partial self-correction \cite{Bravyi_2013} below its temperature-dependent cut-off since the above constant $ C_\beta $ may grow rapidly as $ \beta \to \infty $.

Theorem~\ref{thm:rm} supports the conjectured absence of self-correcting quantum memories among the 3D translation-invariant Pauli stabilizer
Hamiltonians with local topological order \cite{Haah2013}. Beyond translation invariance, CSS codes in 3D with a memory-time lower bound under local thermal dynamics that is exponential in some power of the number of qubits have recently been proposed in~\cite{BalasubramanianDavydovaLin2026}. Although these codes are geometrically local, their lack of translation invariance might render them less practical. \\

Establishing MLSIs strong enough to imply rapid mixing is a substantial challenge for interacting quantum systems. Even for commuting Hamiltonians, the probabilistic conditioning arguments used in classical proofs have no straightforward counterpart for the local quantum dynamics~\cite{bardet2020approximate,stengele_ModifiedlogarithmicSobolev_2025}. Results of this kind remain comparatively scarce: they include commuting spin chains~\cite{bardet2021MLSIDavies1D,Kochanowski.2024} and certain commuting models at high temperature~\cite{capel2020MLSI,Kochanowski.2024,CapelGondolfKochanowskiRouze-RapidMixingCommuting-2024}, while examples valid at every positive temperature in higher dimensions include the two-dimensional toric code and Abelian quantum double models~\cite{stengele_ModifiedlogarithmicSobolev_2025,stengele_AbelianQuantumDouble_2026}. For classical spin systems, there is a route from spatial mixing of the Gibbs measure to logarithmic Sobolev inequalities for Glauber dynamics~\cite{stroock_LogarithmicSobolevInequality_199,martinelli_ApproachEquilibriumGlauber_1994b}. The approach developed in~\cite{stengele_ModifiedlogarithmicSobolev_2025} adapts this strategy to the Davies dynamics of CSS codes: a Dobrushin-Shlosman (DS) condition yields approximate tensorization of quantum relative entropy, and a multiscale argument turns local MLSI bounds into a bound uniform in the system size. We follow this strategy here, establishing the required DS-condition for the two check types of Haah's code through geometric estimates on check configurations. 

For the DS-condition, we use statistical-mechanics methods analogous to those in~\cite{Toninelli2017}, which proved spatial mixing in the classical square and triangular plaquette model. The key geometric idea in our paper is to identify cancellations in the high-temperature expansion of the partially traced Gibbs state entering the DS-condition. The objects that support this expansion are the kernels of the check maps, cf.~Fig.~\ref{fig:gammaillustration} for examples. We  prove bounds on their number and the dimension of their support. 

Part of the geometric difficulty in such estimates is that Haah's cubic code is a fracton model. It is only one example within the broader family of fracton codes and
related exactly solvable models, which includes the Chamon and X-cube
models as well as families of fractal spin
liquids~\cite{Chamon2005,Haah_2011,Yoshida2013Fractal,VijayHaahFu2015,Brown_2016,VijayHaahFu2016,NandkishoreHermele2019,PretkoChenYou2020}.

\begin{figure}[h]
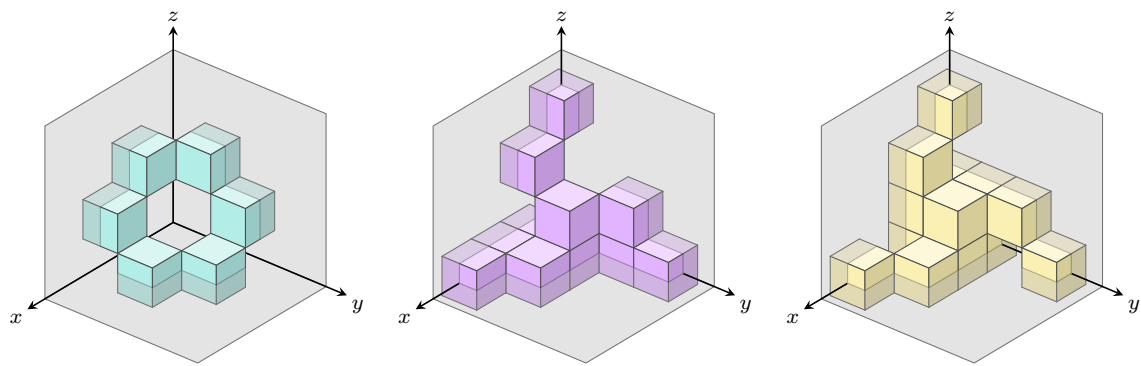

    \centering
       \haahconfigs[cube opacity=1, scale=0.6]
    \caption{Three examples of terms in the high temperature expansion of the $X$-part. These are sets of checks that multiply to the identity inside the grey box.}
    \label{fig:gammaillustration}
\end{figure}

\section{Set-up and main result}

\subsection{Haah's Hamiltonian} 

We consider Haah's cubic code~\cite{Haah_2011} with linear size $L_{\Lambda}\geq3$ and periodic boundary conditions.\footnote{For simplicity we restrict ourselves to the periodic boundary conditions, but we expect the result and proof strategy to carry over to open boundary conditions as well.} The underlying lattice is the discrete torus $\Lambda=(\mathbb Z/L_{\Lambda}\mathbb Z)^3$ with coordinates understood modulo $L_{\Lambda}$.
Each vertex carries two qubits, labelled by $1$ and $2$, giving $n=2|\Lambda|=2L_{\Lambda}^3$ physical qubits. Each elementary cube supports an \emph{$X$-check} $P_s^X$ and a \emph{$Z$-check} $P_s^Z$, with the following Pauli operators:

\begin{equation}\label{eq:haah-checks}
\vspace{0.7em}
    \hspace*{-2.5cm}% Haah's code 1: star and plaquette on identically oriented unit cubes.
% Corner labels are ordered by qubit 1, then qubit 2.
% Projection: x points right, y up, and z towards the lower left.
\begin{tikzpicture}[
    x=1cm, y=1cm,
    baseline={([yshift=-0.5ex]current bounding box.center)},
    hc edge/.style={draw=black!55, line width=0.65pt},
    hc hidden/.style={draw=black!40, line width=0.5pt,
        dash pattern=on 2pt off 2pt},
    hc pauli/.style={font=\small, fill=white, rounded corners=1pt,
        inner xsep=2.5pt, inner ysep=1.5pt, outer sep=0pt},
    hc title/.style={font=\small, inner sep=0pt},
    hc axis/.style={-{Stealth[length=1.5mm]}, draw=black, line width=1pt},
    hc axis label/.style={font=\small, text=black, inner sep=1.5pt}
]
    % Draw the same geometry for both checks before placing their labels.
    \foreach \prefix/\offset/\colour in {X/1/teal,Z/5.2/violet} {
        \begin{scope}[xshift=\offset cm,
            x={(2.2cm,0cm)}, y={(0cm,2.2cm)}, z={(-0.95cm,-0.75cm)}]
            \foreach \site/\xx/\yy/\zz in {
                000/0/0/0,100/1/0/0,010/0/1/0,001/0/0/1,
                110/1/1/0,011/0/1/1,101/1/0/1,111/1/1/1}
                \coordinate (\prefix\site) at (\xx,\yy,\zz);
            \fill[\colour!7]
                (\prefix010) -- (\prefix110) -- (\prefix111) -- (\prefix011) -- cycle;
            \fill[\colour!13]
                (\prefix100) -- (\prefix110) -- (\prefix111) -- (\prefix101) -- cycle;
            \fill[\colour!4]
                (\prefix001) -- (\prefix101) -- (\prefix111) -- (\prefix011) -- cycle;
            \draw[hc hidden]
                (\prefix000) -- (\prefix100)
                (\prefix000) -- (\prefix010)
                (\prefix000) -- (\prefix001);
            \draw[hc edge]
                (\prefix100) -- (\prefix110) -- (\prefix010) -- (\prefix011)
                -- (\prefix001) -- (\prefix101) -- (\prefix100)
                (\prefix111) -- (\prefix110)
                (\prefix111) -- (\prefix011)
                (\prefix111) -- (\prefix101);
        \end{scope}
    }

    \node[hc title, text=teal!65!black] at (2.1,2.75) {$P_s^X$ ($X$-check)};
    \node[hc title, text=violet!65!black] at (6.3,2.75) {$P_s^Z$ ($Z$-check)};

    % The Pauli assignments are unchanged from the algebraic definition.
    \foreach \site/\pauli in {000/XX,100/IX,010/IX,001/IX,110/XI,011/XI,101/XI,111/II}
        \node[hc pauli, text=teal!65!black] at (X\site) {$\pauli$};
    %\node[hc pauli, text=black!50] at (X111) {$II$};
    \foreach \site/\pauli in {000/II,100/IZ,010/IZ,001/IZ,110/ZI,011/ZI,101/ZI,111/ZZ}
        \node[hc pauli, text=violet!65!black] at (Z\site) {$\pauli$};
    %\node[hc pauli, text=black!50] at (Z000) {$II$};

    % A shared coordinate triad fixes the orientation of both cubes.
    \begin{scope}[shift={(-2,0.2)}]
        \draw[hc axis] (0,0) -- (0.6,0)
            node[hc axis label, anchor=west] {$x$};
        \draw[hc axis] (0,0) -- (0,0.6)
            node[hc axis label, anchor=south] {$y$};
        \draw[hc axis] (0,0) -- (-0.38,-0.3)
            node[hc axis label, anchor=north east] {$z$};
    \end{scope}
    \node[font=\small, text=black, inner sep=0pt] at (3.225,-1.3)
        {At each vertex: $PQ=P_1\otimes Q_2$};
\end{tikzpicture}
\vspace{0.7em}
\end{equation}
The checks in~\eqref{eq:haah-checks} have lattice $\ell^\infty$-diameter one, and each qubit belongs to four $X$-checks and four $Z$-checks.
Let $\checkset^X$ and $\checkset^Z$ be the sets of $X$-check and $Z$-check labels, respectively. All checks commute and square to the identity. With the check couplings set to one, the Hamiltonian acting on the Hilbert space $\mathcal H_\Lambda=(\mathbb C^2)^{\otimes n}$ is
\begin{equation}\label{eq:hamiltonian-gibbs}
    H_\Lambda=H_\Lambda^X+H_\Lambda^Z \ ,
    \qquad H_\Lambda^X=-\sum_{s\in\checkset^X}P_s^X \ ,
    \qquad H_\Lambda^Z=-\sum_{s\in\checkset^Z}P_s^Z \ .
\end{equation}
There are $2L_{\Lambda}^3$ checks in total, so $\|H_\Lambda\|_\infty\leq2 L_{\Lambda}^3= n$. 
For statements applying to either check type, we use $\sharp\in\{X,Z\}$, with $P_s^\sharp$ denoting the check indexed by $s\in\checkset^\sharp$. 

The Gibbs state corresponding to~\eqref{eq:hamiltonian-gibbs} is
\begin{equation}\label{eq:Gibbs}
    \rho=\frac{e^{-\beta H_\Lambda}}{Z_\Lambda} \ ,
    \qquad Z_\Lambda=\Tr e^{-\beta H_\Lambda} \ .
\end{equation}
Throughout, $0<\beta<\infty$ is fixed. Constants may depend on $\beta$ and on the bath rates, but not on $L_{\Lambda}$.

\subsection{Davies dynamics}

Our main result concerns the thermal dynamics of the Hamiltonian~\eqref{eq:hamiltonian-gibbs}.
We use the Davies dynamics generated by single-qubit $X$ and $Z$ bath couplings, with the conventions of~\cite{stengele_ModifiedlogarithmicSobolev_2025}. Given the spectral resolution $H_\Lambda=\sum_E E\Pi_E$, we define jump operators
\[
    S_{v,q,a}(\omega)=\sum_{E'-E=\omega}\Pi_E a_{v,q}\Pi_{E'} \ ,
    \qquad v\in\Lambda,\quad q\in\{1,2\},\quad a\in\{X,Z\} \ ,
\]
so that $[H_\Lambda,S_{v,q,a}(\omega)]=-\omega S_{v,q,a}(\omega)$. In the Schrödinger picture, the Davies Lindbladian takes the form
\begin{equation}\label{eq:davies-generator}
    \mathcal L(\sigma)=\sum_{v,q,a,\omega}h_{v,q,a}(\omega)
    \left(S_{v,q,a}(\omega)\sigma S_{v,q,a}(\omega)^\dagger
    -\frac12\{S_{v,q,a}(\omega)^\dagger S_{v,q,a}(\omega),\sigma\}\right).
\end{equation}
The rates satisfy detailed balance and a uniform lower bound:
\begin{equation}\label{eq:bath-rates}
    h_{v,q,a}(-\omega)=e^{-\beta\omega}h_{v,q,a}(\omega)\ ,
    \qquad
    \inf_{L,v,q,a,\omega}h_{v,q,a}(\omega)e^{-\beta\omega/2}
    \geq g_\beta>0\ ,
\end{equation}
where the infimum is over nonzero jump operators. Each coupling changes only the finitely many checks incident on its qubit, so the possible frequencies are bounded independently of $L$. Fixed, strictly positive, translation-invariant rates therefore satisfy this assumption. The two Pauli couplings on every qubit generate the full matrix algebra; the resulting dynamics has the Gibbs state~\eqref{eq:Gibbs} as its unique stationary state $\rho$.

The Davies generator splits into $X$-check and $Z$-check parts,
\[
    \mathcal L=\mathcal L^X+\mathcal L^Z \ .
\]
The $X$-check part $\mathcal L^X$ contains the $ Z$-bath couplings, which change only $X$-check energies; the $Z$-check part $\mathcal L^Z$ contains the $ X$-bath couplings, which change only $Z$-check energies. Thus, the superscript $ \sharp \in \{X,Z\} $ labels the affected checks. The two parts commute~\cite[Lemma~3.3]{stengele_ModifiedlogarithmicSobolev_2025}. For $B\subseteq\Lambda$, $\mathcal L_B^\sharp$ denotes the sum of terms with coupling sites in $B$, retaining their full support in $\Lambda$.

\subsection{Modified logarithmic Sobolev inequality and mixing time}

Let $ \mathcal D(\mathcal H_\Lambda)
    :=\{\sigma\in\mathcal B(\mathcal H_\Lambda):\sigma\geq0,\ \Tr\sigma=1\} $
denote the set of density operators on the system's Hilbert space. For fixed $ \Lambda $ and inverse temperature $\beta>0$, we abbreviate by $\mathcal T_t:=e^{t\mathcal L}$ the Davies semigroup with stationary state $\rho$, the Gibbs state~\eqref{eq:Gibbs}. We measure convergence of the evolution using the trace norm $\|A\|_1:=\Tr\sqrt{A^\dagger A}$. For $0<\varepsilon<1$, the \emph{$\varepsilon$-mixing time} is
\begin{equation}\label{eq:mixing-time-definition}
    t_{\mathrm{mix}}(\varepsilon)
    :=\inf\left\{t\geq0:
    \sup_{\sigma\in\mathcal D(\mathcal H_\Lambda)}
    \|\mathcal T_t(\sigma)-\rho\|_1\leq\varepsilon\right\} \ .
\end{equation}
Since $\mathcal T_t$ preserves $\rho$ and contracts the trace norm on differences of states, this worst-case distance to equilibrium is non-increasing. Thus the mixing time is the earliest time after which every initial state remains within distance $\varepsilon$ of equilibrium. The Davies dynamics considered here converges to $\rho$, so this time is finite. 

At fixed inverse temperature $\beta$, we call the family of dynamics \emph{rapidly mixing} if its mixing time grows at most polylogarithmically with the number $ n $ of qubits: there is a fixed $k\geq0$ such that, for every fixed $0<\varepsilon<1$,
\[
    t_{\mathrm{mix}}(\varepsilon)
    =\mathcal{O}\left((\ln n)^k\right) \ .
\]
If instead $t_{\mathrm{mix}}(\varepsilon)
    =\mathcal{O}(n^q)$, i.e. the mixing time grows at most polynomially, we call the family \emph{fast mixing}. 
Rapid mixing therefore implies fast mixing.

For states $\sigma$ and $\rho$, with $\rho$ of full rank, let
$D(\sigma\|\rho)=\Tr[\sigma(\ln\sigma-\ln\rho)]$ be the quantum relative entropy. We define a \textit{modified logarithmic Sobolev inequality} (MLSI in short) with constant $\alpha>0$  as the inequality
\begin{equation}\label{eq:MLSI}
    2\alpha D(\sigma\|\rho)
    \leq -\Tr\left[\mathcal L(\sigma)(\ln\sigma-\ln\rho)\right]
    =:EP_\Lambda(\sigma)
\end{equation}
for every full-rank state $\sigma$. The right-hand side is the \textit{entropy production}. Integrating this differential inequality gives
\begin{equation}\label{eq:entropy-decay}
    D(e^{t\mathcal L}(\sigma)\|\rho)
    \leq e^{-2\alpha t}D(\sigma\|\rho)
    \leq e^{-2\alpha t}\ln(1/\rho_{\min}) \ ,
\end{equation}
where $\rho_{\min}$ is the smallest eigenvalue of $\rho$. The estimate extends to arbitrary initial states by continuity. Since in this work $\rho$ is always the Gibbs state of a local Hamiltonian, we have $\ln(1/\rho_{\min})\leq c_\beta n$, with some $c_\beta < \infty $. Together with Pinsker's inequality, we find
\begin{equation}
    \|e^{t \mathrm{\mathcal{L}}}(\sigma) - \rho\|_1 \leq \sqrt{2 D(e^{t \mathrm{\mathcal{L}}}(\sigma) \| \rho)} \leq  \sqrt{ 2 c_\beta n }\,  e^{- \alpha t} \ .
\end{equation}
By~\eqref{eq:mixing-time-definition}, we then obtain for any $0<\varepsilon<1$,
\begin{equation}\label{eq:mixing-time}
    t_{\mathrm{mix}}(\varepsilon)
    \leq\frac{1}{2\alpha}\ln\left(\frac{2c_\beta n}{\varepsilon^2}\right)\ .
\end{equation}
Thus, an inverse-polylogarithmic lower bound on $\alpha$ suffices for rapid mixing: if $\alpha\geq a_\beta(\ln n)^{-p}$ for some $a_\beta>0$ and fixed $p\geq0$, independent of $L$, then~\eqref{eq:mixing-time} gives
\[
    t_{\mathrm{mix}}(\varepsilon)
    =\mathcal{O}\left((\ln n)^p
    \bigl(\ln n+\ln\varepsilon^{-1}\bigr)\right)\ .
\]
At fixed temperature and accuracy,  the mixing time hence remains polylogarithmic in the system size. The size-independent MLSI bound that we prove in this paper gives the stronger case $p=0$.

\subsection{DS-condition for rapid mixing}\label{s:DSgeo}

For a subset $B\subseteq\Lambda$, we define
\begin{equation}\label{eq:local-hamiltonian}
    \checkset_B^\sharp
    :=\{s\in\checkset^\sharp:\mathrm{supp}(P_s^\sharp)
    \cap(B\times\{1,2\})\neq\emptyset\}\ ,
    \qquad H_B^\sharp:=-\sum_{s\in\checkset_B^\sharp}P_s^\sharp \ .
\end{equation}
These Hamiltonians include every retained check touching $B$, with its full support in $\Lambda$. Thus, checks may cross the boundary of $B$. Following~\cite[Eq.~(1.19)]{stengele_ModifiedlogarithmicSobolev_2025},  we introduce the conditional Gibbs operators
\begin{equation}\label{eq:conditional-gibbs}
    \widehat\rho_B^\sharp
    :=e^{-\beta H_B^\sharp}\bigl(\Tr_B e^{-\beta H_B^\sharp}\bigr)^{-1}.
\end{equation}
Here $\Tr_B$ is the ordinary, unnormalized partial trace over both qubits at each vertex of $B$. We identify its output with an operator on the full system by tensoring with the identity on $B$. In particular, $\widehat\rho_B^\sharp$ is normalized conditionally on the exterior, $\Tr_B\widehat\rho_B^\sharp=I_{B^c}$; it is generally different from the reduced Gibbs state on $B$. All operators of the $X$- and $Z$-parts are diagonal in the corresponding Pauli product basis, and $\Tr_B e^{-\beta H_B^\sharp}$ is strictly positive.

Boxes are products of consecutive lattice intervals contained in $\Lambda$.
They may inherit periodic boundary conditions in one or more coordinate directions from $\Lambda$.  We write $L_{\min}^B$ and $L_{\max}^B$ to denote their shortest and longest side lengths, measured in numbers of vertices. We call $B$ \emph{fat} if $L_{\max}^B\leq10L_{\min}^B$.

We now specify the \emph{DS-geometry}. Let $U,V,W \subset \Lambda$ be pairwise disjoint and nonempty sets. We write $UV:=U\sqcup V$, $VW:=V\sqcup W$, and $UVW:=U\sqcup V\sqcup W$. The regions $UV,VW,UVW$ are fat boxes with the same intervals in two coordinate directions. Along the remaining direction $a$, either $UVW$ is a proper interval and the three sets occur in the order $U,V,W$, or $UVW$ is the full coordinate cycle and the sets occur in cyclic order $U,V_1,W,V_2$, with $V=V_1\sqcup V_2$. No fatness assumption is imposed on $V$ or its components.
Distances between points on the ambient torus are measured using the infinity norm with periodic boundary conditions. We write $d:=\dist{U,W}=\min_{u\in U,w\in W}\|u-w\|_\infty$.

\begin{figure}[h]
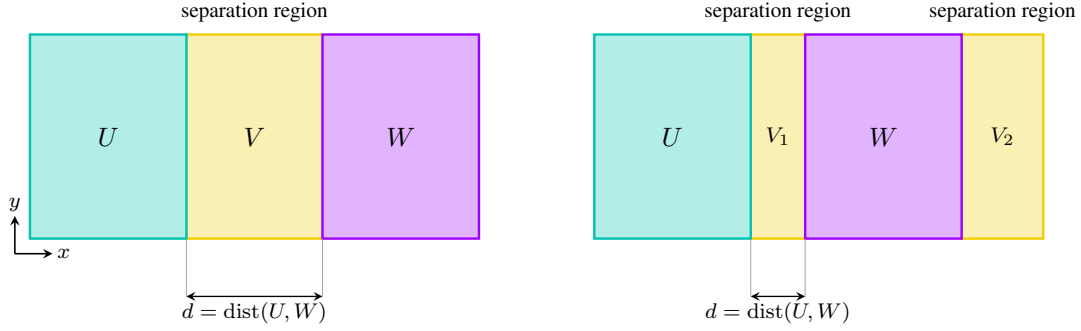

    \centering
    \uvw[measures=false, v width=2, u width=2.3, w width=2.3, v length={d=\dist{U,W}}, v note={\mbox{separation region}}, axes=x/y]
    \qquad\qquad
    \uvwper[measures=false, wrap=false, height=3, v1 width=0.8, v2 width=1.2, v1 length={d=\dist{U,W}}, v2 length={}, u width=2.3, w width=2.3, v note={\mbox{separation region}},]
    \caption{Geometry of the DS estimate with open boundaries. The box $V$ separates $U$ from $W$, and $UV$ and $VW$ overlap on $V$. The left panel illustrates open boundary conditions and the right panel illustrates periodic boundary conditions with $V=V_1\sqcup V_2$. Note that in the left panel, the length of $V$ along the $x$ direction is $d-1$.}

    \label{fig:dsuvw}
\end{figure}

The boundary correlation quantity we control is 
\begin{equation}\label{eq:DS-defect}
    \mathrm{DS}^\sharp(U,V,W)
    :=\left\|\frac{\bigl(\Tr_{UV}e^{-\beta H_{UV}^\sharp}\bigr)
    \bigl(\Tr_{VW}e^{-\beta H_{VW}^\sharp}\bigr)}
    {\bigl(\Tr_Ve^{-\beta H_V^\sharp}\bigr)
    \bigl(\Tr_{UVW}e^{-\beta H_{UVW}^\sharp}\bigr)}-I\right\|_\infty \ .
\end{equation}
Fractions denote multiplication by the inverse of the denominator. When no check touches both $U$ and $W$, cancellation of the interactions that meet only one side gives
\begin{equation}\label{eq:DS-conditional-ratio}
    \frac{\Tr_{VW}\widehat\rho_{UVW}^\sharp}
    {\Tr_V\widehat\rho_{UV}^\sharp}
    =\frac{\bigl(\Tr_{UV}e^{-\beta H_{UV}^\sharp}\bigr)
    \bigl(\Tr_{VW}e^{-\beta H_{VW}^\sharp}\bigr)}
    {\bigl(\Tr_Ve^{-\beta H_V^\sharp}\bigr)
    \bigl(\Tr_{UVW}e^{-\beta H_{UVW}^\sharp}\bigr)} \ .
\end{equation}
For a fixed $0<\delta<1$, we say that the \textit{DS-condition}~\cite{stengele_ModifiedlogarithmicSobolev_2025} holds if there are $K,\xi,L_0>0$, independent of $L$, such that
\begin{equation}\label{eq:DS-condition}
    \mathrm{DS}^\sharp(U,V,W)\leq K e^{-\xi d}
\end{equation}
for every triple in the geometry above with
$L_{\min}^{UVW}>L_0$ and $d\geq(L_{\max}^{UVW})^\delta$. The name of this condition is due to the classical Dobrushin-Shlosman condition~\cite{dobrushinShlosman1987complete}, to which it reduces in the classical case~\cite{stroock_LogarithmicSobolevInequality_199,martinelli_ApproachEquilibriumGlauber_1994a,martinelli_ApproachEquilibriumGlauber_1994b}.
Our main tool for a proof of the main result on the Haah code is the following criterion.

\begin{theorem}[Rapid mixing from the DS-condition]\label{thm:DS-to-MLSI}
    Consider Haah's cubic code on the lattice $\Lambda$ with periodic boundary conditions and the Davies generator~\eqref{eq:davies-generator} with rates satisfying~\eqref{eq:bath-rates}. Fix $0<\beta<\infty$. If both the $X$- and $Z$-checks satisfy~\eqref{eq:DS-condition} for some $0<\delta<1$, then the full generator satisfies~\eqref{eq:MLSI} with a constant $\alpha>0$ independent of $L$. In particular, $t_{\mathrm{mix}}(\varepsilon)=\mathcal{O}(\log n+\log\varepsilon^{-1})$, where $n=2|\Lambda|$.
\end{theorem}

This theorem relies on a minor modification of the criterion established in~\cite{stengele_ModifiedlogarithmicSobolev_2025}. Compared with the square-root overlap scale used in the multiscale construction of~\cite[Section~6.2]{stengele_ModifiedlogarithmicSobolev_2025}, the above formulation allows overlaps $ d $  of width $L^\delta$ and only requires fat boxes. Appendix~\ref{sec:multiscale} proves that this version is sufficient for the MLSI argument.
We will comment on the proof strategy and its modifications in more detail in the next subsection.

\subsection{On the proof of Theorem~\ref{thm:DS-to-MLSI}}\label{sus:comments}

The crucial link between the DS-condition~\eqref{eq:DS-condition} and the MLSI~\eqref{eq:MLSI} is an approximate tensorization inequality for the relative entropy. It applies separately to the two commuting parts $\mathcal L^X$ and $\mathcal L^Z$ of the Davies generator. Following~\cite[Definition~4.1]{stengele_ModifiedlogarithmicSobolev_2025}, we define the Davies conditional expectations in the Schr\"odinger picture and the corresponding conditional relative entropies by
\begin{equation}\label{def:CERE}
    \mathbb E_B^\sharp:=\lim_{t\to\infty}e^{t\mathcal L_B^\sharp} \ ,
    \qquad D_B^\sharp(\sigma):=D(\sigma\|\mathbb E_B^\sharp(\sigma)) \ .
\end{equation}
We abbreviate by $\alpha^\sharp(B)$ the largest constant for which
\[
    2\alpha^\sharp(B)D_B^\sharp(\sigma)
    \leq-\Tr\left[\mathcal L_B^\sharp(\sigma)
    \bigl(\ln\sigma-\ln\mathbb E_B^\sharp(\sigma)\bigr)\right]
\]
for all full-rank states $ \sigma \in \mathcal D(\mathcal H_\Lambda)$. The approximate tensorization theorem~\cite[Theorem~5.1]{stengele_ModifiedlogarithmicSobolev_2025} guarantees  for all $d>4$ and $\mathrm{DS}^\sharp\leq1/28$,
\begin{equation}\label{eq:DS-tensorization}
    D_{UVW}^\sharp(\sigma)
    \leq(1+28\mathrm{DS}^\sharp(U,V,W))
    \bigl(D_{UV}^\sharp(\sigma)+D_{VW}^\sharp(\sigma)\bigr)\ .
\end{equation}
The proof of Theorem~\ref{thm:DS-to-MLSI} then starts by
fixing $\beta$ and choosing common DS constants in~\eqref{eq:DS-condition} for the two parts; taking the larger prefactor $K$ and threshold $L_0$ and the smaller decay rate $\xi$ if necessary. By the DS bound~\eqref{eq:DS-condition}, $\mathrm{DS}^\sharp(U,V,W)\leq Ke^{-\xi d}$. For every admissible triple of regions $(U,V,W)$, the scale assumptions give
\[
    d\geq(L_{\max}^{UVW})^\delta>L_0^\delta.
\]
We may therefore increase $L_0$, independently of $L$, so that $d>4$ and $Ke^{-\xi d}\leq1/28$ for all such triples. Substituting the DS bound into~\eqref{eq:DS-tensorization} yields
\[
    D_{UVW}^\sharp(\sigma)
    \leq\bigl(1+28Ke^{-\xi d}\bigr)
    \bigl(D_{UV}^\sharp(\sigma)+D_{VW}^\sharp(\sigma)\bigr) \ .
\]
This is precisely the approximate tensorization estimate required by Proposition~\ref{lem:mlsirecursion}, with prefactor $28K$, for each part separately.
After increasing $L_0$ further to the fixed scale supplied by that proposition, the local estimate in~\cite[Section~6.1]{stengele_ModifiedlogarithmicSobolev_2025} and the rate assumption~\eqref{eq:bath-rates} give constants $a_\beta^\sharp>0$ such that
\[
    \alpha^\sharp(B)\geq a_\beta^\sharp
    \qquad\text{for every fat box }B\subseteq\Lambda
    \text{ with }L_{\max}^B\leq L_0  \ .
\]
These constants are independent of the ambient system size: the local estimate also controls arbitrary exterior degrees of freedom. Proposition~\ref{lem:mlsirecursion} propagates this fixed-scale bound to all larger fat boxes, with a factor $C_\sharp>0$ independent of $L$. Decreasing $C_\sharp$ to at most one includes the smaller boxes as well. Since $\Lambda$ is itself a fat box, we obtain
\[
    \alpha^\sharp(\Lambda)\geq C_\sharp a_\beta^\sharp>0
    \qquad\text{for every } \Lambda\ .
\]

The generators for the two parts commute and their sum is the full Davies generator. By~\cite[Corollary~1.6]{stengele_ModifiedlogarithmicSobolev_2025}, their MLSI bounds combine to give~\eqref{eq:MLSI} with
\[
    \alpha:=\min_{\sharp\in\{X,Z\}}
    C_\sharp a_\beta^\sharp>0 \ ,
\]
which is independent of $L$. Finally, inserting this constant into~\eqref{eq:mixing-time} gives
\[
    t_{\mathrm{mix}}(\varepsilon)
    \leq\frac{1}{2\alpha}
    \left(\ln(2c_\beta)+\ln n+2\ln\varepsilon^{-1}\right)
    =O_\beta(\log n+\log\varepsilon^{-1})\ ,
\]
as claimed. This concludes the proof of Theorem~\ref{thm:DS-to-MLSI}.

\section{Geometric analysis of Haah's code }

In view of Theorem~\ref{thm:DS-to-MLSI}, for a proof of our main result, Theorem~\ref{thm:rm}, it remains to establish the DS-condition~\eqref{eq:DS-condition} for the $X$- and $Z$-parts.
As the $X$- and $Z$-checks in~\eqref{eq:haah-checks} are related by symmetry, it suffices to carry out the analysis for $X$-checks. Throughout the subsequent geometric analysis, we therefore omit the superscript $X$ and write $\checkset$, etc.

As a preparation, we first explain the reduction of the DS-bound~\eqref{eq:DS-condition} to the geometry of check configurations and then provide the necessary bounds to establish the DS-condition in the next section.  Our strategy to establish the DS-condition is somewhat similar to the proof of strong spatial mixing in classical plaquette models~\cite{Toninelli2017}. 

\subsection{Expansion in check configurations}
 
For $S\subseteq\checkset_B$, write $P_S:=\prod_{s\in S}P_s$ and let $|S|$ count the selected checks. Central objects in our analysis are 
\begin{equation}\label{eq:boundary-configurations}
    \begin{aligned}
        \Gamma_B&:=\left\{S\subseteq\checkset_B:
        \mathrm{supp}(P_S)\cap(B\times\{1,2\})=\emptyset\right\}\ ,\\
        \Sigma_B&:=\sum_{S\in\Gamma_B}\tau^{|S|}P_S\ ,
        \qquad \tau:=\tanh\beta\ .
    \end{aligned}
\end{equation}
An element of $\Gamma_B$ is thus a subset of checks $S$, whose product does not have any support in $B$. Generically, $P_S$ will have some support outside of $B$. 
For examples of elements of $\Gamma_B$ of the cubic code, see Fig.~\ref{fig:gammaillustration}.
The support condition is linear over $\mathbb F_2$, so $\Gamma_B$ is a binary vector space. Commutativity and $(P_s)^2=I$ yield the expansion
\begin{align}\label{eq:high-temperature-expansion}
    \Tr_B e^{-\beta H_B}
    &=\cosh(\beta)^{|\checkset_B|}
    \sum_{S\subseteq\checkset_B}\tau^{|S|}\Tr_B P_S
    =2^{2|B|}\cosh(\beta)^{|\checkset_B|}\Sigma_B\ .
\end{align}
Indeed, a Pauli string has nonzero partial trace over $B$ precisely when it acts trivially there, in which case the trace over the $2|B|$ qubits contributes $2^{2|B|}$. Although usually called a high-temperature expansion, this is an identity valid for every finite $\beta$; the estimates below only require $\tau<1$.

If no check touches both $U$ and $W$, then
\[
    |UV|+|VW|=|V|+|UVW|\ ,
    \qquad
    |\checkset_{UV}|+|\checkset_{VW}|
    =|\checkset_V|+|\checkset_{UVW}|\ .
\]
The scalar factors in~\eqref{eq:high-temperature-expansion} therefore cancel, yielding
\begin{equation}\label{eq:fraction}
    \mathrm{DS}(U,V,W)
    =\left\|\frac{\Sigma_{UV}\Sigma_{VW}
    -\Sigma_{UVW}\Sigma_V}
    {\Sigma_{UVW}\Sigma_V}\right\|_\infty \ .
\end{equation}
This expression reduces the verification of the DS-condition to a small-numerator estimate and a lower bound on the denominator. The configuration spaces $\Gamma_B$ encode the geometry of Haah's checks, cf.~Fig.~\ref{fig:gammaillustration}: their dimensions control the number of terms, while the weights $|S|$ control their suppression by $\tau$.\\

Our next aim is to bound the dimensions of the spaces $\Gamma_B$ introduced in~\eqref{eq:boundary-configurations} and the weights of configurations extending across a box, which will provide the geometric input for the DS estimate.

\subsection{Linear algebra of $ X $-checks}

We choose coordinates so that the $XX$-corner of each check has the smallest $x$-, $y$-, and $z$-coordinates. The check anchored at $(i,j,k)$ is denoted by $s(i,j,k)\in\checkset$, and $(s_x,s_y,s_z)$ denotes the coordinates of a check $s$.

We identify subsets of checks and qubits with their binary indicator vectors, so addition corresponds to symmetric difference. In particular, the configurations $S$ used above are elements of $\mathbb F_2^{\checkset}$. 
The check map is defined by
\begin{align}
\begin{matrix}
    \partial &:& \mathbb{F}_2^{\checkset} &\to&  \mathbb{F}_2^{\Lambda\times\{1,2\}} \\[0.2em] 
   & & S &\mapsto& \mathrm{supp}\Big(\prod_{s\in S}P_s\Big)
\end{matrix}
\end{align}
where the support is identified with its binary indicator vector. Since all checks are $X$-type, multiplication of their Pauli operators adds the corresponding support vectors modulo two. Thus $\partial$ is linear.
For any subset $V\subset \Lambda$ 
we define
the restricted check map $\partial_V:\mathbb F_2^{\checkset_V}\to \mathbb{F}_2^{V\times \{1,2\}}$
 by $\partial_V S:=(\partial S)|_{V\times\{1,2\}}$, Eq.~\eqref{eq:boundary-configurations} becomes
\begin{align}
\Gamma_V=\ker\partial_V \ . 
\end{align}

For a box $B$ with side lengths $L_x,L_y,L_z>0$, we choose the origin so that $B$ consists of the lattice vertices in $[1,L_x]\times[1,L_y]\times[1,L_z]$.
In other words, we pick the minimal coordinate of a box to be $(1,1,1)$. Checks whose support intersects the box may be anchored outside of it, with one or more coordinates equal to $0$. The corresponding check set for a box that does not stretch the full torus is 
\begin{equation}
    \checkset_B = \big\{s(i,j,k) \in \checkset \mid (i,j,k)\in [0,L_x]\times[0,L_y]\times[0,L_z]\big\}\setminus \big\{s(0,0,0)\big\} \ ;
\end{equation} 
see Fig.~\ref{fig:box}. 
If a box stretches the full torus $\Lambda$ in one or more directions, that is, if $L_a=L_{\Lambda}$ for one or multiple directions $a$, we equip it with periodic boundary conditions in these directions and pick an arbitrary origin for those coordinates. 
The check set in this case is
\begin{equation}
    \checkset_B = \big\{s(i,j,k) \in \checkset \mid (i,j,k)\in [0,L_x']\times[0,L_y']\times[0,L_z']\big\} \ ,
\end{equation} 
where, for each direction $a\in\{x,y,z\}$, we set $L_a'=L_a-1$
for periodic boundary conditions and $L_a'=L_a$ for open boundary conditions.
All our results hold for boxes with and without such periodic boundaries.

For any box $B$, any direction $a\in \{x,y,z\}$ and any $\ell\leq L_a^B$, we define the set of all checks with $a$-coordinate less than or equal to $\ell$ by
\begin{equation}
    \checkset_B^{a\leq \ell} = \{s\in \checkset_B\mid 0\leq s_a\leq \ell \}\ .
\end{equation}
Note that if $B$ has periodic boundary conditions along the $a$-direction, $\checkset_B^{a\leq \ell}=\checkset_B$ for $\ell\geq L_a-1$.
This allows us to partition $\Gamma_B$ into 
\begin{equation}
    \Gamma_B = \Gamma_B^{a\leq \ell}\sqcup \Gamma_B^{a>\ell} \qquad \mathrm{with}\qquad \Gamma_B^{a\leq \ell}:= \Gamma_B \cap \mathbb{F}_2^{\checkset_B^{a\leq \ell}} \ ,
\end{equation}
where we take the liberty of switching between identifying $ \Gamma_B$ as a subspace and as a subset, which allows us to define the union and intersection on these objects. Note that while $\Gamma_B^{a\leq \ell}$ is a subspace of $\Gamma_B$, its complement $\Gamma_B^{a>\ell}$ is not.
This partition is central in our analysis. It allows us to separate in the combinatorial sums $\sumtauP{B}$, which represent through~\eqref{eq:high-temperature-expansion} the marginals of  the Gibbs state, a main part and 
\begin{equation}
    \sumtauPup{B}{\dir> \ell} := \sum_{S\in \Gamma^{\dir>\ell}_B} \tau^{|S|} P_S \ .
\end{equation}
This part of $ \sumtauP{B} $ represents checks reaching a distance $ \ell $ away from the origin of the box $ B$. It will be the main contribution to the numerator in~\eqref{eq:fraction} (cf.~Fig.~\ref{fig:sum-split}), and will eventually be bounded in Lemma~\ref{lem:smallsumbound}.

\begin{figure}[h]
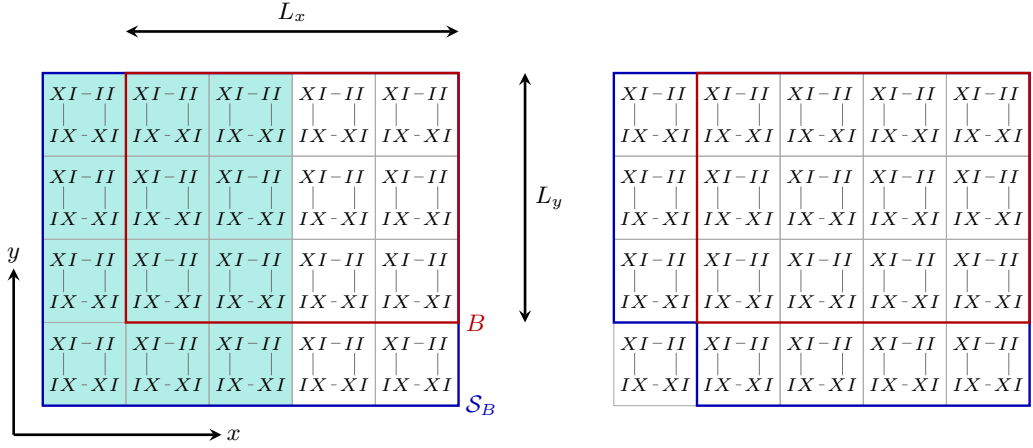

    \centering
    \setlength{\pgridcell}{11mm}
    \begin{pconfig}[grid=5/4, axes=x/y, axes width=1pt]
        \draw[black,stealth-stealth, line width=1pt] (1,0.5) -- (5,0.5) node [midway, above] {\small $L_x$};
        \draw[black,stealth-stealth, line width=1pt] (5.8,0) -- (5.8,-3) node [midway, right] {\small $L_y$};
        \pchecks{5}{4}
        \prectB[blue!70!black]{4}{3}[\checkset_B]
        \prectA[red!70!black]{4}{3}[B]
        \pshade[teal!30]{(0,0) rectangle (3,-4)}
    \end{pconfig}~~~~~
    \begin{pconfig}[grid=5/4]
        \pchecks{5}{4}
        \prectC[blue!70!black]{4}{3}
        \prectA[red!70!black]{4}{3}
    \end{pconfig}
    \caption{Left: an $xy$-slice of a box $B$ at some non-zero $z$-coordinate . The red inner rectangle shows all vertices in $B$; the blue outer rectangle shows all checks connected to $B$. The teal shading marks all the checks in $\checkset_B^{x\leq 2}$. Right: the $xy$-slice of $B$ at $z=0$. Note that the lower left corner is missing from $\checkset_B$, since this check acts as the identity on its upper right corner.}
    \label{fig:box}
\end{figure}

For $v\in B$ and $q\in\{1,2\}$, let $\omega_{v}^{(q)}\in\mathbb{F}_2^{B\times\{1,2\}}$ denote the standard basis vector corresponding to the qubit $(v,q)$. We define constraints 
\begin{align}
     f_v&:=\partial_B^{\mathsf T}\omega_{v}^{(1)}\ ,
    \qquad
    g_v:=\partial_B^{\mathsf T}\omega_{v}^{(2)}\ .
\end{align}
These vectors record which checks act nontrivially on
the first and second qubit at $v$, respectively.
A configuration $S\in\mathbb F_2^{\checkset_B}$ belongs to
$\Gamma_B$ if and only if for every $v\in B$, it holds that
\begin{align}
    \langle f_v,S\rangle=\langle g_v,S\rangle=0 \ ,
\end{align}
where $\langle\cdot,\cdot\rangle$ denotes the standard binary scalar product. 
Let $e_x,e_y,e_z$ denote the coordinate unit vectors of the lattice.
Let $\delta_v\in \mathbb{F}_2^{\mathcal S_B}$ is the standard basis vector
indexed by check $s(v)\in\mathcal S_B$ at position $v\in B$.
We express these constraints as
\begin{align}
 f_v&=\delta_v+\delta_{v-e_x-e_y}
       +\delta_{v-e_x-e_z}+\delta_{v-e_y-e_z},\\
 g_v&=\delta_v+\delta_{v-e_x}
       +\delta_{v-e_y}+\delta_{v-e_z}.
\end{align}
We use the following graphical description:
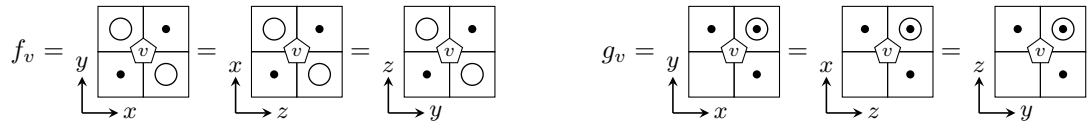
\begin{equation}
    f_v = \pgridv[axes=x/y]{1/1/v}{{o,.},{.,o}} = \pgridv[axes=z/x]{1/1/v}{{o,.},{.,o}}= \pgridv[axes=y/z]{1/1/v}{{o,.},{.,o}}
    \qquad\qquad
    g_v = \pgridv[axes=x/y]{1/1/v}{{.,.o},{-,.}}= \pgridv[axes=z/x]{1/1/v}{{.,.o},{-,.}}= \pgridv[axes=y/z]{1/1/v}{{.,.o},{-,.}} \qquad
\end{equation}
Here $v$ is the position of the qubits in the grid, and the dots and circles correspond to the support of the constraint on the checks in the layers above and below the qubit. 
{\setlength{\pgridcell}{4mm} Specifically, \pgrid{{.}} denotes support on the layer closer to the reader, \pgrid{{o}} denotes support on the layer farther from the reader, and \pgrid{{.o}} denotes support on both layers.}

We derive two families of constraints that will be used in the proofs below. 
Let $(a,b,c)$ be a cyclic permutation of the coordinate directions $(x,y,z)$. We set
\begin{align}
    \Delta_v^{c}&:=f_v+g_{v-e_a}+g_{v-e_b}\ ,  &&\text{ provided } v,v-e_a,v-e_b\in B\ , \notag \\
    \eta_v^a&:=f_v+g_{v-e_a}\ ,  &&\text{ provided } 
    v,v-e_a\in B \ .\label{eq:DeltaEta}
\end{align}

Consequently, $\langle\Delta_v^{c},S\rangle=0$ and
$\langle\eta_v^a,S\rangle=0$ for every $S\in\Gamma_B$
whenever the corresponding conditions on $v$ hold.
In the graphical notation:
\begingroup
\[
\begin{gathered}
  \begin{tikzpicture}[x=1cm,y=1cm,baseline=-2pt]
    \draw[-stealth,line width=.5pt] (0,0)--(.48,0)
      node[right] {\small$a$};
    \draw[-stealth,line width=.5pt] (0,0)--(0,.48)
      node[above] {\small$b$};
  \end{tikzpicture}
\end{gathered}
\quad
\begin{gathered}
 \underset{\textstyle f_v}{
    \pgridv{2/1/v}{{-,o,.},{-,.,o},{-,-,-}}}
  \; + \;
  \underset{\textstyle  g_{v-e_a}}{
    \pgridv{1/1/}{{.,.o,-},{-,.,-},{-,-,-}}}
  \; + \;
  \underset{\textstyle g_{v-e_b}}{
    \pgridv{2/2/}{{-,-,-},{-,.,.o},{-,-,.}}}
  \; = \;
  \underset{\textstyle \Delta_v^{c}}{
    \pgridv{2/1/v,1/1/,2/2/}
      {{.,.,.},{-,.,.},{-,-,.}}}
  \\[12pt]
  \underset{\textstyle f_v}{
    \pgridv{2/1/v}{{-,o,.},{-,.,o}}}
  \; + \;
  \underset{\textstyle g_{v-e_a}}{
    \pgridv{1/1/}{{.,.o,-},{-,.,-}}}
  \; = \;
  \underset{\textstyle \eta_v^a}{
    \pgridv{2/1/v,1/1/}
      {{.,.,.},{-,-,o}}}
\end{gathered}
\]
\endgroup
The pentagons mark the vertices at which the local constraints are applied. All these vertices must lie in $B$. Dots and circles indicate check coefficients, which add modulo two. We note that the constraint $\Delta_v^c$ cannot be applied on the $ab$-slice with $c$ coordinate equal to $0$ when $c$-direction is non-periodic    , see~\eqref{eq:DeltaEta}. 

Later we will also draw $\eta^a_v$ in the $bc$ plane. To draw the three stacked checks in $a$-direction we introduce a new symbol:
\begin{equation}
    \eta^a_v =\pgridv[axes=b/c]{1/1/v}{{-,.oo},{.,-}}
\end{equation}

\subsection{Bounds on dimensions}

\begin{lemma}\label{lem:dimofkernel}
Let $B$ be a box. Then $\dim\Gamma_B \leq 2(L_x+L_y+L_z)$.
\end{lemma}
\begin{proof}
We first prove the claim for the case where $ B $ has open boundary conditions, i.e. $ B $ does not wrap around the torus. By translating coordinates, we may assume that $B=[1,L_x]\times[1,L_y]\times[1,L_z]$.
For $S\in\Gamma_B$, let $S_i(j,k)$ denote the coefficient of the check $s(i,j,k)$, where $0\leq i\leq L_x$, $0\leq j\leq L_y$ and $0\leq k\leq L_z$. In this case, $S_0(0,0)=0$, since $s(0,0,0)\notin\checkset_B$.

We proceed by induction on the $x$-coordinate, starting with the slice $S_{L_x}$.
For $2\leq j\leq L_y$ and $2\leq k\leq L_z$, the three
vertices $(L_x,j,k)$, $(L_x,j-1,k)$ and $(L_x,j,k-1)$
lie in $B$. Hence the triangular constraint
$\langle\Delta_{(L_x,j,k)}^{x},S\rangle=0$ gives
\begin{align}
    S_{L_x}(j,k) 
    &= S_{L_x}(j-2,k) +S_{L_x}(j-1,k-1)+ S_{L_x}(j-1,k)+S_{L_x}(j,k-2)+S_{L_x}(j,k-1) \ . \label{eq:SLx}
\end{align}
Every entry on the right-hand side is of the form $S_{L_x}(j',k') $, with $j'+k'<j+k$. Thus, by induction on $j+k$, the entries with $j\in\{0,1\}$ or $k\in\{0,1\}$ determine the entire slice; see Fig.~\ref{fig:SLx} for illustration.
There are
\begin{align}
    2(L_y+1)+2(L_z+1)-4=2(L_y+L_z)
\end{align}
such entries. The space of possible terminal slices therefore has dimension at most $2(L_y+L_z)$.

\begin{figure}[ht]
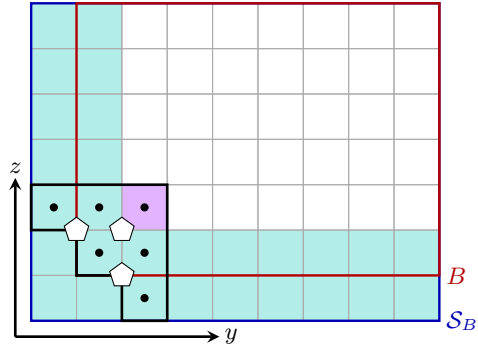

    \centering
    \begin{pconfig}[grid=9/7, axes=y/z, axes width=1pt]
        \prectB[blue!70!black]{8}{6}[\checkset_B]
        \prectA[red!70!black]{8}{6}[B]
        \pshade[teal!30]{(0,0) -- (0,-7) -- (9,-7) -- (9,-5) -- (2,-5) -- (2,0) -- cycle;}
        \pshade[violet!30]{(2,-5) rectangle (3,-4);}
        \draw[black,line width=1pt]
            (0,-4)--(3,-4)--(3,-7)--(2,-7)--(2,-6)
            --(1,-6)--(1,-5)--(0,-5)--cycle;
        \pcellsv[0,-4]{2/1/,1/1/,2/2/}{{.,.,.},{-,.,.},{-,-,.}}
    \end{pconfig}
    \caption{Terminal $yz$-slice $S_{L_x}$ of a check configuration $S\in\Gamma_B$. The $2(L_y+L_z)$ teal entries determine the entire slice by repeated application of the triangular constraint $\Delta^{x}$. The violet shading marks the check being determined in this step.}
    \label{fig:SLx}
\end{figure}

Next, assume the configuration of a slice $S_i$, with $1\leq i\leq L_x$, is known. We show that there are at most $4$ possible configurations of the slice $S_{i-1}$. For $1\leq j\leq L_y$ and $1\leq k\leq L_z$, the vertex
$(i,j,k)$ lies in $B$, so the constraints
$g_{(i,j,k)}$ and $f_{(i,j,k)}$, yield
\begin{align}
    S_{i-1}(j,k)
    &=S_i(j,k)+S_i(j-1,k)+S_i(j,k-1)\ ,\label{eq:diminduction1}\\
    S_{i-1}(j-1,k)+S_{i-1}(j,k-1)
    &=S_i(j,k)+S_i(j-1,k-1)\ .\label{eq:diminduction2}
\end{align}
These equations also apply to the step from $S_1$ to $S_0$.
The additional condition $S_0(0,0)=0$ can only reduce the
number of degrees of freedom. However, for simplicity, we choose not to include this contribution in the presented bound.

First, \eqref{eq:diminduction1} determines all entries
$S_{i-1}(j,k)$ with $1\leq j\leq L_y$ and $1\leq k\leq L_z$.
To determine the boundary entries $S_{i-1}(j,0)$ for $2\leq j\leq L_y$, we apply \eqref{eq:diminduction2} with $k=1$, obtaining
\begin{equation}
    S_{i-1}(j,0) = S_i(j,1) + S_{i}(j-1,0) + S_{i-1}(j-1,1) \ .
\end{equation}
All three terms on the right-hand side are determined by $S_i$: the first two are entries of this slice, while the third, $S_{i-1}(j-1,1)$, is fixed by \eqref{eq:diminduction1}. Similarly, $S_{i-1}(0,k)$ is determined by $S_i$ for every $2\leq k\leq L_z$. Three entries remain undetermined: $S_{i-1}(0,0)$, $S_{i-1}(0,1)$, and $S_{i-1}(1,0)$. Applying \eqref{eq:diminduction2} with $j=k=1$ yields 
\begin{align}
S_{i-1}(0,1) + S_{i-1}(1,0) = S_{i}(0,0) + S_i(1,1) \ .
\end{align}
This leaves at most two degrees of freedom: the value of $S_{i-1}(0,0)$ and the value of either $S_{i-1}(0,1)$ or $S_{i-1}(1,0)$.

\begin{figure}[!ht]
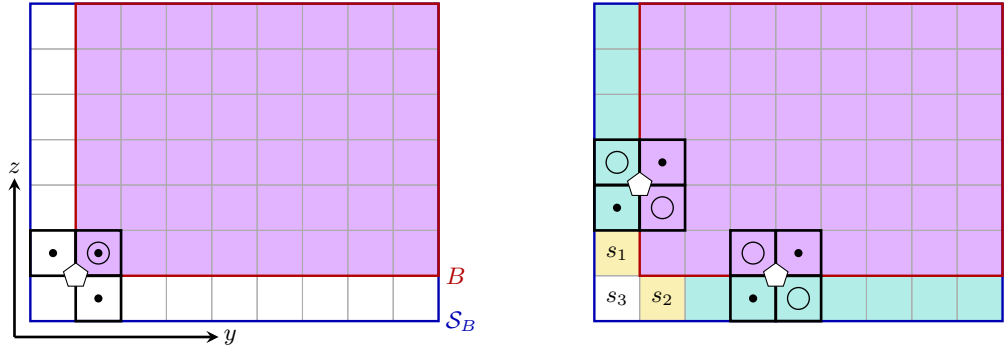

\centering
    \begin{pconfig}[grid=9/7,axes=y/z, axes width=1pt, boxes width=1pt]
        \prectB[blue!70!black]{8}{6}[\checkset_B]
        \prectA[red!70!black]{8}{6}[B]
        \pshade[violet!30]{(1,-6) rectangle (9,0);}
        \pboxesv[0,-5]{1/1/}{{.,.o},{-,.}}
    \end{pconfig}
    \qquad\qquad
   \begin{pconfig}[grid=9/7,  boxes width=1pt]
        \prectB[blue!70!black]{8}{6}
        \prectA[red!70!black]{8}{6}
        \pshade[violet!30]{(1,-6) rectangle (9,0);}
        \pshade[teal!30]{(0,-5) rectangle (1,0);}
        \pshade[teal!30]{(2,-7) rectangle (9,-6);}
        \pboxesv[0,-3]{1/1/}{{o,.},{.,o}}
        \pboxesv[3,-5]{1/1/}{{o,.},{.,o}}
        \pshade[yellow!30]{(0,-6) rectangle (1,-5);}
        \pshade[yellow!30]{(1,-7) rectangle (2,-6);}
        \pcells[0,-5]{{s_1,-},{s_3,s_2}}
    \end{pconfig}
    
    \caption{Constraining the layer $S_{i-1}$ in two steps for a fixed configuration on layer $S_i$. This leaves at most four compatible configurations. Left: $g_v$ determines the values in the shaded violet box. We place the checks such that the circle lies on $S_{i-1}$ and the dots on $S_i$. Right: $f_v$ determines the values in the two teal boundary strips. Of the three remaining entries, $s_1=S_{i-1}(0,1)$ and $s_2=S_{i-1}(1,0)$ have a fixed parity imposed by $f_{(i,1,1)}$. This leaves at most two degrees of freedom, and hence at most four compatible configurations. }
    \label{fig:Siminus1}
\end{figure}

By the rank--nullity theorem, each step therefore increases the dimension by at most two. Combining this with the bound on the terminal
slice and taking $L_x$ steps, we obtain
\begin{align}
    \dim\Gamma_B
    &\leq 2(L_y+L_z)+2L_x =2(L_x+L_y+L_z) \ .
\end{align}

Now let us consider the case when $B$ has periodic boundary conditions in one or multiple directions.
We choose a cut in each periodic direction $a$ and index the distinct check layers by
$0,\ldots,L_a-1$. Then we proceed as in the case of open boundary conditions, using these index sets. Note that this uses only constraints $f_v$, $g_v$, and $\Delta_v$ that do not cross a periodic cut.
Since the number of indices does not increase, the number of configurations of the  slice $S_{L_x}$ and the number of steps can not increase. Since the constraint $S_0(0,0)=0$ was not used in the open case, the same bound follows.

This completes the proof.
\end{proof}

\begin{lemma}\label{lem:coreweight}
Let $B=[1,L_x]\times[1,L_y]\times[1,L_z]$ be a box and let $S\in\Gamma_B$. Let $r\ge 1$ be an integer such that $L_{\min}\geq 3r-2$. Fix a cyclic permutation $(a,b,c)$ of the coordinates. Set $C_r=[r,L_b]\times [r, L_c]$. If $S_i|_{C_r}\ne 0$  for some $i\in\{r,\ldots,L_a\}$, then
\begin{align}
    |S|\geq \frac{r(r+1)}{2}
\end{align}
\end{lemma}
\begin{proof}
By symmetry, let $(a,b,c)=(x,y,z)$.
We first prove a lower bound for a layer whose restriction to $C_r$
is nonzero, and a stronger bound when the preceding layer vanishes
on $C_r$. We then combine these estimates over $r$ consecutive layers.

For $r=1$, the assumption $S_i|_{C_r}\neq0$ implies $|S|\geq1$, proving the claim.
Henceforth, assume $r\geq2$. Consider a $yz$-layer $S_i$, where $1\leq i\leq L_x$.  Assume that $S_i|_{C_r}\neq0$.

First, we show that
\begin{align}
    |S_i| \geq \lfloor r/2\rfloor+1 \geq \frac{r+1}{2} \ . \label{eq:eq:check1first}
\end{align}

For $2\leq j\leq L_y$ and $2\leq k\leq L_z$, the three
vertices $(i,j,k)$, $(i,j-1,k)$ and $(i,j,k-1)$ are contained in $B$.
Hence the triangular relation
$\langle\Delta_{(i,j,k)}^x,S\rangle=0$ yields 
\begin{align}
    S_i(j,k) &=S_i(j-1,k)+S_i(j,k-1)+S_i(j-2,k)+S_i(j-1,k-1)+S_i(j,k-2)\ .\label{eq:check1triangle}
\end{align}

The following construction is illustrated in
Figure~\ref{fig:path}. Choose $(j_0,k_0)\in C_r$ with $S_i(j_0,k_0)=1$.
Whenever the current entry $S_i(j,k)$ has $j,k\geq2$, the triangular constraint~\eqref{eq:check1triangle} forces at least one entry on its right-hand side to be nonzero. Choose such an entry and repeat the argument.

\begin{figure}[!ht]
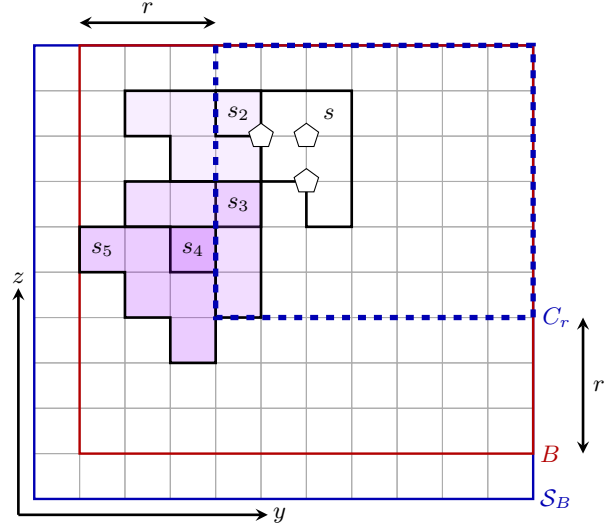

\centering
   \begin{pconfig}[grid=11/10, axes=y/z, axes width=1pt, boxes width=1pt]
        \prectB[blue!70!black]{10}{9}[\checkset_B]
        \prectA[red!70!black]{10}{9}[B]
        \ptri[black, line width=1pt, shift={(2,-1)},fill=violet!33, fill opacity=0.2]{3}
        \ptri[black, line width=1pt, shift={(2,-3)},fill=violet!66, fill opacity=0.2]{3}
        \ptri[black, line width=1pt, shift={(1,-4)},fill=violet, fill opacity=0.2]{3}
        \ptri[black, line width=1pt, shift={(4,-1)}]{3}[2/1/,1/1/,2/2/]
        
        \prectA[blue!70!black,dashed, shift={(3,0)},line width=2pt]{7}{6}[C_r]
        \pcells[6,-1]{s}
        \pcells[4,-1]{s_2}
        \pcells[4,-3]{s_3}
        \pcells[3,-4]{s_4}
        \pcells[1,-4]{s_5}
        
        \draw[black,stealth-stealth, line width=1pt] (12.1,-6) -- (12.1,-9) node [midway, right] {\small $r$};
        \draw[black,stealth-stealth, line width=1pt] (1,0.5) -- (4,0.5) node [midway, above] {\small $r$};
    \end{pconfig}
    
    \caption{An example of a path of constraints. Assuming $S_i$ has support on $s$, then it has to have some support in the triangle constraint. Suppose e.g.\ on $s_2$. Placing a triangle constraint on $s_2$ yields another spot where $S_i$ has support, for example on $s_3$. Iterating this, every new site has its distance to the left and bottom boundaries reduced by at most $2$. Thus there are at least $\lfloor r/2 \rfloor$ such steps until the boundary is reached, where the triangle constraint no longer holds.}
    \label{fig:path}
\end{figure}

At each step, both coordinates are nonincreasing and decrease
by at most two, while their sum strictly decreases.
Thus all chosen entries are distinct, and the process
terminates at an entry with one coordinate in $\{0,1\}$,
where the three vertices needed for the triangular constraint
do not all lie in~$B$. 

Let $m$ be the number of steps taken.
Since the starting entry lies in $C_r$, both coordinates
of the terminal entry are at least $r-2m$.
One of them is at most one, so
$m\geq\lfloor r/2\rfloor$. 
If $B$ has periodic boundary conditions on $y$ or $z$, the process could continue further until it completes a full loop around the torus. For simplicity, we terminate it after the same number $\lfloor r/2\rfloor$ of steps as for $B$ with open boundary condition.

Including the starting entry this gives at least
$\lfloor r/2\rfloor+1$ distinct nonzero entries in $S_i$,
and therefore proves \eqref{eq:eq:check1first}.

Second, assume additionally that $S_{i-1}|_{C_r}=0$.
We show that
\begin{align}
    |S_i|\geq r(2r-1)\geq \frac{r(r+1)}{2}.
    \label{eq:check1second}
\end{align}
Since $C_r$ is contained in the coordinate domain of the layer, this will work for both kinds of boundary conditions.
For $r+1\leq j\leq L_y$ and $r+1\leq k\leq L_z$,
the constraints $f_{(i,j,k)}$ and $g_{(i,j,k)}$, together with $S_{i-1}|_{C_r}=0$, yield
\begin{align}
    S_i(j,k)&=S_i(j-1,k-1) \ , \label{eq:corediagonal}\\
    S_i(j,k)+S_i(j-1,k)+S_i(j,k-1)&=0 \ . \label{eq:coretriangle}
\end{align}
These two constraints are illustrated in Fig.~\ref{fig:emptysuccessor}.
The first identity shows that the entries of $S_i$ are
constant along each diagonal $j-k=d$ in $C_r$.
Let $a_d$ denote the common value on the diagonal $j-k=d$,
where $r-L_z\leq d\leq L_y-r$.
Equation~\eqref{eq:coretriangle} gives
\begin{align}
    a_{d-1}+a_d+a_{d+1}=0
    \qquad \text{ whenever } r-L_z<d<L_y-r \ .
    \label{eq:corerecurrence}
\end{align}
Since the sum in Eq.~\eqref{eq:corerecurrence} is taken modulo two, each triple $(a_{d-1},a_d,a_{d+1})$ contains an even number of entries equal to one, hence either none or exactly two. If such a triple were zero, repeated application of Eq.~\eqref{eq:corerecurrence} in both directions would force all diagonal values to vanish, contradicting
$S_i|_{C_r}\neq 0$. Thus, exactly two of every three consecutive diagonals
have value one. Since consecutive entries in a row lie on consecutive
diagonals, every three consecutive entries in a row contain exactly two nonzero entries.

Since $L_y,L_z\geq3r-2$, the rectangle $R=[r,3r-2]\times[r,3r-2]$
is contained in $C_r$. Each row of $R$ contains $2r-1$ entries. The zero entries are separated by at least three positions, so each row contains at most $\lceil(2r-1)/3\rceil$ zero entries. Since $r\geq2$, we have $\lceil(2r-1)/3\rceil\leq r-1$, so each row contains at least $r$ nonzero entries. Summing over the $2r-1$ rows gives
\begin{align}
    |S_i|\geq r(2r-1)\geq \frac{r(r+1)}{2}  \ ,
\end{align}
proving~\eqref{eq:check1second}.

\begin{figure}[!ht]
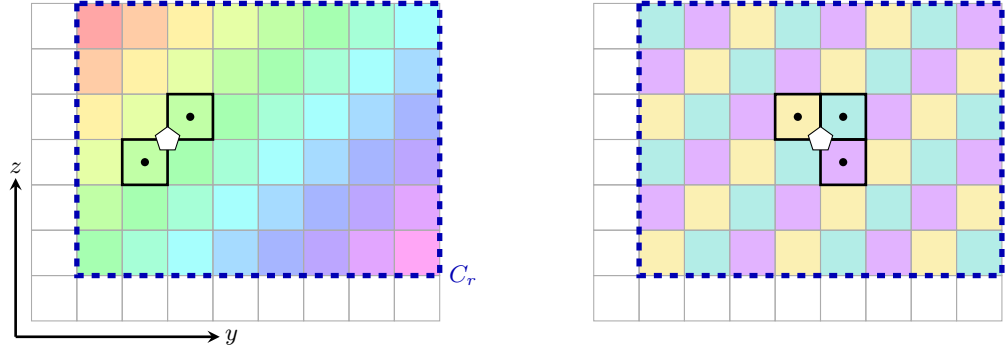

\centering
   \begin{pconfig}[grid=9/7, axes=y/z, axes width=1pt, boxes width=1pt]
        \pboxesv[2,-2]{1/1/}{{-,.},{.,-}}
        \pdiag[shift={(1,0)}]{8}{6}
        \prectA[blue!70!black,dashed, shift={(0,0)},line width=2pt]{8}{6}[C_r]
    \end{pconfig}
    \qquad\qquad
    \begin{pconfig}[grid=9/7, boxes width=1pt]
        \pshade[blue!20]{(1,-1) rectangle (2,-2);}
        \pboxesv[4,-2]{1/1/}{{.,.},{-,.}}
        \pdiag[diag colors=cycle,shift={(1,0)}, diag cycle={teal,violet,yellow}, diag saturation=0.3]{8}{6}
        \prectA[blue!70!black,dashed, shift={(0,0)},line width=2pt]{8}{6}
    \end{pconfig}
    
    \caption{The core of the slice $S_m$ assuming $S_{m-1}|_{C_r}=0$. Left: The constraints $f_v$ between the layers $S_m$ and $S_{m-1}$ reduce to a parity constraint on the diagonal. This forces values along a diagonal to be identical. Right: The constraints $g_v$ reduce to small triangle constraints. Since $S_m$ has at least some support in $C_r$, these constraints force two out of every three diagonals to be in the support of $S_m$.}
    \label{fig:emptysuccessor}
\end{figure}

To conclude, consider the $r$ consecutive layers
$S_i,S_{i-1},\ldots,S_{i-r+1}$. If every one of them satisfies $S_m|_{C_r}\neq 0$ then Eq.~\eqref{eq:eq:check1first} yields
\begin{align}
|S| \geq \sum_{m=i-r+1}^{i}|S_m| \geq r\cdot\frac{r+1}{2} \ .    
\end{align}
Otherwise some consecutive pair satisfies $S_m|_{C_r}\neq0$ and $S_{m-1}|_{C_r}=0$. 
Since $r<L_x$, periodic boundary conditions in the $x$-direction do not affect this argument. The second claim~\eqref{eq:check1second} thus completes the proof.
\end{proof}

\begin{lemma}\label{lem:checksize}
    Let $B=[1,L_x]\times[1,L_y]\times[1,L_z]$ be a box, $\dir\in \{x,y,z\}$ and let $1\leq \ell\leq L_a-2$.
    Then, every element $S\in \Gamma_B^{\dir> \ell}$ has size
    \begin{align}
        |S|\geq \frac{1}{100}\min\{\ell^2, L_{\min}^2\} \ .
    \end{align}
\end{lemma}
Note that we require $\ell\leq L_a-2$, since if $B$ has periodic boundary conditions along the $a$-direction, $\Gamma_B^{\dir> L_a-1} =\emptyset$.
\begin{proof}
By symmetry, we assume $\dir=x$.
We set
\begin{align}
n=\min\{\ell,L_{\min}\}\qquad \text{and}\qquad r=\lfloor n/3\rfloor-1\ .\label{eq:nrparams}
\end{align}
If $n<6$, the claim follows from $|S|\geq 1$. We thus assume $n\geq 6$, so that $r\geq1$. 
Since $S\in\Gamma_B^{x>\ell}$, there is a check $s(c_x,c_y,c_z)\in S$ with $c_x>\ell$.

We will use the constraints to find a check in $S$ whose coordinates are at least~$r$. Since $3r+2\leq n\leq L_{\min}$, Lemma~\ref{lem:coreweight} will then apply. If $c_y,c_z\geq r$, set $(c'_x,c'_y,c'_z)=(c_x,c_y,c_z)$. Otherwise, exchanging $y$ and $z$ if necessary, assume
$c_y<r$. We distinguish two cases.

First, assume that $c_z\geq L_z-r$. Whenever $(i,j+1,k)\in B$, the constraint for $g_{(i,j+1,k)}$ gives
\begin{align}
S_i(j,k) &=S_i(j+1,k)+S_{i-1}(j+1,k)+S_i(j+1,k-1)\ . \label{eq:inwardg}
\end{align}
If $S_i(j,k)=1$, at least one entry on the right-hand side of~\eqref{eq:inwardg} has value $1$.
We choose such an entry and repeat the argument. Each step increases the $y$-coordinate by one and decreases
each of the $x,z$-coordinates by at most one.
Starting from $(c_x,c_y,c_z)$, we can take $r+1$ steps. Indeed, the required vertices have $y$-coordinates between
$1$ and $c_y+r+1\leq2r\leq L_y$. Their $x,z$-coordinates never exceed their initial values
and remain bounded below by $c_x-r-1>0$ and $c_z-r-1\geq L_z-2r-1\geq r+1$, respectively.
Thus all required vertices lie in $B$. The resulting nonzero entry has coordinates
$c'_x,c'_y,c'_z$ satisfying
\begin{align}
    c'_x&\geq c_x-(r+1)>r \ ,\qquad
    c'_y=c_y+r+1\geq r+1 \ ,\qquad
    c'_z\geq c_z-(r+1)\geq r+1 \ .
\end{align}

Second, assume that $c_z<L_z-r$. Whenever both $v=(i,j+1,k+1)$ and $v-e_x$ lie in $B$,
the constraint $\langle\eta_v^x,S\rangle=0$ gives 
\begin{align}
    S_i(j,k)
    &=S_i(j+1,k+1)+S_{i-1}(j+1,k+1)
      +S_{i-2}(j+1,k+1).
    \label{eq:etafacingreader}
\end{align}
Again, if $S_i(j,k)=1$, at least one entry on the right-hand side
of~\eqref{eq:etafacingreader} is nonzero. Choose such an entry and repeat the argument.
Each step increases both $y,z$-coordinates by one and decreases the $x$-coordinate by at most two.
Starting from $(c_x,c_y,c_z)$, we can again take $r+1$ steps. Indeed, the required vertices have positive $y,z$-coordinates, bounded above by $c_y+r+1\leq2r\leq L_y$ and $c_z+r+1\leq L_z$, respectively.
Their $x$-coordinates never exceed $c_x\leq L_x$ and are bounded below by $c_x-2r-1>0$.
Thus both vertices required for each constraint lie in $B$.
The resulting nonzero entry has coordinates
$c'_x,c'_y,c'_z$ satisfying
\begin{align}
    c'_x&\geq c_x-2(r+1)>r \ ,\qquad
    c'_y=c_y+r+1\geq r+1 \ ,\qquad
    c'_z=c_z+r+1\geq r+1 \ .
\end{align}

The two constructions are illustrated in Figure~\ref{fig:checks2}.
Periodic boundary conditions in $y$- or $z$-direction can only decrease the number of steps needed to reach $C_r$ since checks at one boundary are identified with those at the opposite boundary.
In all cases, $c'_x\geq r$ and $S_{c'_x}|_{C_r}\neq 0$. 
Lemma~\ref{lem:coreweight} therefore yields, for both open and periodic boundary conditions, that
\begin{align}
    |S|
    \geq \frac{r(r+1)}{2}
    \geq \frac{(r+1)^2}{4}
    \geq \frac{n^2}{100}
    = \frac{1}{100}\min\{\ell^2,L_{\min}^2\},
\end{align}
where we used $r\geq1$ and $n<5(r+1)$. This is the claim.
\end{proof}
\begin{figure}[!h]
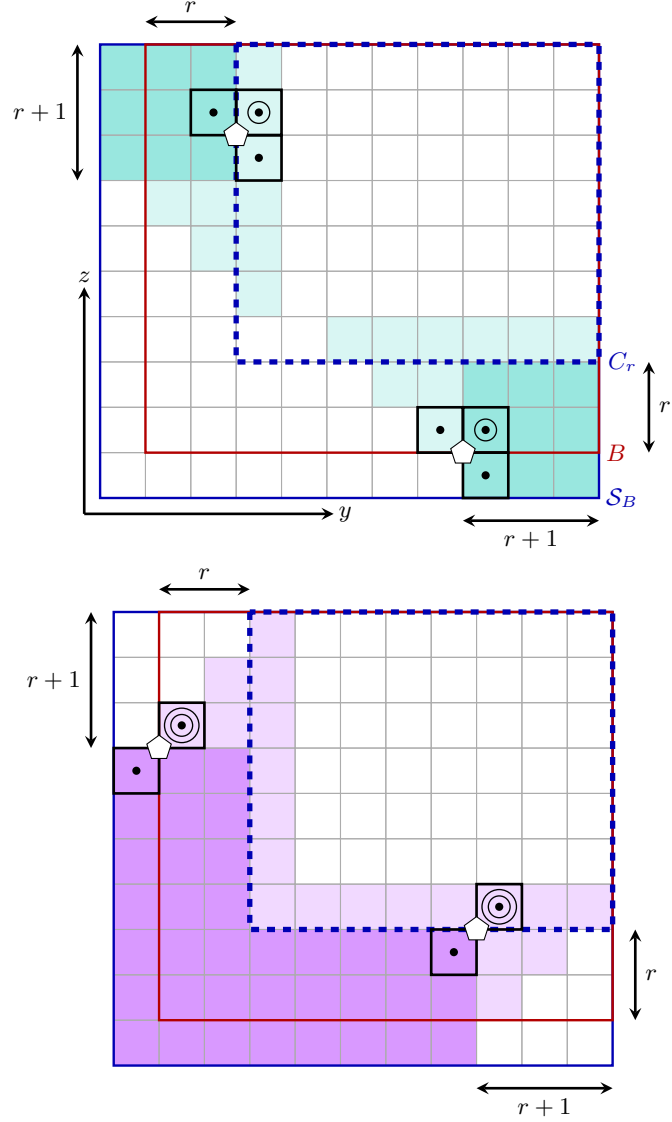

\centering
   \begin{pconfig}[grid=11/10, axes=y/z, axes width=1pt, boxes width=1pt]
        \prectB[blue!70!black]{10}{9}[\checkset_B]
        \prectA[red!70!black]{10}{9}[B]
        \prectA[blue!70!black,dashed, shift={(2,0)},line width=2pt]{8}{7}[C_r]
        \pshade[teal!40]{(0,0) rectangle (3,-3)};
        \pshade[teal!15]{(3,0) rectangle (4,-3)};
        \pshade[teal!15]{(1,-3) rectangle (4,-4)};
        \pshade[teal!15]{(2,-4) rectangle (4,-5)};
        \pshade[teal!15]{(3,-5) rectangle (4,-6)};
        
        \pshade[teal!40]{(8,-7) rectangle (11,-10)};
        \pshade[teal!15]{(8,-6) rectangle (11,-7)};
        \pshade[teal!15]{(7,-6) rectangle (8,-9)};
        \pshade[teal!15]{(6,-6) rectangle (7,-8)};
        \pshade[teal!15]{(5,-6) rectangle (6,-7)};

        \pboxesv[2,-1]{1/1/}{{.,.o},{-,.}}
        \pboxesv[7,-8]{1/1/}{{.,.o},{-,.}}

        \draw[black,stealth-stealth, line width=1pt] (12.1,-7) -- (12.1,-9) node [midway, right] {\small $r$};
        \draw[black,stealth-stealth, line width=1pt] (1,0.5) -- (3,0.5) node [midway, above] {\small $r$};
        \draw[black,stealth-stealth, line width=1pt] (-0.5,0) -- (-0.5,-3) node [midway, left] {\small $r+1$};
        \draw[black,stealth-stealth, line width=1pt] (8,-10.5) -- (11,-10.5) node [midway, below] {\small $r+1$};
    \end{pconfig}
    \quad
    \begin{pconfig}[grid=11/10, boxes width=1pt]
        \prectB[blue!70!black]{10}{9}
        \prectA[red!70!black]{10}{9}

        \prectA[blue!70!black,dashed, shift={(2,0)},line width=2pt]{8}{7}

        \pshade[violet!40]{(0,-3)--(0,-10)--(8,-10)--(8,-7)--(3,-7)--(3,-3)--cycle;};
        \pshade[violet!15]{(1,-2) rectangle (2,-3)};
        \pshade[violet!15]{(2,-1) rectangle (3,-3)};
        \pshade[violet!15]{(3,0) rectangle (4,-7)};
        \pshade[violet!15]{(3,-6) rectangle (11,-7)};
        \pshade[violet!15]{(8,-7) rectangle (9,-9)};
        \pshade[violet!15]{(9,-7) rectangle (10,-8)};

        \pboxesv[0,-2]{1/1/}{{-,.oo},{.,-}}
        \pboxesv[7,-6]{1/1/}{{-,.oo},{.,-}}

        \draw[black,stealth-stealth, line width=1pt] (11.5,-7) -- (11.5,-9) node [midway, right] {\small $r$};
        \draw[black,stealth-stealth, line width=1pt] (1,0.5) -- (3,0.5) node [midway, above] {\small $r$};
        \draw[black,stealth-stealth, line width=1pt] (-0.5,0) -- (-0.5,-3) node [midway, left] {\small $r+1$};
        \draw[black,stealth-stealth, line width=1pt] (8,-10.5) -- (11,-10.5) node [midway, below] {\small $r+1$};
    \end{pconfig}
    
    \caption{Top: Every non-zero entry in one of the two solid teal squares in the corner can be transported using the $g_v$ constraints to a non-zero entry in $C_r$ using at most $r+1$ steps. Bottom: Non-zero entries in the solid violet shape can be transported inside $C_r$ using the $\eta_v^x$ constraints and a maximum of $r+1$ steps. }
    \label{fig:checks2}
\end{figure}

\section{Proof of the DS-condition}

In this section, we use the geometric properties of the checks established in the previous section to prove the DS-condition~\eqref{eq:DS-condition} for the $ X $-part. (The $ Z $-part proceeds similarly.) The geometry of the regions $U,V,W$ for open and periodic boundary conditions is illustrated in Fig.~\ref{fig:uvw}. By~\eqref{eq:fraction}, this amounts to bounding the operator norm of
\begin{align}
    \frac{\Sigma_{UV}\Sigma_{VW}- \Sigma_{UVW}\Sigma_{V}}{\Sigma_{UVW}\Sigma_{V}} \ . 
\end{align}
Since $\Sigma_{UVW}\Sigma_V$ is positive definite, it suffices to bound the norm of the numerator from above and the smallest eigenvalue of the denominator from below.

\begin{figure}[h]
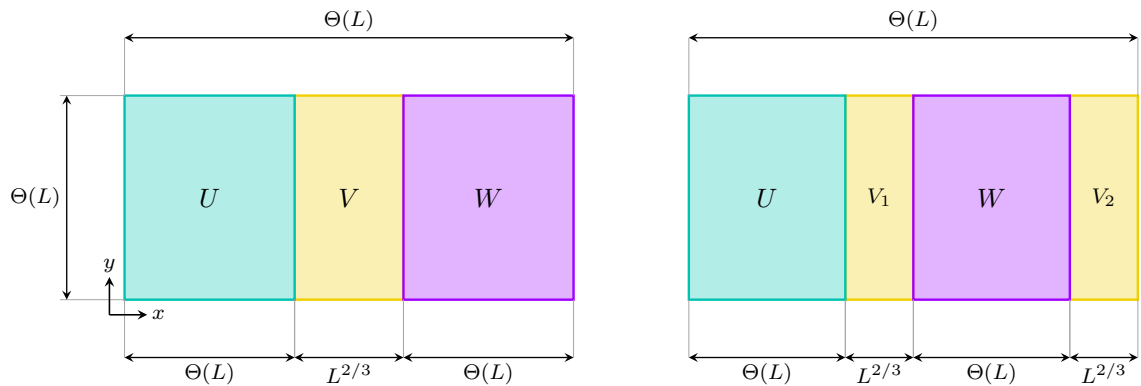

    \centering
    \uvw[height label={\Theta(L)}, u width=2.5, w width=2.5, axes=x/y]
    \qquad\qquad
    \uvwper[wrap=false, height=3, v width=1, u width=2.3, w width=2.3]
    \caption{Geometry of the DS-condition. Left: open boundary conditions. Right: periodic boundary conditions, specialising Fig.~\ref{fig:dsuvw} to our setting.}
    \label{fig:uvw}
\end{figure}

\subsection{Bounding the denominator }

The bound on the denominator relies on the following general polynomial estimate involving subspaces $J\subset\mathbb F_2^m$. We call a map
$\chi:J\to\{-1,1\}$ a \emph{character of $J$} if, for all $p,q\in J$, it satisfies $\chi(p+q)=\chi(p)\chi(q)$.
\begin{lemma}\label{lem:characterbound}
Let $J\subset\mathbb F_2^m$ be a $j$-dimensional subspace and let $\chi:J\to\{-1,1\}$ be a character of~$J$. Then, for every $0\le t<1$, it holds that
\begin{align}
    \sum_{p\in J}t^{|p|}\chi(p)\ge(1-t)^j \ .
\end{align}
Here, $|p|$ denotes the number of nonzero coordinates of $p$,
and we use the convention $t^0=1$, including when $t=0$.
\end{lemma}

\begin{proof}
We express the weighted sum as the expectation of a nonnegative random variable. Fix $0\le t<1$. We use Gaussian elimination to choose
a set $I\subset [m]$ of $j$ coordinate positions
such that the only vector in $J$ vanishing on $I$ is the
zero vector. We form a random set $S\subset [m]$ by
including each coordinate position independently with
probability $1-t$. Define the subspace of vectors vanishing on $S$ and its character sum by
\begin{align}
    K:=\{p\in J \mid p_j=0\text{ for all }j\in S\} \ 
    \qquad \text{ and }\qquad
    \Xi :=\sum_{p\in K}\chi(p) \ .
\end{align}

We argue that the random variable $\Xi$ is nonnegative.
If $\chi$ equals $1$ at every point of $K$, then $\Xi=|K|$.
Otherwise, we choose $q\in K$ such that $\chi(q)=-1$.
Since translation by $q$ permutes $K$, we have
\begin{align}
    \Xi=\sum_{p\in K}\chi(p+q)
     =\chi(q)\sum_{p\in K}\chi(p)
     =-\Xi \ ,
\end{align}
and thus $\Xi=0$. Therefore $\Xi\ge 0$ in either case.
Moreover, if $I\subset S$, injectivity of restriction to $I$ implies $K=\{0\}$, so $\Xi=\chi(0)=1$.

Finally, $\Pr[p\in K]=t^{|p|}$ for every $p\in J$, since
$p\in K$ exactly when $S$ contains none of its nonzero
coordinate positions. By linearity of expectation,
\begin{align}
    \sum_{p\in J}t^{|p|}\chi(p)
    &=\mathbb E[\Xi]
    \ge \Pr[I\subset S]
    =(1-t)^j \ .
\end{align}
The last inequality follows because $\Xi$ is nonnegative and
equals $1$ on the event $I\subset S$.
\end{proof}

We use Lemma~\ref{lem:characterbound} to control the denominator of~\eqref{eq:fraction}.
\begin{lemma}\label{lem:denominator}
Let $U,V,W$ be pairwise disjoint regions such that $UV$, $VW$ and $UVW$ are  boxes and $V$ is a box or a disjoint union of two boxes $V_1,V_2$. 
Then
\begin{align}
\lambda_{\min}\left(\sumtauP{UVW}\sumtauP{V} \right) 
    &\ge (1-\tau)^{\dim \Gamma_{UVW}+\dim\Gamma_V} \label{eq:denomfirstbound}\\
    &\geq (1-\tau)^{18 L_{\max}^{UVW}}\ .
\end{align}
\end{lemma}
\begin{proof}
Note that $V$ is either a box, or if $UVW$ has periodic boundary conditions, the union of two disjoint boxes $V=V_1\sqcup V_2$. The proof works in either case.
We choose a common eigenbasis for the commuting self-adjoint products $\{P_S\mid S\in\Gamma_{UVW}\}\cup\{P_T \mid T\in\Gamma_V\}$ and fix one basis vector. We can write $\chi_{UVW}(S)$ and $\chi_V(T)$ for the corresponding eigenvalues, which take values in $\{-1,1\}$. The identity $P_{S+T}=P_SP_T$ implies that these maps are characters of $\Gamma_{UVW}$ and $\Gamma_V$, respectively.
Lemma~\ref{lem:characterbound} is hence applicable:
\begin{align}
    \sum_{S\in\Gamma_{UVW}}\tau^{|S|}\chi_{UVW}(S)
    &\geq (1-\tau)^{\dim\Gamma_{UVW}}\ , \\
    \sum_{T\in\Gamma_V}\tau^{|T|}\chi_V(T)
    &\geq (1-\tau)^{\dim\Gamma_V} \ . \label{eq:charboundsep}
\end{align}
Multiplying these positive lower bounds shows that every eigenvalue of the operator product is at least $(1-\tau)^{\dim\Gamma_{UVW}+\dim\Gamma_V}$. This implies~\eqref{eq:denomfirstbound}.
Since $V\subset UVW$, we have
$L_\dir^V\leq L_\dir^{UVW}$ for each $\dir \in\{x,y,z\}$. In the case of periodic boundary conditions, we have $L_\dir^{V_1},L_\dir^{V_2}\leq L_\dir^{UVW}$ and $\dim \Gamma_V =  \dim \Gamma_{V_1} + \dim \Gamma_{V_2}$. Lemma~\ref{lem:dimofkernel} therefore gives
\begin{align}
    \dim\Gamma_{UVW}+\dim\Gamma_V
    &\leq 6\bigl(L_x^{UVW}+L_y^{UVW}+L_z^{UVW}\bigr)
    \leq 18 L_{\max}^{UVW} \ ,
\end{align}
which completes the proof.
\end{proof}

\subsection{Bounding the numerator}
\begin{figure}[h]
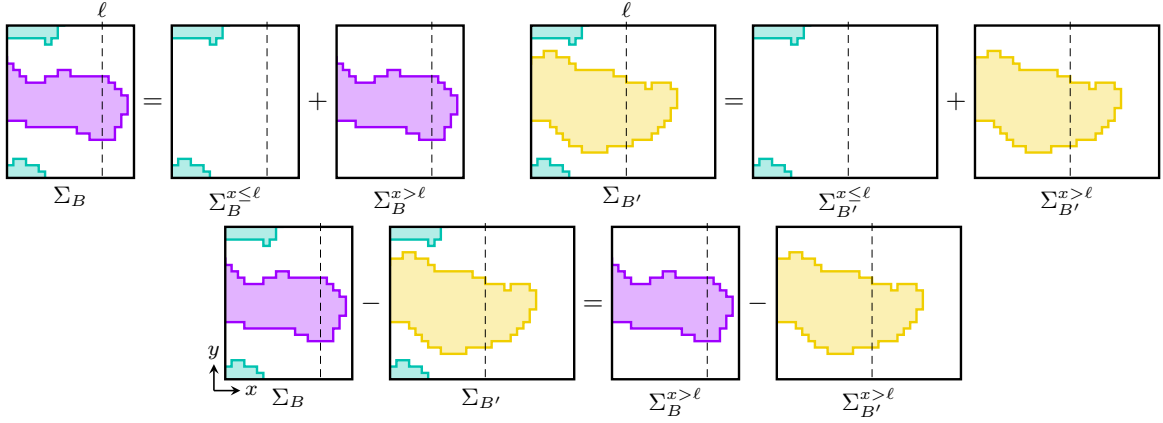

    \centering
    \uvwbbid
    \caption{The partition of $\Sigma_B$ and $\Sigma_{B'}$ used in Lemma~\ref{lem:switching-light}. The sum is over all different configurations in $\Gamma_B$. Note that the configurations drawn here are only a sketch and are not the slice of any actual elements of $\Gamma_B$. The parts supported to the left of $\ell$ are identical for $B$ and $B'$ and vanish when subtracting $\Sigma_B - \Sigma_{B'}$. The parts farther away can differ, but they have a large support.}
    \label{fig:sum-split}
\end{figure}

We start with a crucial observation concerning cancelation in the sums $ \Sigma_B $ for nested boxes.
\begin{lemma}\label{lem:switching-light}
Let $(a,b,c)$ be a cyclic permutation of $(x,y,z)$ and consider the nested boxes
\begin{align}
    B = [1,L_a]\times[1,L_b]\times[1,L_c] \qquad \text{and} \qquad B' = [1,L_a']\times[1,L_b]\times[1,L_c] \ ,
\end{align}
with $L_a\leq L_a'$. Assume both $B,B'$ do not have periodic boundary conditions in $a$-direction. Then
\begin{equation}
    \Sigma_B-\Sigma_{B'} = \Sigma_B^{a > \ell} -\Sigma_{B'}^{a > \ell} \ ,
\end{equation}
for every $1\leq \ell\leq L_a-2$.
\end{lemma}
\begin{proof}
First note that for every $1\leq \ell\leq L_a-2$
\begin{equation}
    \Gamma_B^{a \leq \ell} = \Gamma_{B'}^{a \leq \ell} .
\end{equation}
Indeed, checks $s$ with $s_a\leq\ell$ have no support in $B'\setminus B$. 
Thus $\partial_B s= \partial s \cap B= \partial s \cap B' =\partial_{B'}s$.
The identity follows by linearity.

Recall that $\Sigma_B$ is the sum over elements in $\Gamma_B$. The claim then follows by partitioning $\Gamma_C=\Gamma_C^{a\leq\ell}\sqcup\Gamma_C^{a>\ell}$, for any $C\in\{B,B'\}$,
and cancelling the identical contributions from $\Gamma_B^{a\leq\ell}$ and $\Gamma_{B'}^{a\leq\ell}$; see Figure \ref{fig:sum-split} for an illustration of the partition.
\end{proof}

Our argument also relies on the following energy-entropy bound on the part in $ \Sigma_B $, which has long-range support. 
\begin{lemma}\label{lem:smallsumbound}
    There exist constants $\nu_{\max},\nu_{\min}>0$ independent of system size such that for any 
    box $B$, direction $\dir\in \{x,y,z\}$ and $1\leq \ell\leq L_a-2$ 
    \begin{align}
        \| \sumtauPup{B}{\dir > \ell}\|_\infty &\leq \exp(\nu_{\max} L_{\max} - \nu_{\min} \min\{\ell^2, L_{\min}^2\}) \ .
    \end{align}
\end{lemma}
\begin{proof}
    By the triangle inequality and the fact that $\|P_S\|_\infty=1$, we obtain
    \begin{align}
       \| \sumtauPup{B}{\dir > \ell}\|_\infty =  \left\|\sum_{S\in \Gamma_B^{\dir> \ell}} \tau^{|S|} P_S \right\|_\infty 
        &\leq \left|\Gamma_B^{\dir> \ell}\right| \max_{S\in \Gamma_B^{\dir> \ell}} \tau^{|S|}
    \end{align}
    Applying Lemmas~\ref{lem:dimofkernel} and~\ref{lem:checksize} to bound the number of terms and their weights, respectively, we obtain
    \begin{align}
        \| \sumtauPup{B}{\dir > \ell}\|_\infty 
        &\leq \exp\left( 2\ln2(L_x+L_y+L_z) + \frac{\ln\tau}{100} \min\{\ell^2, L_{\min}^2\}\right) \notag \\
        &\leq \exp\left( 6(\ln2 )L_{\max} + \frac{\ln\tau}{100} \min\{\ell^2, L_{\min}^2\} \right) \ .
    \end{align}
    Setting $\nu_{\max}=6\ln 2$ and $\nu_{\min}=-(\ln\tau)/100$ completes the proof.
\end{proof}
Putting the above items together, we arrive at an estimate for the numerator in~\eqref{eq:fraction}.
\begin{lemma}\label{lem:numerator}
    There exist constants $\mu_L,\mu_d,L_0>0$, depending only on $\beta$, such that the following holds in the DS-geometry described in Subsection~\ref{s:DSgeo} with  $d =\dist{U,W}$.
     If $d\geq 4$ and $L_{\min}^{UVW}>L_0$, then
    \begin{align}
     \left\| \sumtauP{UV}\sumtauP{VW} - \sumtauP{UVW}\sumtauP{V}\right\|_\infty
       &\leq 5 \, \exp(\mu_L L_{\max}^{UVW} - \mu_d d^2) \ .
    \end{align}
\end{lemma}
\begin{proof}
The side lengths of the boxes $UV$, $VW$, and $UVW$ coincide in two coordinate directions. We denote the remaining direction by $\dir\in \{x,y,z\}$. If $UVW$ has periodic boundary conditions in the $\dir$-direction, then $V$ is the disjoint union $V_1\sqcup V_2$ of two boxes, with $U,V_1,W,V_2$ arranged in this cyclic order along that direction, see Fig.~\ref{fig:uvw}. Since $d\geq3$, no check is connected to both $V_1$ and $V_2$. It follows that $\Gamma_V=\Gamma_{V_1}\oplus\Gamma_{V_2}$ and hence $\Sigma_V=\Sigma_{V_1}\Sigma_{V_2}$.

We denote the smallest and largest side lengths among the boxes $UV$, $VW$, and $UVW$ by $L_{\min}^{\mathrm{tot}}$ and $L_{\max}^{\mathrm{tot}}$, respectively. Since $UV,VW\subset UVW$, we have
$L_{\max}^{\mathrm{tot}}=L_{\max}^{UVW}$.
Since the boxes are fat and share two side lengths,
we have $100 L_{\min}^\mathrm{tot} \geq L_{\max}^{\mathrm{tot}}$ and $d\leq L_{\max}^\mathrm{tot}$. Note that $L_a^V\geq d-1$ (or $L_a^{V_1}, L_a^{V_2}\geq d-1$ in the periodic case).

Our strategy is to pair the boxes so that certain contributions
cancel by Lemma~\ref{lem:switching-light}. To control the remaining
terms, we first establish bounds for every $B\in\{V,UV,VW,UVW\}$ in the case of open boundary conditions, and for every $B\in\{V_1,V_2,UV,VW,UVW\}$ in the periodic case.

Applying Lemma~\ref{lem:smallsumbound} with $\ell=d-3 \leq L_a^{B}-2$ yields for any $B$ without periodic boundary conditions in $a$-direction
\begin{align}
    \|\Sigma_{B}^{a > d-3}\|_\infty 
    &\leq  \exp(\nu_{\max} L_{\max}^B- {\nu_{\min} \min\{(d-3)^2, (L_{\min}^B)^2\}}) \notag \\
    &\leq   \exp(\nu_{\max} L_{\max}^{\mathrm{tot}}- \frac{\nu_{\min}}{10000} d^2)\ , \label{eq:sigmaboundone}
\end{align}
where we used that $L_{\min}^B\geq d/100$, $d-3\geq d/100$ and $L_{\max}^B\leq L_{\max}^{\mathrm{tot}}$ to obtain the last inequality. Here $\nu_{\max}=6 \ln 2$, as in Lemma~\ref{lem:smallsumbound}.
In both the open and periodic cases, for every $B\in\{V,UV,VW,UVW\}$, and also for $B=V_1,V_2$ in the periodic case, we have
\begin{align}
\|\Sigma_B\|_\infty
    &\leq 2^{\dim\Gamma_B} 
    \leq \exp(2\nu_{\max}L_{\max}^{\mathrm{tot}}) \ . \label{eq:sigmaboundtwo}
\end{align}
The second inequality follows from
Lemma~\ref{lem:dimofkernel}, applied to each box. The additional factor of two accounts for $B=V_1\sqcup V_2$ in the periodic case, where $\Gamma_B=\Gamma_{V_1}\oplus\Gamma_{V_2}$.

Under periodic boundary conditions in the $a$-direction,
fix a nonzero $S\in\Gamma_{UVW}$. Choose the origin in that direction so that one check in $S$ has $a$-coordinate $L_a^{UVW}-1$.
Then $S\in\Gamma_{UVW}^{a>L_a^{UVW}-2}$.
Lemma~\ref{lem:checksize} yields
    \begin{align}
    |S|\geq \frac{1}{100}\min\{(L_a^{UVW}-2)^2, (L_{\min}^{UVW})^2\}    \ .
    \end{align}
The triangle inequality and Lemma~\ref{lem:dimofkernel} yield
\begin{align}
        \|\Sigma_{UVW}-I\|_\infty \leq 2^{\dim \Gamma_{UVW}} \exp(-\frac{\nu_{\min}}{10000} d^2) \leq \exp(\nu_{\max} L_{\max}^{\mathrm{tot}}- \frac{\nu_{\min}}{10000} d^2)\ , \label{eq:sigmaboundthree}
\end{align}
where we used that $L_a^{UVW}-2, L_{\min}^{UVW}\geq L_{\max}^{\mathrm{tot}}/100\geq d/100$ for $L_0$ large enough.
We now pair $UV$ with $V_2$, $VW$ with $V_1$, and compare $\Sigma_{UVW}$ with $I$, obtaining
\begin{align}
    \sumtauP{UV}\sumtauP{VW} - \sumtauP{UVW}\sumtauP{V} &= (\Sigma_{UV}-\Sigma_{V_2})\Sigma_{VW}+\Sigma_{V_2}(\Sigma_{VW}-\Sigma_{V_1})-(\Sigma_{UVW}-I)\Sigma_{V_1}\Sigma_{V_2} \ .
\end{align}
Applying Lemma~\ref{lem:switching-light} yields
\begin{align}
   (\Sigma_{UV}-\Sigma_{V_2})\Sigma_{VW}+\Sigma_{V_2}(\Sigma_{VW}-\Sigma_{V_1}) 
   = (\Sigma_{UV}^{a > \ell}-\Sigma_{V_2}^{a > \ell})\Sigma_{VW}+\Sigma_{V_2}(\Sigma_{VW}^{a > \ell}-\Sigma_{V_1}^{a > \ell}) \ .\notag
\end{align}

For open boundary conditions, we pair $UV$ with $UVW$ and $V$ with $VW$ and use Lemma~\ref{lem:switching-light} and find
\begin{align}
    \sumtauP{UV}\sumtauP{VW} - \sumtauP{UVW}\sumtauP{V} &= 
    (\sumtauP{UV} - \sumtauP{UVW})\sumtauP{VW} +\sumtauP{UVW}(\sumtauP{VW}-\sumtauP{V}) \notag \\
    &= 
    \left(\sumtauP{UV}^{a > \ell} - \sumtauP{UVW}^{a > \ell}\right)\sumtauP{VW} +\sumtauP{UVW}\left(\sumtauP{VW}^{a > \ell}-\sumtauP{V}^{a > \ell}\right) .
\end{align}

In each of the cases, using the triangle inequality, the submultiplicativity of the norm and the bounds in~\eqref{eq:sigmaboundone}, \eqref{eq:sigmaboundtwo} and~\eqref{eq:sigmaboundthree}, we arrive at
\begin{equation}
    \left\| \sumtauP{UV}\sumtauP{VW} - \sumtauP{UVW}\sumtauP{V}\right\|_\infty
    \leq 5  \exp(3\nu_{\max} L_{\max}^{\mathrm{tot}}- \frac{\nu_{\min}}{10000} d^2) \ .
\end{equation}
Setting $\mu_L=3\nu_{\max}$ and $\mu_d=\nu_{\min}/10000$ completes the proof.
\end{proof}

\subsection{Combining the bounds}

\noindent
We are finally ready to establish the DS-condition~\eqref{eq:DS-condition} using the expression~\eqref{eq:fraction}. This will complete the proof of Theorem~\ref{thm:rm}.
\begin{theorem}
    There exist constants $K,\gamma,L_0>0$, depending only on $\beta$, such that the following holds in the DS-geometry described in Subsection~\ref{s:DSgeo} with  $d =\dist{U,W}$.
    If $d\geq (L_{\max}^{UVW})^{2/3}$ and $L_{\min}^{UVW}>L_0$, then
\begin{equation}
    \left\|\frac{\Sigma_{UV}\Sigma_{VW}- \Sigma_{UVW}\Sigma_{V}}{\Sigma_{UVW}\Sigma_{V}}\right\|_\infty \leq K e^{-\gamma d^2} \ . \label{eq:thmcombinedbound}
\end{equation}
\end{theorem}
\begin{proof}
We write the operator inside the norm
in~\eqref{eq:thmcombinedbound} as $AB^{-1}$, where $ A:=\Sigma_{UV}\Sigma_{VW}- \Sigma_{UVW}\Sigma_{V} $ and $ B:=\Sigma_{UVW}\Sigma_{V} $. 
Since $B$ is positive definite, Lemma~\ref{lem:denominator} and~\ref{lem:numerator}
yield
    \begin{align}
    \|AB^{-1}\|_\infty
    &\leq \|A\|_\infty\,\lambda_{\min}(B)^{-1}\notag \\
    &\leq 5 \cdot e^{\mu_L L_{\max}^{UVW} - \mu_d d^2 }  (1-\tau)^{-18 L_{\max}^{UVW}} \notag \\
    &=5\exp\left(
        (\mu_L+\mu_L')L_{\max}^{UVW}-\mu_d d^2\right) \ ,
    \end{align}
where we set $\mu_L':=18 |\ln(1-\tau)|$.

Note that the assumptions imply $L_{\max}^{UVW}\leq d^{3/2}$ and $d\geq L_0^{2/3}$.  Increasing $L_0$ if necessary, we have
$(\mu_L+\mu_L')d^{3/2}\leq\gamma d^2$ for every $d\geq L_0^{2/3}$ where we abbreviate $\gamma = \mu_d/2$. Consequently,
\begin{align}
    (\mu_L+\mu_L')L_{\max}^{UVW}-\mu_d d^2
    &\leq (\mu_L+\mu_L')d^{3/2}-\mu_d d^2 \leq -\gamma d^2 \ .
\end{align}
Thus $\|AB^{-1}\|_\infty\leq 5 e^{-\gamma d^2}$. This proves~\eqref{eq:thmcombinedbound} with $K=5$.
\end{proof}

\appendix
\section{Multiscale analysis}\label{sec:multiscale}
The multiscale analysis propagates a lower bound on the local MLSI constant from a fixed scale to all larger scales, with a loss bounded independently of the system size. The argument in \cite[Section~6.2]{stengele_ModifiedlogarithmicSobolev_2025} requires the DS-condition for boxes of arbitrary aspect ratio satisfying the distance bound $\mathrm{dist}(U,W)\geq \sqrt{L_{\max}^{UVW}}$. Here we show that, for any fixed $0<\delta<1$, it suffices to assume the DS-condition for sufficiently large fat boxes satisfying the distance bound $\mathrm{dist}(U,W)\geq (L_{\max}^{UVW})^\delta$.\\

We only work with the $X$-part ($ \sharp = X $), since the $ Z$-part is identical. Recall the main ingredients, for a box $B\subset\Lambda$ and a full-rank state $\sigma$ on $\Lambda$, we abbreviate
\begin{equation}
    D_B:=D\bigl(\sigma\|\mathbb{E}_B(\sigma)\bigr)\ ,
    \qquad EP_B:=  -\mathrm{Tr} (\mathcal{L}_B(\sigma)(\ln(\sigma)- \ln(\mathbb{E}_B(\sigma))) ) \ ,
\end{equation}
cf.~\eqref{def:CERE}.
In this notation, the modified logarithmic Sobolev constant $\alpha(B)$ is the largest
constant such that    $ 2\alpha(B)D_B\leq EP_B $.

The multiscale argument in \cite[Section~6.2]{stengele_ModifiedlogarithmicSobolev_2025} (which is based on \cite{stroock_LogarithmicSobolevInequality_199,martinelli_ApproachEquilibriumGlauber_1994b}) propagates a bound on local MLSI constants to larger scales. Here we present the modified version, which is based on \cite{martinelli_ApproachEquilibriumGlauber_1994a}.

Recall that  $L_\Lambda$ is the side length of $\Lambda$. We call a box \emph{regular} if each of
its coordinate intervals is either a full cycle of $\Lambda$ or a proper interval of
length at most $2 L_\Lambda/3$. For a length scale $L$, we denote
\[
  \alpha(L) := \inf\{\alpha(B) \mid B \subset \Lambda \text{ is a regular fat box},\
  L^B_{\max} \le L\} \ .
\]

The multiscale analysis controls both the entropy production and the relative entropy.
The former satisfies positivity and additivity on disjoint regions \cite[Appendix~A]{stengele_ModifiedlogarithmicSobolev_2025}:
\[
    EP_A+EP_B\leq EP_R
    \qquad\text{if }A,B\subset R,\quad A\cap B=\emptyset \ ,
\]
as well as monotonicity, $EP_A\leq EP_R$ for $A\subset R$.

The relative entropy is controlled by \cite[Theorem 5.1]{stengele_ModifiedlogarithmicSobolev_2025}, which states that for boxes $UV,VW$, the DS-condition implies approximate tensorization:
\[
    D_{UVW}\leq
    \bigl(1+K' e^{-\gamma \dist{U,W}}\bigr)\bigl(D_{UV}+D_{VW}\bigr) \ .
\]

\begin{figure}
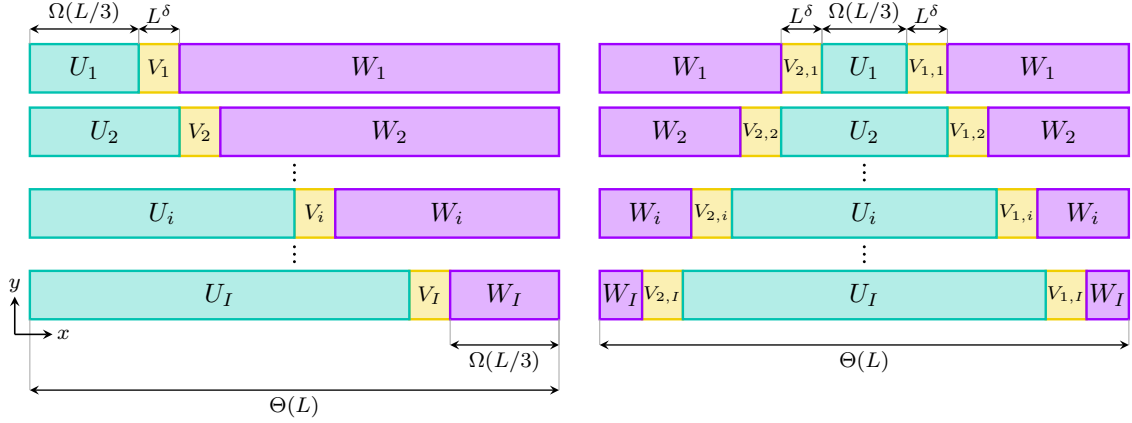

    \centering
    \uvwshift[stacked, axes=x/y, v length={L^{\delta}},total length={\Theta(L)}]
    \quad
    \uvwshiftper[wrap=false, stacked, v length={L^{\delta}},total length={\Theta(L)}]
    \caption{The geometry for the multiscale analysis. Here we used $I=\lfloor L^{(1-\delta)/2}\rfloor$ as in Eq.~\eqref{eq:I}. Left: open boundary conditions in the $x$-direction. Right: periodic boundary conditions in the $x$-direction.}
    \label{fig:multiscale}
\end{figure}

We then show the following stronger multiscale analysis:

\begin{proposition}\label{lem:mlsirecursion}
Assume that there are $K,\xi>0$ and $0<\delta<1$ such that
\[
    D_{UVW}\leq
    \bigl(1+K e^{-\xi \dist{U,W}}\bigr)\bigl(D_{UV}+D_{VW}\bigr),
\]
for every full-rank $\sigma$, whenever $U,V,W$ are disjoint,
$UV,VW,UVW$ are fat regular boxes with sufficiently
large minimum side length, and
$\dist{U,W}\geq\max\{5,(L_{\max}^{UVW})^\delta\}$.
Then there exist $L_0$ and $C>0$, independent of the system size, such that
\[
    \alpha(L)\geq C\alpha(L_0)
    \qquad\text{for all }L\geq L_0.
\]
\end{proposition}
In the case of Haah's code, we show the DS-condition with $\delta = 2/3$. 
A bound on $\alpha(L_0)$ is given in \cite{stengele_ModifiedlogarithmicSobolev_2025}.

\begin{proof} 
Fix a regular fat box $B$ with $L<L_{\max}^B\leq 4L/3$. Along a longest side, we
choose
\begin{align}
    I:=\left\lfloor L^{(1-\delta)/2}\right\rfloor \ ,\label{eq:I}
\end{align}
partitions $B=U_i\sqcup V_i\sqcup W_i$ with overlap strips $V_i$ of width
$\lceil(4L/3)^\delta\rceil$, as illustrated in Fig.~\ref{fig:multiscale}.
If $B$ reaches around the torus along the chosen side, $V_i$ has two connected components, $V_{1,i}$ and $V_{2,i}$, each of width $\lceil(4L/3)^\delta\rceil$.

We place the strips near the midpoint and shift each successive strip by its width, so that $U_iV_i$ and
$V_{i+1}W_{i+1}$ are disjoint. In the case of periodic boundary conditions, we shift both parts of $V$, growing $U$ and shrinking $W$ as $i$ increases, see Fig.~\ref{fig:multiscale}. In both cases, the total displacement is
$O(L^{(1+\delta)/2})$.
Consequently, for sufficiently large $L$, both $U_iV_i$ and $V_iW_i$ have length between
$L_{\max}^B/3$ and $2L_{\max}^B/3$ in the cut direction, for all $i$.

The resulting boxes remain fat and regular. If the cut side is a full cycle, then $L^B_{\max}=L_\Lambda$, and each piece has length at most $2L_\Lambda/3$ in the cut direction. Otherwise, each piece has length strictly less than that of the cut side. All other coordinate intervals remain unchanged.
Their length along the cut direction is at most $2L_{\max}^B/3\leq 8L/9<L$.
Since $U_i$ and
$W_i$ have the same cross-section, $\mathrm{dist}(U_i,W_i)$ is the shorter of the two
separations between them along the cut direction. If the cut side is a full cycle, both separations cross a component of $V_i$. Otherwise, one crosses $V_i$, and the other
leaves $B$ and has length at least $L_\Lambda - L^B_{\max} + 1$.
For $L_0$ large enough, this length is larger than the width of $V_i$, $\lceil (4L/3)^\delta\rceil$ .
 In all cases,
$\mathrm{dist}(U_i,W_i) \ge (4L/3)^\delta \ge (L^B_{\max})^\delta$.

Pairing consecutive disjoint pieces and bounding the two remaining
entropy productions by $EP_B$ gives
\[
    \frac1I\sum_{i=1}^{I}
    \bigl(EP_{U_iV_i}+EP_{V_iW_i}\bigr)
    \leq\left(1+\frac1I\right)EP_B.
\]
Combining this with approximate tensorization and the local MLSI yields
\[
    \alpha(B)\geq
    \frac{\displaystyle
    \min_{1\leq i\leq I}\{\alpha(U_iV_i),\alpha(V_iW_i)\}}
    {\bigl(1+K e^{-\xi(4L/3)^\delta}\bigr)
     \bigl(1+I^{-1}\bigr)}.
\]
Repeat this step on a piece attaining the minimum whenever its longest
side exceeds $L$. Each cut reduces that side below $L$ and preserves
fatness, so at most three cuts are needed. Boxes whose sides are already
at most $L$ require no cut. Taking the infimum therefore gives
\begin{equation}\label{eq:mlsirecursion}
    \alpha(4L/3)\geq
    \left[
      \bigl(1+K e^{-\xi(4L/3)^\delta}\bigr)
      \left(1+\frac{1}{\lfloor L^{(1-\delta)/2}\rfloor}\right)
    \right]^{-3}\alpha(L).
\end{equation}
Choose $L_0$ large enough for all preceding estimates and iterate along
$L_k=(4/3)^kL_0$. The error is bounded by
\begin{equation}
    \prod_{k=0}^{\infty}\left[
      \bigl(1+K e^{-\xi(4L_k/3)^\delta}\bigr)
      \left(1+\frac{1}{\lfloor L_k^{(1-\delta)/2}\rfloor}\right)
    \right]^{-3} 
      \geq \exp\left( -3 K\sum_{k\geq0}e^{-\xi(4L_k/3)^\delta} -3 \sum_{k\geq0}\frac{1}{\lfloor L_k^{(1-\delta)/2}\rfloor}\right) ,
\end{equation}
where we used $\ln(1+x)\leq x$ for $x\geq 0$.

Both
\[
    \sum_{k\geq0}e^{-\xi(4L_k/3)^\delta}
    \quad\text{and}\quad
    \sum_{k\geq0}\frac{1}{\lfloor L_k^{(1-\delta)/2}\rfloor}
\]
are finite. Hence the product of the factors in
\eqref{eq:mlsirecursion} has a positive lower bound $C$, independent of
the system size. Monotonicity of $\alpha(L)$ extends the bound from
$L_k$ to every $L\geq L_0$.
\end{proof}

\section*{Acknowledgments}
 The DFG supported this work under the grants TRR 352–Project-ID 470903074 (A.\,C., S.\,S., S.\,W.) and EXC-2111-390814868 (S.W.). L.C. acknowledges funding from the European Research Council under Grant Agreement No.~101001976 (project EQUIPTNT) and from the Munich Quantum Valley, which is supported by the Bavarian state government with funds from the Hightech Agenda Bayern Plus. This project was funded within the QuantERA II Programme, which received funding from the European Union's Horizon 2020 research and innovation programme under Grant Agreement No.~101017733.

\section*{AI Statement}
The main conceptual contributions were developed by the authors. We acknowledge the use of LLMs for a proof of Lemma~\ref{lem:characterbound} as well as testing, refining proof ideas, proofreading, and generating Ti$k$Z figures. 
 The authors take full responsibility for the final manuscript in its entirety.

\printbibliography
\end{document}